\documentclass[11pt,twoside]{report}
\usepackage[a4paper,top=3cm, bottom=2.8cm, left=3.5cm, right=2.5cm]{geometry}
\usepackage{fancyhdr}
\fancypagestyle{plain}{
\fancyhf{}
\renewcommand{\headrulewidth}{0pt} 
\renewcommand{\footrulewidth}{0pt}}
\usepackage{rotating}
\usepackage{amsmath}
\usepackage{amsfonts}
\usepackage{amssymb}
\usepackage{amsthm}
\usepackage{color}
\usepackage{float}
\usepackage{makeidx} 
\newtheorem{definition}{Definition}[section]
\newtheorem{theorem}[definition]{Theorem}
\newtheorem{proposition}[definition]{Proposition}

\newtheorem{example}[definition]{Example}
\newtheorem{remark}[definition]{Remark}

\newcommand{\N}{\mathbb{N}}

\newcommand{\Ks}{\mathbb{K}_s}
\newcommand{\Kt}{\mathbb{K}_t}
\newcommand{\Ktt}{\mathbb{K}_{tt}}
\newcommand{\Ktmp}{\mathbb{K}_{tmp}}
\newcommand{\agrave}{``}

\DeclareMathOperator{\alw}{alw}
\DeclareMathOperator{\nev}{nev}
\DeclareMathOperator{\ev}{ev}

\usepackage[header=none, border=none, number=page, cols=2, toc=true]{glossary} 
\makeindex
\makeglossary

\begin{document}

\begin{titlepage}
\setlength{\textwidth}{145mm} \setlength{\oddsidemargin}{12mm} \setlength{\marginparwidth}{15mm}
\addtolength{\headheight}{0.6cm}
  \begin{center}
  \rule{\linewidth}{2pt}\\[0.6cm]
    {\Huge Attribute Exploration\\[0.6cm]
    of Gene Regulatory Processes}\\[1.2cm]
    {\Large Dissertation}\\[0.6cm]
    zur Erlangung des akademischen Grades\\
    doctor rerum naturalium (Dr. rer. nat.)\\[1.2cm]
    {\Large vorgelegt dem Rat der Fakult\"at f\"ur Mathematik und Informatik}\\[0.6cm]
    {\Large der Friedrich-Schiller-Universit\"at Jena}\\ [1.2cm]
    \rule{\linewidth}{2pt}\\[3.6cm]
  \end{center}
  \begin{tabular}{l l l}

  Eingereicht von                 &: & Dipl. Theol. / Dipl. Math. Johannes Wollbold\\[0.3cm]
  Geboren am                     &: & 28.10.1958 in Saarbr\"ucken\\[0.3 cm]
  \end{tabular}
\end{titlepage}

\begin{titlepage}
$\qquad$ \\[18cm]
\textbf{Gutachter}                      
\begin{enumerate}
\item PD Dr. Peter Dittrich (Universit\"at Jena)
\item Prof. Dr. Bernhard Ganter (Technische Universit\"at Dresden)
\item PD Dr. Reinhard Guthke (Universit\"at Jena)
\end{enumerate}
\vspace{0.3cm}
\textbf{Tag der \"offentlichen Verteidigung:} 15. Juli 2011
\end{titlepage}

\setcounter{page}{1}
\thispagestyle{empty}
\tableofcontents

\renewcommand{\glossaryname}{\numberline{}Symbols and abbreviations}
\fancyhf{}
\fancyhead[LE,RO]{\fancyplain{}{\thepage}}
\printglossary
\fancyhf{}
\fancyhead[LE,RO]{\fancyplain{}{\thepage}}
\newpage

\fancyhf{}
\fancyhead[LE,RO]{\fancyplain{}{\thepage}}
\section*{\centerline{Abstract}}

The present thesis aims at the logical analysis of discrete processes, in particular of such generated by gene regulatory networks. States, transitions and operators from temporal logics are expressed in the language of Formal Concept Analysis (FCA). This mathematical discipline is a branch of the theory of ordered sets. It has practical applications in various fields including data and text mining, knowledge management, semantic web, software engineering, economics or biology. By the attribute exploration algorithm, an expert or a computer program is enabled to validate a minimal and complete set of implications, e.g. by comparison of predictions derived from literature with observed data. Within gene regulatory networks, the rules of this knowledge base represent temporal dependencies, e.g. coexpression of genes, reachability of states, invariants or possible causal relationships.

This new approach is embedded into the theory of universal coalgebras, particularly automata, Kripke structures and Labelled Transition Systems. A comparison with the temporal expressivity of Description Logics (DL) is made, since there are applications of attribute exploration to the construction of DL knowledge bases. The main theoretical results concern the integration of background knowledge into the successive exploration of the defined data structures (formal contexts).

In the practical part of this work, a Boolean network from literature modelling the initiation of sporulation in \textit{Bacillus subtilis} is examined. Coregulation and mutual exclusion of genes were checked systematically, also dependent from specific initial states. Conditions for sporulation were clarified by queries to the knowledge base generated by attribute exploration.

Finally, by interdisciplinary collaboration, we extracted literature information to develop an interaction network containing 18 genes important for extracellular matrix formation and destruction in the context of rheumatoid arthritis. Subsequently, we constructed an asynchronous Boolean network with biologically plausible time intervals for mRNA and protein production, secretion and inactivation. Experimental gene expression data was obtained from synovial fibroblast cells stimulated by transforming growth factor beta I (TGF$\beta$1) or by tumor necrosis factor alpha  (TNF$\alpha$) and discretised thereafter. The Boolean functions of the initial network were improved iteratively by the comparison of the simulation runs to the observed time series and by exploitation of expert knowledge. This resulted in adapted networks for both cytokine stimulation conditions. 

The simulations were further analysed by the attribute exploration algorithm of FCA, integrating the observed time series in a fine-tuned and automated manner. The resulting temporal rules yielded new contributions to controversially discussed aspects of fibroblast biology (e.g. considerable expression of TNF and MMP9 by fibroblasts following TNF$\alpha$ stimulation) and corroborated previously known facts (e.g. co-expression of collagens and MMPs after TNF$\alpha$ stimulation), but also generated new hypotheses regarding literature knowledge (e.g. MMP1 expression in the absence of FOS). 

\addcontentsline{toc}{chapter}{\numberline{}Abstract (eng./dt.)}
\newpage
\fancyhf{}
\fancyhead[LE,RO]{\fancyplain{}{\thepage}}
\section*{\centerline{Zusammenfassung}}

Ziel dieser Doktorarbeit ist die logische Analyse diskreter Prozesse, insbesondere von genregulatorischen Netzwerken. Zust\"ande, Transitionen und Operatoren der temporalen Logik werden in der Sprache der Formalen Begriffsanalyse (FCA) ausgedr\"uckt. Diese mathematische Disziplin ist ein Teilgebiet der Ordnungstheorie. Sie findet in vielf\"altigen Bereichen praktische Anwendung, wie Data und Text mining, Wissensmanagement, Semantic Web, Softwareentwicklung, Wirtschaft oder Biologie. Mittels des Merkmalexplorations-Algorithmus kann ein Experte oder ein Computerprogramm eine minimale und vollst\"andige Menge von Implikationen validieren, beispielsweise durch Vergleich von aus der Literatur abgeleiteten Vorhersagen mit Beobachtungsdaten. Im Rahmen genregulatorischer Netzwerke dr\"ucken die Regeln dieser Wissensbasis zeitliche Abh\"angigkeiten aus, etwa Koexpression von Genen, Erreichbarkeit von Zust\"anden, Invarianten oder m\"ogliche kausale Zusammenh\"ange.

Dieser neue Ansatz wird in die Theorie der Universellen Coalgebren eingebettet, die insbesondere Automatentheorie, Kripkestrukturen und Labelled Transition Systems einschlie{\ss}t. Au{\ss}erdem wird ein Vergleich zu temporalen Aspekten von Beschreibungslogiken gezogen; Anwendungen der Merkmalexploration auf die Konstruktion beschreibungslogischer Wissensbasen stellen ein neues Forschungsgebiet dar. Die wichtigsten theoretischen Resultate der vorliegenden Arbeit betreffen die Integration von Hintergrundwissen in die sukzessive Exploration der definierten Datenstrukturen (formalen Kontexte).

Im praktischen Teil der Arbeit wird zun\"achst ein Boolesches Netzwerk aus der Literatur untersucht, das die Einleitung der Sporenbildung in \textit{Bacillus subtilis} modelliert. Koregulation und gegenseitiger Ausschluss von Genen werden systematisch untersucht, auch in Abh\"angigkeit von spezifischen Ausgangszust\"anden. Bedingungen f\"ur die Einleitung der Sporenbildung werden durch Abfragen der Wissensbasis gekl\"art, die mittels Merkmalexploration erzeugt wurde.

Schlie{\ss}lich wurde in interdisziplin\"arer Zusammenarbeit und nach umfassender Literaturrecherche ein Netzwerk entwickelt, das 18 Gene enth\"alt, die f\"ur die Bildung und den Abbau extrazellul\"arer Matrix im Kontext rheumatoider Arthritis Bedeutung haben. Daraus wurde ein asynchrones Boolesches Netzwerk konstruiert mit biologisch plausiblen Zeitintervallen f\"ur mRNA- und Proteinsynthese, Sekretion und Inaktivierung. Experimentelle Genexpressionsdaten stammten von synovialen Fibroblasten, die mit Transforming growth factor beta I (TGF$\beta$1) beziehungsweise Tumor necrosis factor alpha (TNF$\alpha$) stimuliert wurden. Die Booleschen Funktionen des anf\"anglichen Netzwerks wurden in mehreren Durchl\"aufen optimiert, indem Simulationen mit den beobachteten Zeitverl\"aufen verglichen wurden. Dabei wurde zus\"atzliches Expertenwissen auch zur Signaltransduktion eingebracht. Daraus resultierten zwei Netzwerke, die jeweils an eine der beiden Bedingungen angepasst waren (Stimulation durch die beiden Zytokine).

Die endg\"ultigen Simulationen wurden mittels Merkmalexploration untersucht, wobei die gemessenen Zeitreihen weiter und automatisch integriert wurden. Die erhaltenen Regeln bringen neue Aspekte in kontrovers diskutierte Fragen der Biologie von Fibroblasten ein (z.B. betr\"achtliche Expression von TNF und MMP9 nach Stimulation durch TNF$\alpha$). Sie best\"atigen bekannte Tatsachen wie die Koexpression von Kollagenen und Matrix-Metalloproteasen (MMPs) nach TNF$\alpha$-Stimulation, erzeugten aber auch bez\"uglich des Literaturwissens neue Hypothesen (z.B. Expression von MMP1 in Abwesenheit des Transkriptionsfaktors FOS). 


\newpage
\fancypagestyle{plain}{
\fancyhf{}
\renewcommand{\headrulewidth}{0pt} 
\renewcommand{\footrulewidth}{0pt}}
\cleardoublepage

\fancyhead[RE]{\fancyplain{}{\leftmark}}
\fancyhead[LO]{\fancyplain{}{\rightmark}}
\fancyhead[LE,RO]{\fancyplain{}{\thepage}}
\fancypagestyle{plain}{
\fancyhf{}
\renewcommand{\headrulewidth}{0pt} 
\renewcommand{\footrulewidth}{0pt}}

\chapter{Introduction}
\thispagestyle{empty}
During the early 1980s, the mathematical methodology of \textit{Formal Concept Analysis (FCA)} 
emerged within the community of set and order theorists, algebraists and discrete mathematicians.
The aim was to find a new, concrete and meaningful approach to the understanding of complete
lattices (ordered sets such that for every subset the supremum and infimum exist). The
following discovery proved fruitful: Every complete lattice\index{Lattice!complete} is representable as a hierarchy of concepts, which were conceived as sets of objects sharing a maximal set of attributes. This paved the way for using the field of lattice theory for a transparent and complete representation of very different types of knowledge. 

Originally FCA was inspired by the educationalist Hartmut von
Hentig \cite{Hen72} and his program of restructuring sciences aiming at interdisciplinary
collaboration and democratic control. The philosophical background traces back to Charles S. Peirce
(1839 - 1914), who condensed some of his main ideas to the pragmatic maxim:
\begin{quote}
\emph{Consider what effects, that might conceivably have practical bearings, we
conceive the objects of our conception to have. Then, our conception of these effects is the whole
of our conception of the object.} \cite[5.402]{Pei35}
\end{quote}
In that tradition, FCA aims at unfolding the observable, elementary properties defining the objects
subsumed by scientific concepts. If applied to temporal transitions, effects of specific combinations of state attributes can be modelled and predicted in a clear and concise manner. Thus, FCA seems to be appropriate to describe causality -- and the limits of its understanding.

At present, FCA is a well developed mathematical theory and there are practical applications in
various fields such as data and text mining, knowledge management, semantic web, software engineering or economics \cite{GSW05}. The main application of this thesis is related to molecular and systems biology. Due to the rapid accumulation of data about molecular inter-relationships, there is an increasing demand for approaches to analyse the resulting regulatory network models (for a short introduction see Section \ref{sec:geneRegProc}, an example is represented in Figure \ref{fig:ECMNetwork}). Therefore, we developed a formal representation of processes, especially biological processes. The purpose was to construct knowledge bases of rules expressing temporal dependencies within gene regulatory (or signal transduction and metabolic) networks.

As algorithm, \textit{attribute exploration}\index{Attribute exploration} was employed: For a given set of
interesting properties, it builds a sound, complete and non-redundant set of implications (logically strict rules). During this process, each implication can be approved or rejected by an expert or a computer program, e.g. by comparison of knowledge-based predictions with data. Attribute exploration provides a mathematically strict framework for the validation of rules, and the resulting implicational base presents the related domain knowledge to the expert in a compact manner. This \textit{stem base}\index{Stem base} is open to intuitive human discoveries, activating resonance effects with the whole knowledge of a scientist. Its completeness related to the explored context ensures that the validity of an arbitrary implication of interest can be decided by logical derivations from the stem base (automatic reasoning). 

Corresponding to the discrete, logical and interactive focus of FCA, we selected classical \textit{Boolean networks}\index{Boolean network} \cite{Kau69} for modelling, which are easy to interpret. They consist of sets of Boolean functions, i.e. the value of one variable (e.g. gene expressed or not) after one time step depends on the present values of a subset of the variables. It is also possible to use mathematical and logical derivations in order to decide many implications automatically. Furthermore, sets of Boolean rules are applied as knowledge bases in decision support or expert systems.

FCA was used for the analysis of gene expression data in \cite{Mot08} and \cite{Kay11}. The present study is the first approach of applying it to the dynamics of (gene) regulatory network models. With this application domain in mind, we developed a formal structure as general as possible, since discrete temporal transitions occur in a variety of domains: control of engineering processes, development of the values of variables in a computer program, change of interactions in social networks, a piece of music, etc. Within the domain of discrete and symbolic process modelling, the present work aims at providing a framework that may be useful to validate and further analyse models formulated in very general classes of languages, the most important being the $\mu$-calculus (comprising as a subset \textit{Linear Temporal Logic (LTL)} and \textit{Computation Tree Logic (CTL)}) on the syntactic,\index{Syntax} and several types of \textit{universal coalgebras} on the semantic level.\index{Semantics} Simulations and analyses by \textit{Petri nets}\index{Petri net} may be integrated as well; for an example see Chapter \ref{ch:bSubtilis} and page \pageref{Petri net}.

The thesis is structured as follows: In Chapter \ref{ch:FCA}, the basic data structures of FCA, the attribute exploration algorithm and \textit{Temporal Concept Analysis (TCA)} (applied to the analysis of gene expression data in \cite{Wol11}) are introduced. In Chapter \ref{ch:logic}, \textit{automata}, \textit{Kripke structures} and \textit{Labelled Transition Systems (LTS)} with attributes are presented as universal coalgebras, furthermore \textit{Propositional Tense Logic}, CTL and \textit{Description Logics (DL)}. The DL notion of a role may be interpreted temporally. In addition, an explicit temporal extension is presented \cite{Art07}.

The starting point of the thesis consisted in further developing an FCA language for discrete dynamic systems sketched in \cite{Rud01}. Results of this modelling will be presented in Chapter \ref{ch:FCAModelling}, and the application of attribute exploration to the defined data structures (formal contexts) in Chapter \ref{ch:useAttrEx}. They express knowledge concerning states, transitions and attributes from temporal logics. Chapter \ref{ch:bg} presents a method to derive inference rules integrating already acquired knowledge into the successive exploration of the four formal contexts. This method uses attribute exploration on a higher level. Section \ref{sec:rulesBooleAttr} is a revised part of my paper \cite{Wol08}. It investigates inference rules for Boolean attributes. 

The second major part of the present work develops a systems biology method to analyse the dynamics of gene regulatory networks. The rules of the knowledge bases generated by attribute exploration represent temporal relationships within gene regulatory networks, e.g. coexpression of genes. Reachability of states is mainly expressed by rules with the temporal operators \texttt{eventually} and \texttt{never}, invariants by \texttt{always}. We focus on the corresponding semantic level, i.e. implications pertaining to transitions (compare Remark \ref{rem:semanticsKtt} or Proposition \ref{prop:KtKs}). Rules pointing at possible causal relations have a structure like:
\begin{quote}\emph{If gene 1 is expressed and gene 2 is not expressed at some state, then at the next state /
at all following states / eventually / always gene 3 will be expressed.}
\end{quote}
In Chapters \ref{ch:bSubtilis} and \ref{ch:geneRegNets}, the background, methods and results of two published papers of own work are reported:
\begin{enumerate}
 \item Sporulation of \textit{Bacillus subtilis}: A simulation from literature was further analysed by concept lattices and attribute exploration \cite{Wol08}.
 \item  In \cite{Wol09}, we developed by interdisciplinary collaboration a Boolean network model for \textit{extracellular matrix (ECM)} formation and degradation within the context of rheumatic diseases. It is based on interactions reported in literature and was adapted to gene expression time series for fibroblast cells following transforming growth factor beta I (TGF$\beta$1) and tumor necrosis factor alpha (TNF$\alpha$) stimulation, respectively. The resulting simulations were further analysed by attribute exploration, integrating the observed time series in a fine-tuned and automated manner.
\end{enumerate}

In Chapter \ref{ch:outlook}, the results of the biological applications are discussed. Possibilities of further research aiming at an even better usability of this approach are sketched. Among other mathematical and logical questions I will outline how a state, transition or temporal context may be expressed by DL. This is the basis so that results applying to the computation and extension of DL knowledge bases by attribute exploration can be used  as well as fast DL reasoners  \cite{Baa09b} \cite{Baa09a}. I will also point at reasons for concentrating on a classical FCA framework and on parts of temporal logic being particularly meaningful to human experts in real world applications.

\section{Acknowledgements}
First, I am grateful to PD Dr. Reinhard Guthke for posing the biomedical question and for many critical hints aiming at biological meaning, realistic applicability and suggestion of hypotheses by modelling. He offered to me great freedom for developing a method to integrate knowledge and data. Since Prof. Dr. Bernhard Ganter was equally broad-minded regarding the application of FCA to biology, a large search space was  opened, ranging from pure mathematics (lattice and FCA theory, categories) over FCA applications, different logics, gene expression data analysis, systems biology up to specialised questions of molecular biology and genetics. I am very grateful to Bernhard Ganter for creative discussions which generated the main ideas of this thesis to structure the large area, to solve mathematical and logical questions and to show the applicability on two biologically relevant questions. PD Dr. Peter Dittrich accepted to resume the final supervision and made many useful remarks aiming in particular at a better comprehensibility of the FCA framework.

The collaboration with Ulrike Gausmann, Ren\'{e} Huber, Raimund Kinne and Reinhard Guthke resulting in \cite{Wol09} was very inspiring to me. At all steps of the work a feedback was given between biological knowledge, data and formal abstractions. Beyond programming and developing the methods for discretisation, simulation and attribute exploration, I participated in structuring the comprehensive literature search and constructed the Boolean network accordingly. After intense discussions, we adapted the Boolean functions to the data. I analysed the final temporal rules together with Ren\'{e} Huber.

I thank Christian Hummert (HKI Jena), Michael Hecker (HKI Jena and STZ for Proteome Analysis Rostock), Felix Steinbeck (STZ for Proteome Analysis Rostock), Mike Behrisch and Daniel Borchmann (Institute of Algebra of Dresden University) as well as other colleagues for fruitful discussions, for some programming (see Section \ref{sec:scripts}) and for reading parts of the manuscript.

I got encouragement from my supervisors as well as from many other people, to whom I expressed my gratitude personally.

\chapter[Formal Concept Analysis]{Formal Concept Analysis and the attribute exploration algorithm}\label{ch:FCA}
\index{Formal Concept Analysis (FCA)}
\glossary{name={FCA}, description={Formal concept analysis}}
\thispagestyle{empty}
This chapter provides formal definitions, a theorem as well as more intuitive introductions to known FCA notions used from Chapter \ref{ch:FCAModelling}. There, I will give examples of new formal structures and applications, which mostly should be sufficient to follow the argumentation of this thesis. Instead of reading this chapter in advance, it might therefore serve as a reference in order to clarify notions as needed. For more detailed questions, I refer to the textbook \cite{GW99}. 

\section{Formal contexts and concept lattices}
\begin{table}
\begin{center}
\begin{tabular}{|l||c|c|c|}\hline
&MMP1 &TIMP1 &MMP9\\ \hline \hline
(190,0) &&&\\ \hline
(190,1) &&&\\ \hline
(190,2) &$\times$&&\\ \hline
(190,4) &$\times$&&\\ \hline
(190,12) &$\times$&&\\ \hline
(202,0) &&&\\ \hline
(202,1) &&&\\ \hline
(202,2) &$\times$&&\\ \hline
(202,4) &$\times$&&\\ \hline
(202,12) &$\times$&&$\times$\\ \hline
(205,0) &&&\\ \hline
(205,1) &&&\\ \hline
(205,2) &$\times$&&$\times$\\ \hline
(205,4) &$\times$&&$\times$\\ \hline
(205,12) &&&\\ \hline
(220,0) &&$\times$&$\times$\\ \hline
(220,1) &&&$\times$\\ \hline
(220,2) &&&$\times$\\ \hline
(220,4) &$\times$&$\times$&$\times$\\ \hline
(220,12) &$\times$&$\times$&$\times$\\ \hline
(221,0) &&$\times$&\\ \hline
(221,1) &&$\times$&\\ \hline
(221,2) &$\times$&$\times$&\\ \hline
(221,4) &$\times$&$\times$&\\ \hline
(221,12) &$\times$&$\times$&$\times$\\ \hline
(87,0) &&$\times$&\\ \hline
(87,1) &$\times$&&\\ \hline
(87,2) &$\times$&$\times$&\\ \hline
(87,4) &$\times$&$\times$&\\ \hline
(87,12) &$\times$&$\times$&$\times$\\ \hline
\end{tabular}
\end{center}
\label{tab:3GenesCxt}
\caption{(One-valued) formal context representing a part of a gene expression data set for cells from osteoarthritis patients \textit{190, 202, 205} and rheumatoid arthritis patients \textit{220, 221, 87}. The objects (rows) designate measurements for a cell culture at \textit{t=0, 1, 2, 4} and \textit{12} hours after stimulation with TNF$\alpha$. A cross in the column for the attribute \textit{MMP1, TIMP1} or \textit{MMP9} means: The gene collagenase, tissue inhibitor of metalloproteaseses 1 or matrix metalloprotease 9 is expressed above a threshold set by the discretisation method proposed in \cite[Section 2.3]{Wol11}. For the biomedical background see Section \ref{sec:ECMBg}.}
\end{table}

One of the classical aims of FCA is the structured, compact but complete visualisation of a data
set by a conceptual hierarchy. The subsequent basic definitions of formal contexts, scaling and formal concepts are applied in Sections
\ref{sec:runEx}, \ref{sec:stateCxt} and \ref{simNoStress}. 

A data table with binary attributes is called a (one-valued) formal context (Table \ref{tab:3GenesCxt}):
\begin{definition}\label{def:context}\index{Formal context}\index{Object}\index{Attribute} \cite[Definitions 18 and 19]{GW99}
A \textsc{Formal Context} $(G,M,I)$ defines a relation $I \subseteq G \times M$ between objects
from a set $G$ and attributes from a set $M$. The set of the attributes common to all objects in $A
\subseteq G$ is denoted by the $'\,$-operator:\glossary{name={$'$}, description={Derivation operator for object or attribute sets}, sort=$.Der$}
\[A' := \{m \in M \mid (g,m) \in I \text{ for all } g \in A\}.\]
The set of the objects sharing all attributes in $B \subseteq M$ is
\[B' := \{g \in G \mid (g,m) \in I \text{ for all } m \in B\}.\]
\end{definition}
If the derivation operators $'$ are ambiguous, they will be denoted by the relation of the respective formal context, e.g. $B^I$.\glossary{name={$.^I$}, description={Derivation operator for a formal context $(G,M,I)$}, sort=$Id$} Also $gIm$\glossary{name={$gIm$}, description={$(g,m) \in I$ for a formal context $(G,M,I)$},sort=$Gim$} will be used instead of $(g,m) \in I$. 

A formal context is called \textit{clarified},\index{Formal context!clarified} if there are no objects with the same attribute set (object intent, see Definition \ref{def:concept}) and no attributes with the same extent, i.e. the context does not contain rows / columns identical except for object / attribute names. A formal context is \textit{row reduced},\index{Formal context!reduced} if all objects are deleted of which the intent is an intersection of other object intents. The definition of a \textit{column reduced} context is analogous. With the exception of the test context in Section \ref{sec:ruleEx}, the formal contexts will not be reduced, since then the information regarding a part of the objects (for instance gene expression measurements) or attributes (e.g. genes) is lost.

\begin{definition}\label{def:MVcontext}
\cite[Definition 27]{GW99}
A \textsc{Many-Valued Context} $(G,M,W,J)$ consists of sets $G$, $M$ and $W$ and a ternary relation $J \subseteq G \times M \times W$ for which it holds that 
\[(g,m,w) \in J \wedge (g,m,v) \in J \Rightarrow w=v.\]
\index{Formal context!many-valued}
\end{definition}
\glossary{name={$J$}, description={Relation of a many-valued context}}
$(g,m,w) \in J$ means ``for the object $g$, the attribute $m$ has the value $w$''. Thus, the many-valued attributes are identifiable with partial maps $m\colon G \rightarrow W,$ where $m(g) = w \Leftrightarrow (g,m,w) \in I$. A many-valued context represents a specifically formalised view on an arbitrary table in a relational database. It is translated into a \textit{derived}\index{Formal context!derived} ordinary, one-valued formal context by a process called \textit{conceptual scaling}. Scaling makes the mathematical results of FCA applicable to many-valued contexts and offers manifold possibilities of data discretisation.\index{Discretisation}
\begin{definition}\cite[Definition 28]{GW99}
A \textsc{Scale} for the attribute $m$ of a many-valued context $(G,M,W,J)$ is a (one-valued) context $\mathbb{S}_m:= (G_m,M_m,I_m)$ with $m(G) \subseteq G_m$. The objects of a scale are called \textsc{Scale Values}, the attributes \textsc{Scale Attributes}.
\end{definition}

\begin{table}
\begin{center}a) 
\begin{tabular}{|l||c|c|}\hline
&$m_1$ &$m_2$\\ \hline \hline
$g_1$ &0&1\\ \hline
$g_2$ &1&1\\ \hline
\end{tabular}$\qquad$
b)
\begin{tabular}{|l||c|c|}\hline
$\mathbb{N}_2$ &$0$ &$1$\\ \hline \hline
$0$ &$\times$&\\ \hline
$1$ &&$\times$\\ \hline
\end{tabular}$\qquad$
c)
\begin{tabular}{|l||c|c|c|c|}\hline
&$m_1.0$  &$m_1.1$ &$m_2.0$ &$m_2.1$\\ \hline \hline
$g_1$ &$\times$&&&$\times$\\ \hline
$g_2$ &&$\times$&&$\times$\\ \hline
\end{tabular}
\end{center}
\caption{a) A many-valued context with object set $G:=\{g_1,g_2\}$, attribute set $M:=\{m_1,m_2\}$ and values $W:=\{0,1\}$, b) the nominal scale with two values, i.e. the dichotomic scale\index{Scaling!dichotomic} for each of the two attributes, c) the scaled context. For instance, in Table \ref{tab:3GenesCxt} the attribute MMP1 could be replaced by MMP1.0 and MMP1.1. This provides explicit information whether a gene is not expressed and has practical relevance mainly for attribute exploration (see Table \ref{tab:subtilisTransCxt} and the implications on page \pageref{impSubtilisGrowth}).}
\label{tab:DichotomicScaling}
\end{table}
By (plain) scaling, an attribute $m$ is replaced by the respective row of the scale context $\mathbb{S}_m$.
\begin{definition}\index{Scaling}\label{def:Scaling}
\cite[Definition 29]{GW99}
If $(G,M,W,J)$ is a many-valued context and $\mathbb{S}_m,\: m \in M$ are scale contexts, then the \textsc{Derived Context With Respect To Plain Scaling} is the context $(G,N,I)$ with
\[N:= \bigcup_{m \in M} \{m\} \times M_m,\]
and 
\[gI(m,n): \Leftrightarrow m(g) = w \text{ and } wI_m n.\] 
\end{definition}

The most elementary example is \textit{nominal scaling}\index{Scaling!nominal}\label{nomScaling} \cite[Definition 31]{GW99}: The scale context for an attribute $m$ of the many-valued context is a diagonal matrix, i.e. each attribute value $w \in G_m \subseteq W$ is represented by itself. Then, $m$ is replaced by derived attributes which can be mapped bijectively to the (possible) attribute values $G_m = M_m$. Nominal scaling with two scale attributes is called \textit{dichotomic scaling} (Table \ref{tab:DichotomicScaling}). In the following, mostly a variant of a dichotomic scale is applied, where the threshold discretisation - for each gene separately - presupposed in Table \ref{tab:3GenesCxt} is made explicit, e.g.:
\begin{center}
\begin{tabular}{|l||c|c|}\hline
MMP1 &$0$ &$1$\\ \hline \hline
$\leq 5710$ &$\times$&\\ \hline
$>5710$ &&$\times$\\ \hline
\end{tabular}$\qquad$
\begin{tabular}{|l||c|c|}\hline
MMP9 &$0$ &$1$\\ \hline \hline
$\leq 144$ &$\times$&\\ \hline
$>144$ &&$\times$\\ \hline
\end{tabular}
$\qquad$
\begin{tabular}{|l||c|c|}\hline
TIMP1 &$0$ &$1$\\ \hline \hline
$\leq 34890$ &$\times$&\\ \hline
$>34890$ &&$\times$\\ \hline
\end{tabular}
\end{center}
Further scaling methods will be introduced in Section \ref{sec:stateCxt}.

In this thesis, two context constructions are needed. $\overset{.}{M}_1 \cup \overset{.}{M}_2$ denotes the disjoint union of two attribute sets $M_1$ and $M_2$, i.e. $\overset{.}{M}_j := \{j\} \times M_j$ for $j \in \{1,2\}$. Analogously, $\overset{\:.}{I}_j := \{(g,(j,m))\mid (g,m) \in I_j\}$ for two relations $I_j \subseteq G \times M_j$.
\begin{definition}
\textnormal{\cite[Definition 30]{GW99}}
The \textsc{Apposition} of two formal contexts $\mathbb{K}_1:= (G,M_1,I_1)$ and $\mathbb{K}_2 := (G, M_2, I_2)$ is defined by
\index{Formal context!apposition}\glossary{name={$\mathbb{K}_1 \mid \mathbb{K}_2$}, description={Apposition of formal contexts},sort=$K1$}
\[\mathbb{K}_1 \mid \mathbb{K}_2:= (G, \overset{.}{M}_1 \cup \overset{.}{M}_2, \overset{\:.}{I}_1 \cup \overset{\:.}{I}_2).\]
\end{definition}

\begin{definition}\label{def:semiproduct}
\textnormal{\cite[Definition 33]{GW99}}
The \textsc{Semiproduct} of two formal contexts $\mathbb{K}_1:= (G_1,M_1,I_1)$ and $\mathbb{K}_2 := (G_2, M_2, I_2)$ is defined by
\index{Semiproduct of formal contexts}\glossary{name={$\begin{sideways}$\bowtie$\end{sideways}$}, description={Semiproduct of formal contexts}, sort=$.Semiproduct$}
\[\mathbb{K}_1\, \begin{sideways}$\bowtie$\end{sideways}\: \mathbb{K}_2 := (G_1 \times G_2, \overset{.}{M}_1 \cup \overset{.}{M}_2, \nabla) \]
with
\[(g_1,g_2) \nabla (j,m) :\Leftrightarrow g_jI_jm \qquad \text{for } j \in \{1,2\}.\]
\end{definition}

\begin{figure}
 \centering
 \includegraphics[height=10.5cm]{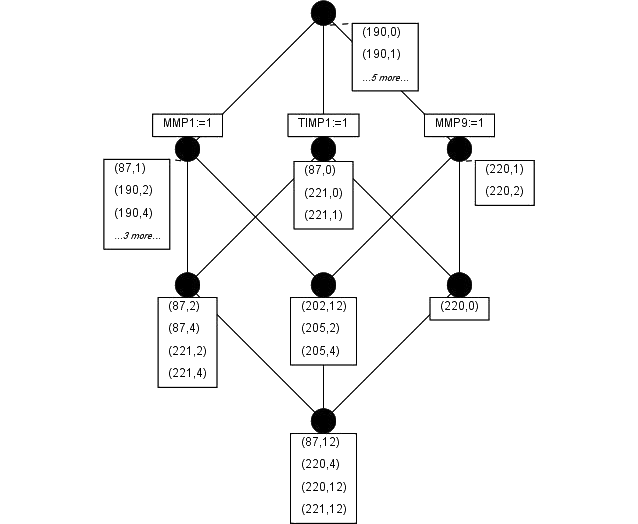}
\caption{Concept lattice (Hasse diagram)\index{Concept lattice (hierarchy)}\index{Lattice} of the formal context in Table \ref{tab:3GenesCxt}. Formal concepts are represented by black circles and the order relation by lines. The extent\index{Extent} of a formal concept is the set of all objects (measurements at specific time points) at or below the circle, following the lines. Dually, the intent\index{Intent} is the set of all attributes above. Thus, the semantics of the bottom concept is: all three genes are expressed for patients 87 at 12~h, 220 at 4 h and 12 h, and 221 at 12 h after stimulation with TNF$\alpha$.}
\label{fig:TNFMMPLattice}
\end{figure}
The central notion of FCA is a \textit{formal concept}, a maximal set of objects $A \subseteq G$ together with all attributes $B \subseteq M$  shared by them, or a maximal rectangle in a formal context, after row and column permutation.
\begin{definition}\label{def:concept}\index{Formal concept}
\cite[Definition 20]{GW99}
A \textsc{Formal Concept} of the context $(G,M,I)$ is a pair $(A,B)$ with $A \subseteq G, \:
B\subseteq M,\: A'=B$ and $B'=A$. $A$ is the \textsc{Extent}, $B$ the \textsc{Intent} of the
concept $(A,B)$.
\end{definition}
For a context $(G,M,I)$,
\begin{align*}
\mathcal{C}_M\colon \mathfrak{P}(M) &\rightarrow \mathfrak{P}(M)\\
B \subseteq M &\mapsto B''\text{, and}\\
\mathcal{C}_G\colon \mathfrak{P}(G) &\rightarrow \mathfrak{P}(G)\\
A \subseteq G &\mapsto A'' 
\end{align*}
are \textit{closure operators}\index{Closure operator}, i.e. operators with the properties monotony, extension and idempotency \cite[Definition
14]{GW99}. From \cite[Theorem 1]{GW99} it follows that the set of all extents and intents, respectively, of a formal context is a \textit{closure system},\index{Closure system, closed set} i.e. it is closed under arbitrary intersections. Hence, an intent is sometimes called a \textit{closed set} \cite{Pfa02}, \cite{Hol07}.

Formal concepts\index{Concept lattice (hierarchy)} can be ordered by set inclusion of the extents or -- dually and with the inverse order relation -- of the intents. With this order, the set of all concepts of a given formal context
is a complete lattice,\index{Lattice!complete} i.e. a partially ordered set, where supremum (\textit{join}) and infimum (\textit{meet}) exist for any subset. It is visualised by a \textit{Hasse diagram} (Figure \ref{fig:TNFMMPLattice}).\index{Hasse diagram}

\section{Attribute exploration}
\index{Implication}Compared to a concept lattice, logical implications offer an even more compact possibility of representing a row reduced formal context without loss of information. An implication $A \rightarrow B$ between attribute sets \textit{holds}\index{Implication!holds} in a formal context $(G,M,I)$, if an object $g \in G$ having all attributes $a \in A$ (premise) has also the attributes $b \in B$ (conclusion). This is expressed by \agrave every object intent \textit{respects}\index{Implication!respected} $A\rightarrow B$" or $\forall g \in G\colon g' \models A \rightarrow B$.\glossary{name={$\models$}, description={Modelling relation}, sort=$.Models$} The attribute exploration algorithm generates a special set of attribute implications, the \textit{stem base}.\index{Stem base} Their premises are pseudo-intents:
\index{Pseudo-intent}
\glossary{name={$P$}, description={Pseudo-intent}}
\begin{definition}
\cite[Definition 40]{GW99}
Let $M$ be a finite set. $P\subseteq M$ is called a \textsc{Pseudo-Intent} of $(G,M,I)$ if and only if $P\neq P''$ and
$Q''\subseteq P$ holds for every pseudo-intent $Q\subseteq P$, $Q\neq P$.
\end{definition}
This recursive definition starts with the empty set. Only the first condition has to be checked:
$\emptyset' = G$ (no distinctive attribute is required for an object), thus the closure
$\emptyset''$ contains the attributes common to all $g \in G$. Often they have not been made
explicit in the context under consideration, and the empty set is closed. In this case also for the sets with one element only the closure must be tested, and so on. The first non-closed set in a linear order of set inclusion is a pseudo-intent $Q$. Then the second condition $Q''\subset P$ has to be tested for every superset $P$. For an example of pseudo-intents (respectively pseudo-closed sets) within the context of attribute exploration see p. \pageref{ex:pseudoclosed}.

The following theorem gives the theoretical foundation of attribute exploration. 
\begin{theorem}[Duquenne-Guigues]\label{theorem:dg}\index{Theorem!of Duquenne-Guigues}
Given a formal context $(G,M,I)$, the set of implications \[\mathcal{L}:= \{P\rightarrow P''\,|\, P \text{ pseudo-intent}\}\] is sound, complete and non-redundant.
\end{theorem}
\begin{proof}
By definition of the closure operator $''$, the implications in $\mathcal{L}$ respect all object intents. Thus, they hold in the underlying context $(G,M,I)$ and $\mathcal{L}$ is sound. For the proof of completeness and non-redundancy see \cite[Theorem 8]{GW99}. 
\end{proof} 

The set $\mathcal{L}$ is called stem base of a formal context. In general, its implications are noted in the short form $P \rightarrow P'' \setminus P$, as $P \rightarrow P$ is trivial.\index{Stem base} Completeness\index{Completeness (logics)} means that every implication holding in a given formal context $(G,M,I)$ can be derived logically from  $\mathcal{L}$. This property is lost, if a single implication is removed from the stem base (non-redundancy)\index{Non-redundancy (logics)}. For complete syntactic inference, the \textit{Armstrong rules}\index{Armstrong rules} 1, 2 (\ref{eq:2ndArmstrong}) and 6 \eqref{eq:premConclusion} can be used \cite[Proposition 21]{GW99}. They are sound\index{Soundness (logics)} in the sense that every implication proven by the implications in $\mathcal{L}$ and the Armstrong rules is semantically\index{Semantics} valid, i.e. holds in $(G,M,I)$.


By reason of these strong properties (where non-redundancy is not necessary), an object reduced formal context can be reconstructed from its stem base as well as the order relation of the corresponding concept lattice. Thus, Figure \ref{fig:TNFMMPLattice} represents a Boolean lattice\index{Lattice!Boolean}, i.e.~its order relation is given by set inclusion of the elements in the power set $\mathfrak{P}(M)$,\glossary{name=$\mathfrak{P}(M)$, description={Power set of $M$},sort=$PM$} $M:=\{\text{MMP1:=1, TIMP1:=1, MMP9:=1}\}$. Hence, there is no restriction on intents and the stem base is empty. For examples of correlations between implications and a concept lattice see Section \ref{simNoStress}.


\index{Attribute exploration}
During the interactive attribute exploration algorithm \cite[p.~85ff.]{GW99}, an expert is asked about the general validity of basic implications $A \rightarrow B$ between the attributes of a given formal context $(G,M,I)$. If the expert rejects the statement, (s)he must provide a
counterexample, i.e. a new object of the context. If she accepts, the implication is added to the
stem base\index{Stem base} of the -- possibly enlarged -- context, which at the end is precisely the set $\mathcal{L}$
of Theorem \ref{theorem:dg}. In many applications, one is only interested in the set of all implications of a fixed formal context. Then, no expert is needed for a confirmation of the implications. Sometimes I refer to this algorithm as ``computing the stem base" of a formal context.


A counterexample\index{Counterexample}\label{maxCounterEx} has to be chosen carefully, since its object intent defines a new closed set. It must correspond to the explored (mathematical or other) reality, i.e. either a single object with this attribute set exists, or a class of objects has exactly these attributes in common. Otherwise a valid implication may be precluded between a pseudo-closed set and the larger, correct intent. If the counterexample intent is chosen too large, this can be corrected by new counterexamples. A counterexample contradicting already accepted implications is immediately rejected by the implementations of the algorithm.

In this work, mostly the Java implementation \texttt{Concept Explorer}\index{Concept Explorer} \cite{ConExp} was used. It handles large contexts, offers the possibility of lattice visualisation with highlighting of filters and ideals, reads -- among other formats -- tabulator separated \texttt{*.txt} or \texttt{*.csv} files, and its graphical user interface is easy to use. The DOS and Linux command line tool \texttt{ConImp} \cite{Bur03}\index{ConImp} is restricted to 255 objects and attributes, respectively, but offers enlarged possibilities like handling incomplete or background (cf. Chapter \ref{ch:bg}) knowledge. 

\section{Temporal Concept Analysis (TCA)}\index{Temporal Concept Analysis (TCA)}
\begin{figure}
 \centering
 \includegraphics[height=10.5cm]{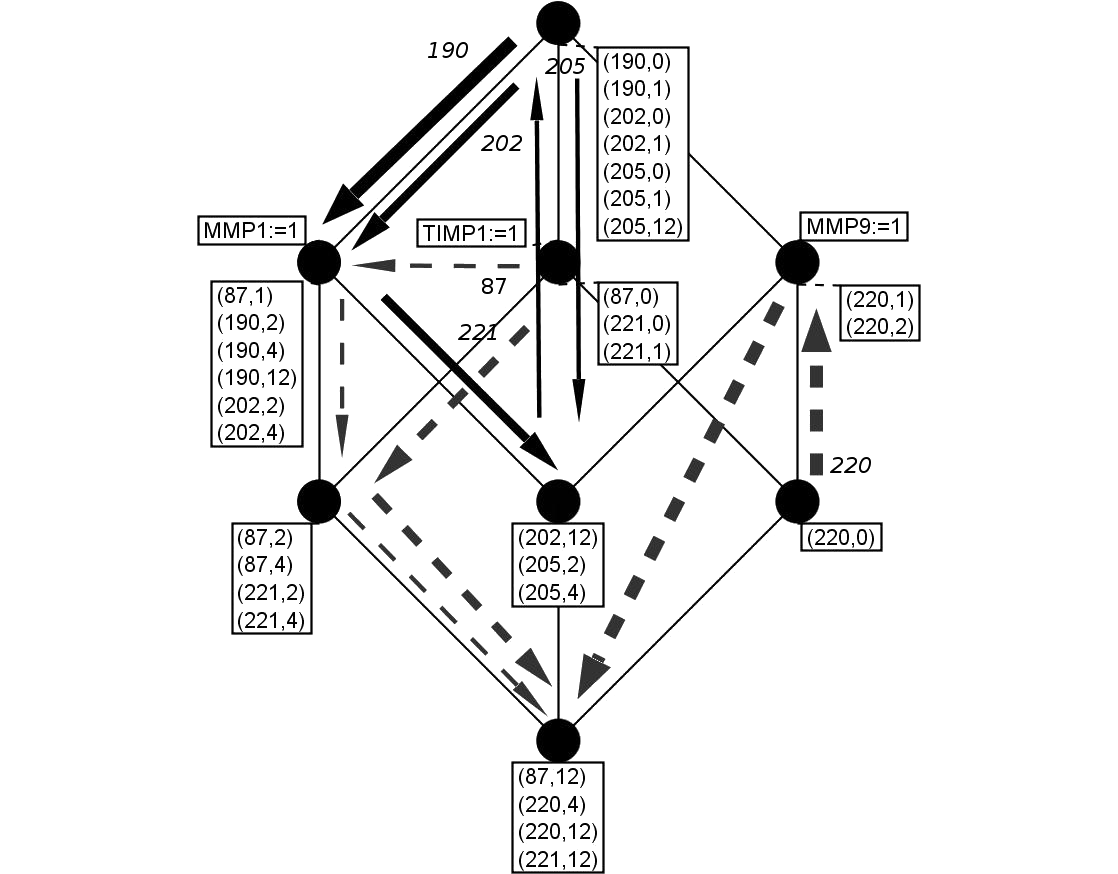}
\caption{Visualisation by \textit{life tracks} of the gene expression time series from Table \ref{tab:3GenesCxt}, for cells from osteoarthritis patients 190, 202, 205 (\textit{solid arrows}) and rheumatoid arthritis patients  220, 221, 87 (\textit{dashed arrows}), at $t=0, 1, 2, 4$ and $12$ hours after stimulation with TNF. The cubus represents the Hasse diagram $\mathfrak{B}(\mathbb{K}_{C})$ of the formal context \textbf{C}, the space part of a CTSOT; it is identical to the concept lattice in Figure \ref{fig:TNFMMPLattice}. States -- indicated by black circles -- are object concepts of $\mathfrak{B}(\mathbb{K}_{C})$, i.e. sets of measurements $(o,t), o \in O, t \in T$ of cell cultures $O$ at specific time points $T$, where a common set of genes is expressed and the remaining genes are not. For data analyses based on similar TCA diagrams see \cite{Wol11}.}
\label{fig:TNFMMP}
\end{figure}
\glossary{name={TCA}, description={Temporal Concept Analysis}}
In this section only a short impression of TCA is given, in order to compare it with our independent approach with partly different purposes. Therefore, this section may be skipped, or an intuitive understanding of Figure \ref{fig:TNFMMP} may be sufficient.

K.E.~Wolff developed an FCA based \agrave temporal conceptual granularity theory for movements of general objects in abstract or 'real' space and time'' \cite[p.~127]{Wol05}. \textit{Temporal Concept Analysis} (TCA) is based on many-valued contexts representing, e.g., several observed time series: row entries are \textit{actual objects}, i.e. pairs $(o,t) \in \Pi \subseteq O \times T$,\glossary{name={$O$}, description={Set of objects in TCA}}\glossary{name={$T$}, description={Set of time granules in TCA}} where $o$ is called an \textit{object} (interpreted for example as a manufacturing machine, a metereological station or a person), and $t$ a \textit{time granule} (interpreted for example as a time point or a time interval). 

\glossary{name={CTSOT}, description={Conceptual Time System with Actual Objects and a Time Relation}}
More exactly, a \textit{Conceptual Time System with Actual Objects and a Time Relation  (CTSOT)} is defined as a pair $(\mathbf{T}, \mathbf{C})$ of two many-valued contexts on the same set $\Pi$ of actual objects, together with a relation $R \subseteq \Pi \times \Pi$. The attributes discriminate between the \textit{time part} $\mathbf{T}$ and the \textit{event part} or \textit{space part} $\mathbf{C}$. $\mathbb{K}_{TC} := \mathbb{K}_T \mid \mathbb{K}_C$ is the apposition \index{Formal context!apposition} of the respective derived, one-valued contexts. This leads to the definition of a \textit{state}\index{State} as an object concept of $\mathbb{K}_C$, and of a \textit{situation} as an object concept of $\mathbb{K}_{TC}$.

The principal aim of TCA is the visualisation of temporal data within the corresponding concept lattices $\mathfrak{B}(\mathbb{K}_{C})$ and $\mathfrak{B}(\mathbb{K}_{TC})$. The object concept mapping \mbox{$\gamma\colon \Pi \rightarrow \mathfrak{B}(\mathbb{K}_{TC})$} yields the directed graph of \textit{life tracks}, connecting the object concepts $\gamma(o,t)$ of a single object $o$ according to $R$.
The \textit{Life Track Lemma} \cite[p.~139]{Wol05} gives attention to a mapping $\gamma_C\colon \Pi \rightarrow \mathfrak{B}(\mathbb{K}_{C})$ onto the state instead of the situation lattice (compare Figure \ref{fig:TNFMMP}), as well as to mappings onto lattices according to any other restriction of the attribute set (\textit{view}). It specifies how life tracks in $\mathfrak{B}(\mathbb{K}_{TC})$ may be mapped onto life tracks in the sublattices, which are meet-preservingly embedded into $\mathfrak{B}(\mathbb{K}_{TC})$.

\textit{Conceptual Semantic Systems (CSS)} \cite{Wol05c} include spatially distributed objects covering yet quantum incertainty, but they are much too general for our approach.

In \cite{Wol11}, TCA is applied to the graphical analysis of gene expression data, while the present work principally aims at temporal logic. Life tracks are a structure supplementary to the underlying formal context and concept lattice. In contrast, a single context representing the time relation $R$ is required in order to apply attribute exploration. We also wanted to start from the most general FCA framework to take advantage of the broad range of mathematical results and of existing software.
Therefore, we constructed a parallel modelling approach that is based on automata theory (and similar approaches) just like TCA. The mutual relation will be specified in Section \ref{sec:TCALTSA}.

\chapter{Algebraic and logic process modelling}\label{ch:logic}
The present work aims at providing a framework that may be useful to validate and further analyse process models formulated in very general classes of languages. On the semantic level\index{Semantics}, this chapter gives an overview on \textit{automata}, \textit{Labelled Transition Systems with Attributes} (LTSA, an extension of semiautomata) and \textit{Kripke structures}. In order to reveal connnections, they are presented as different types of \textit{universal coalgebras}. Focusing on the syntactic level,\index{Syntax} I concentrate on A. S. Priors logic of time and Computation Tree Logic (CTL). Since there is important research concerning connections of FCA and Description Logics (DL), and DL relations may be considered as temporal, the basic definitons and an explicitly temporal extension are presented in Section \ref{sec:DL}. Section \ref{sec:compDL} discusses how the defined formal contexts may be translated to a DL. In Section \ref{sec:sysBiol}, further modelling languages used in systems biology are mentioned, in particular Boolean networks.

The present approach is based on LTSA, more exactly on Kripke structures, since different actions are not distinguished. For simulations, Boolean networks are used, and the temporal logic CTL for dynamic assertions.

Within this chapter, it is not possible to give a self-contained introduction to the broad range of theories. It aims more at drawing connections which might be interesting for readers familiar with a theory, thus at anchoring my own approach defined without special presuppositions in Chapter \ref{ch:FCAModelling}. There, references to the present chapter might be overread. To understand the immediate background of Chapter \ref{ch:FCAModelling}, it should be sufficient to read the introduction to Boolean networks in Section \ref{sec:BooleNets}, the definitions of an LTSA (\ref{def:LTSA}), of a Kripke structure (\ref{def:Kripke}) and of the CTL operators (Section \ref{sec:CTL}), possibly also the paragraphs concerning their origin in propositional tense logic (Section \ref{sec:Prior}). The main part of this section discusses philosophical ideas of A.S. Prior regarding the flux of events, their fixation in a data frame, open future, freedom and limits of temporal knowledge. It is an excursus fitting well to my view of FCA as a method aiming to support human understanding and responsible discussion of the reach of data analyses. 

\section{An unifying approach: Universal coalgebras}\label{sec:coalg}\index{Coalgebra}
In computer science the mathematical discipline of Universal Coalgebra achieved large success as a common theory of state\index{State} based systems, generalising, e.g., automata and Kripke structures. The observable output of such dynamic systems depends on an input as well as on an internal state, where the input may change the state. As will be shown in important special cases, this can be described by a set $S$ of states and a mapping from $S$ to a combination of states and outputs. 

An (universal) coalgebra is defined in the language of category theory, which aims to describe structural similarities between mathematical theories. A \textit{category}\index{Category} $\mathfrak{C}$ consists of a class of \textit{objects} (e.g. the class of sets, groups or vector spaces) and a class of \textit{morphisms}\index{Morphism} (e.g. homomorphisms): For two objects $A,B \in \mathfrak{C}$, a set $\operatorname{Mor}(A,B)$ is defined, and the axioms are satisfied (the very natural first axiom is omitted) \cite[p. 53]{Lan02}:
\begin{itemize}
\item[] \textbf{CAT 2}: For each object $A \in \mathfrak{C}$ there is a morphism $\operatorname{id}_A \in \operatorname{Mor}(A,A)$, the identical map on $A$.
\item[] \textbf{CAT 3}: The class of morphisms is closed against composition. The law of composition is associative: If $A,B,C,D \in \mathfrak{C}$ and $f \in \operatorname{Mor}(A,B),\: g \in \operatorname{Mor}(B,C)$ and $h \in \operatorname{Mor}(C,D)$, then
\[(h \circ g) \circ f = h \circ (g \circ f).\]
\end{itemize}

A \textit{functor}\index{Functor} defines exactly how objects and morphisms of one category can be transferred to another category. We only need covariant functors respecting the direction of morphisms:
\begin{definition}\cite[p. 62]{Lan02}
A \textsc{covariant functor} $\Omega$ of a category $\mathfrak{C}$ into a category $\mathfrak{D}$ is a rule which to each object $A \in \mathfrak{C}$ associates an object $\Omega(A) \in \mathfrak{D}$, and to each morphism $f: A \rightarrow B$ associates a morphism $\Omega(f)\colon \Omega(A) \rightarrow \Omega(B)$ so that:
\begin{itemize}
\item[] \textnormal{\textbf{FUN 1.}} For all $A \in \mathfrak{C}$ we have $\Omega(\operatorname{id}_A) = \operatorname{id}_{\Omega(A)}$.
\item[] \textnormal{\textbf{FUN 2.}} If $f\colon A \rightarrow B$ and $g\colon B \rightarrow C$ are two morphisms of $\mathfrak{C}$ then
\[\Omega(g \circ f) = \Omega(g) \circ \Omega(f).\]
\end{itemize}
\end{definition}
 
\glossary{name={$\mathcal{A}$}, description={Coalgebra}, sort=$A$}
\glossary{name={$\Omega$}, description={Endofunctor of an (universal) coalgebra}, sort=$O$}
Now a coalgebra is defined by means of a functor on the category of sets alone, i.e. an \textit{endofunctor}\index{Endofunctor} $\Omega\colon \text{\textbf{Set}} \rightarrow \text{\textbf{Set}}$. It maps sets of states $S$ to sets of (in general) higher complexity, including the Cartesian product of sets or sets of functions, like $S^\Sigma := \{\Sigma \rightarrow S\}$,\glossary{name={$S^\Sigma$}, description={Set of mappings from an alphabet $\Sigma$ to the state set $S$},sort=$S^S$} for two sets $\Sigma$ and $S$. 
\glossary{name={$\alpha_S$}, description={$S \rightarrow \Omega(S)$, mapping of a coalgebra},sort=$AlphaS$}
\begin{definition}\index{Coalgebra}
\cite[Definition 3.0.1]{Gum03} Let a \textsc{Type} be an endofunctor $\Omega\colon \textnormal{\textbf{Set}} \rightarrow \textnormal{\textbf{Set}}$. Then a \textsc{Coalgebra of Type} $\Omega$ is a pair $\mathcal{A} = (S, \alpha_S)$ consisting of a set $S$ and a mapping
\[S \overset{\alpha_S}{\longrightarrow} \Omega(S).\]
\end{definition}

In the following subsections, automata, Kripke structures and LTSA will be presented as universal coalgebras. Since only basic structural similarities are highlighted, set-theoretic morphisms mostly are not considered explicitly.

\subsection{Automata theory}\label{sec:automata}\index{Automaton}

\begin{definition}{\cite[1.4]{Gum03}} An \textsc{Automaton} is a tuple
\[A:=\{S, \Sigma, \delta, D, \gamma\}, \text{ with:}\]
\begin{enumerate}
\item A set of \textsc{States} $S$.\glossary{name={$S$}, description={State set}}\index{State}
\item A finite set of \textsc{Input Symbols} $\Sigma$.\glossary{name={$\Sigma$}, description={Input symbols (automaton)},sort=$Sigma$}
\item A \textsc{Transition Function} $\delta\colon S \times \Sigma \rightarrow \mathfrak{P}(S)$. \glossary{name={$\delta$}, description={Transition function}, sort=$Delta$}
\item A set of \textsc{Data} $D$.\glossary{name={$D$}, description={Data (automaton)}}
\item An \textsc{Output Function} $\gamma\colon S \rightarrow D$.\glossary{name={$\gamma$}, description={Output function},sort=$Gamma$}
\end{enumerate}
In the case of a \textsc{Deterministic Automaton}\index{Process!deterministic}, the transition function $\delta$ maps to the set of singletons identifiable with $S$. A \textsc{Finite Automaton} has a finite set $S$.
\label{def:automaton}
\end{definition}
The value of the transition function\index{Transition function} $\delta(s,e)$ can be interpreted as denoting the possible states the automaton is in after reading the input $e$ while in the state $s$. However, states often cannot be observed directly, but by means of the output function\index{Output function} $\gamma$: Each internal state $s \in S$ can only be observed by an external attribute $\gamma(s) \in D$. 

In a main field of application only the paths are interesting that lead from a fixed start state $s_0$ to a \textit{final} (or accepting) \textit{state}\index{State!final} $s \in F \subseteq S$ ($S$ finite). $F$ may be coded by its characteristic function $\gamma\colon S \rightarrow D$ where $D:= \{0,1\}$. This type of automaton
is also called \textit{acceptor}\index{Automaton!Acceptor} \cite[p.~165]{Gum03}. Then, an automaton defines a language of all successful words $(a_1,...,a_n) \in \Sigma^n,\: n \in \mathbb{N}$, corresponding to a path of transitions from $s_0$ to an accepting state $s_n$.

\begin{figure}
 \centering
 \includegraphics[width=10cm]{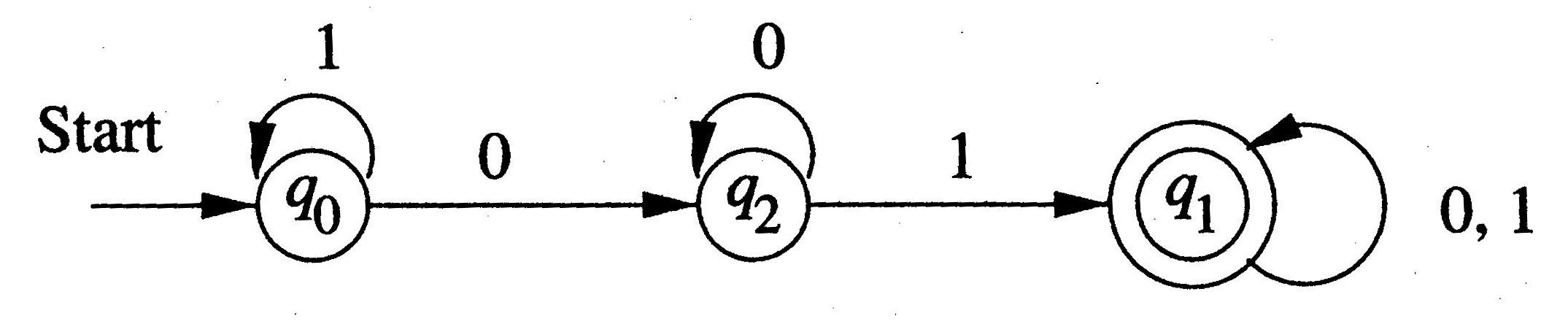}
 \caption{Transition diagram for the automaton example \ref{ex:01automaton} (recognition of a 01-substring), from \cite[p.~48]{Hop01}.}
 \label{fig:01automaton}
\end{figure}
\begin{example}\label{ex:01automaton}
\upshape \cite[p.~46-49]{Hop01}. A finite deterministic automaton accepts strings over an alphabet $\Sigma := \{0,1\}$, and the aim is to decide whether the string contains the sequence $01$. Besides the final state $q_1$, indicating that the substring has been found, there is the initial state $q_0$ (no input or last input 1) and an intermediate state $q_2$\index{State} (most recent input 0). The graph of Figure \ref{fig:01automaton} represents the possible transitions depending on the next input symbol. Explicitly, three types of input words are distinguished, each being preceded by
the corresponding state of the automaton:

\[q_0\colon \underset{\geq 0}{\underbrace{1...1}}, 
 \qquad q_2\colon\underset{\geq 0}{\underbrace{1...1}}\underset{> 0}{\underbrace{0...0}},
 \qquad q_1\colon(...)01(...)\]
\end{example}

In order to make a clear difference to Kripke structures (Section \ref{sec:Kripke}), a deterministic automaton is represented as a coalgebra: Let a functor $\Omega\colon \textnormal{\textbf{Set}} \rightarrow \textnormal{\textbf{Set}}$ be defined by $\Omega(S) := D \times S^\Sigma$ for a set $S$. For $f\colon S_1 \rightarrow S_2$, define $\Omega(f)\colon \Omega(S_1) \rightarrow \Omega(S_2)$ by $\Omega(f)(d, \delta') := (d, f \circ \delta')$, for all $d \in D,\: \delta' \in X^\Sigma$. The properties of a set functor are fulfilled, since the constant functor, the power functor and the cross product of functors are functors \cite[Beispiel 2.4.2, 2.4.3 and 2.4.6]{Gum03}). Then a deterministic automaton is a coalgebra $\mathcal{A} = (S, \alpha_S)$ of type $\Omega$ with
\glossary{name={$S^\Sigma$}, description={Set of mappings from an alphabet $\Sigma$ to the state set $S$},sort=$S^S$}
\[\alpha_S\colon S \rightarrow \Omega(S):= D \times S^\Sigma,\]
where $\alpha_S (s) = (\gamma(s), \delta(s, \cdot))$, for all $s \in S$.

A nondeterministic automaton is identifiable with a coalgebra where
\[\alpha_S\colon S \rightarrow \Omega(S):= D \times (\mathfrak{P}(S))^\Sigma.\]

\subsection{Kripke structures}\label{sec:Kripke}
\index{Kripke structure}
\begin{definition}\label{def:Kripke}
A \textsc{Kripke structure} consists of a set of states $S$, a set of atomic propositions $M$, an output function $\gamma\colon S \rightarrow \mathfrak{P}(M)$ and a relation $R \subseteq S \times S$.
\end{definition}

A Kripke structure may be considered as a special case of a nondeterministic automaton:\index{Process!nondeterministic} With the trivial alphabet $\Sigma = \{\text{update}\}$, i.e., without a special input, the transition function\index{Transition function} gets $\delta_R\colon S \rightarrow \mathfrak{P}(S)$, which can be identified with a relation $R_\delta \subseteq S \times S$ as follows:
\begin{align*}
R_\delta &:= \{(s,s')\mid s \in S, s' \in \delta_R(s)\}\\
\delta_R &: \quad s \mapsto \{s' \in S\mid (s,s') \in R_\delta\} =: [s]R_\delta
\end{align*}

The output function\index{Output function} is given over a set $M$ by letting $D:= \mathfrak{P}(M)$, i.e. $\gamma\colon S \rightarrow \mathfrak{P}(M)$. An arbitrary set $M$ may be considered as a set of atomic propositions. Then $\gamma(s)$ is the set of atomic propositions being true in the state $s$. 

Accordingly, a Kripke structure is a universal coalgebra with
\[\alpha_S\colon S \rightarrow \Omega(S):= \mathfrak{P}(M) \times \mathfrak{P}(S).\]

\begin{example}\label{ex:Kripke2}\index{State}\index{Transition}
\upshape (compare \cite[p.~167]{Gum03}).
Due to the number of interacting measuring or control devices with observable output $M$, a technical or engineering system may not be predictable in detail, but the set of allowed transitions $R$ may be restricted indirectly. Thus, in a computer program with parallel processes, $R$ constrains the transitions between states within and between processes. States are defined by an assignment of the variables declared for a process. Such variable assignments are an example of atomic propositions $m \in M$ (compare Table \ref{tab:01automatonKripke}), and allowed transitions are given by Boolean expressions over $M$ (\textit{compound attributes} in the language of FCA \cite[p.~101]{Gan05}) relating an \textit{input}\index{State!input} to an \textit{output}\index{State!output} state. If the precondition applies to the input state, then the postcondition has to be true in the subsequent state, as in the example describing the main steps of the German legislation process:
\begin{verbatim}bundestag.vote AND bundesrat.vote AND NOT president.veto ==> law.published.\end{verbatim}
\end{example}

\begin{table}
\centering
\begin{tabular}{|c||c|c|c|c|}\hline
&input.0 &input.1 &position.start &position.final\\ \hline \hline
$q_0^a$ &&&x&\\ \hline
$q_0^b$ &&x&x&\\ \hline
$q_1^a$ &x&&&x\\ \hline
$q_1^b$ &&x&&x\\ \hline 
$q_2$ &x&&&\\ \hline
\end{tabular}
\caption{Assignment of atomic propositions $\gamma\colon S \rightarrow \mathfrak{P}(M)$ for the Kripke structure modelling Example \ref{ex:01automaton} represented as a formal context $(S, M, I)$. It is derived by nominal scaling\index{Scaling!nominal} from the many-valued context $(S, \tilde{M}, W, J)$, where $\tilde{M} := \{\text{input}, \text{position}\}$ and $W:=\{0,1,\text{start},\text{final}\}$. Atomic propositions are mappings $m \in M \subseteq \{\tilde{M} \rightarrow W\}$, i.e. variable assignments.}
\label{tab:01automatonKripke}
\end{table}
Compared to the acceptor type of an automaton, in a Kripke structure a state may be mapped to a set of output attributes, not only to single values like \textit{final state}\index{State!final} or \textit{no final state}. An advantage of a Kripke structure is the possibility of a differentiate state description by a large number of attributes, which will be important for our main biological application. In Example \ref{ex:01automaton}, instead of ``colouring" a state by $D=\{\text{final, not final}\}$, attribute combinations from $M:=\{0, 1, \text{start}, \text{final}\}$ may be assigned. Table \ref{tab:01automatonKripke} represents the output function $\gamma: S \rightarrow \mathfrak{P}(M)$ as a formal context, which will be named \textit{state context} (Definition \ref{def:stateCxt}). $0$ and $1$ describe the input leading to a state, hence only the most recently read value of the input string is relevant. In Kripke structures, input strings are not considered explicitly, since $\Sigma$ contains only one element. However, a part of this information could be preserved by remembering two or more input values as attributes, for instance by $\bigcup_{i \in I}\{0_i,1_i\} \subseteq M$. In our example, the defined set of atomic propositions $M$ is sufficient to distinguish the three states and to get valuable information regarding the dynamic system (cf.\ the interpretation of the states $q_0, q_1, q_2$ in Example \ref{ex:01automaton}).

A relation $R'$ is given by the transitions defined by the automaton graph (Figure \ref{fig:01automaton}): 
\[R' = \{(q_0^a, q_0^b), (q_0^a, q_2), (q_0^b, q_0^b), (q_0^b, q_2), (q_2, q_2), (q_2, q_1^b),\]
\[(q_1^a, q_1^a), (q_1^a, q_1^b), (q_1^b, q_1^a), (q_1^b, q_1^b)\}.\]
Inversely, in this example the transition graph may be reconstructed from the Kripke structure, together with some drawing conventions. For this purpose, it might be desirable to have an equivalent of the original automaton states: $\{q_0\}$ and $\{q_1\}$ are concept extents, whereas $\{q_2\}$ is given as extent of a supplementary attribute \agrave intermediate".

\subsection{Labelled Transition Systems with Attributes (LTSA)}\label{sec:LTSA}
\index{Labelled Transition System with Attributes (LTSA)}
For our FCA process modelling, we start from the definition of LTSA in \cite[Definition 1]{Rud01} generalising abstract automata. \textit{Labelled Transition Systems (LTS)} or \textit{State transition systems} are finite state automata, where all states are final states,\index{State!final} or semiautomata.\index{Semiautomaton} They are also used in \textit{operational semantics} and may be described by \textit{process algebras}.\footnote{Keijo Heljanko, Networks and Processes: Process Algebra, 2004. \texttt{www.fmi.uni-stuttgart.de/} \texttt{szs/teaching/}\texttt{ws0304/nets/slides20.ps}} The notion of LTSA complements this structure by state attributes or atomic propositions in the language of Kripke structures:
\begin{definition}\label{def:LTSA}\glossary{name={LTSA}, description={Labelled Transition System with Attributes}}
A \textsc{Labelled Transition System with Attributes (LTSA)} is a 5-tuple $(S,M,I,A,R)$ with
\begin{enumerate}
\item $S$ being a set of states.
\item $M$ being a set of state attributes.
\item $I \subseteq S \times M$ being a relation.
\item $A$ being a finite set of actions.\glossary{name={$A$}, description={Set of actions}}
\item $R \subseteq (S \times A \times S)$ being a set of transitions, where $(s_1, a, s_2) \in R$ means \agrave action $a$ can cause the transition from state $s_1$ to state $s_2$". \end{enumerate}
\end{definition}

\begin{proposition}\label{prop:LTSAAutomaton}
Every automaton can be represented by an LTSA, and vice versa.
\end{proposition}
\begin{proof}
An LTSA is a an automaton according to Definition \ref{def:automaton}, with the specialisation $D:=\mathfrak{P}(M)$ for the output function $\gamma\colon S \rightarrow D,$ \mbox{$s \mapsto \{m \in M \mid (s,m) \in I\}$}. For the rest, $\Sigma:=A$, and $\delta\colon S \times A \rightarrow \mathfrak{P}(S)$ is given by $(s,a) \mapsto \{s' \in S\mid (s,a,s') \in R\}$. (Compare \cite[p.~343]{Wol02}.)

On the other hand, each automaton can be translated into an LTSA via $A:= \Sigma$, $M:=D$ and $R:= \{(s^{in},a,s^{out})\mid s^{in} \in S,\: a \in \Sigma,\: s^{out} \in \delta(s^{in},a) \subseteq\mathfrak{P}(S)\}$; the relation $I$ is given by the function $\gamma\colon S \rightarrow D \cong \{\{d\} \mid d \in D\} \subseteq \mathfrak{P}(D)$.
\end{proof}

By the proposition, an LTSA is also a universal coalgebra:
\[S \rightarrow \Omega(S) := \mathfrak{P}(M) \times (\mathfrak{P}(S))^A \]

In our FCA model, start and final states\index{State!final} are not considered explicitly. Investigating the attribute logic we are not mainly interested in automata as acceptors, i.e. in testing allowed languages. The definition in \cite{Rud01} is slightly changed assuming a finite set of actions according to automata theory. Infinite words over the alphabet $\Sigma := A$ are allowed, but an infinite set of actions is not very meaningful.

\subsection{TCA -- LTSA -- automata theory}\label{sec:TCALTSA}
There is also a strong relation to TCA: An LTSA and especially an automaton can be described by and reconstructed from a CTSOT. In \cite[2.2]{Wol02}, an automaton is defined like an acceptor, but slightly more general: a set of start states is admitted, thus an output function $\gamma\colon S \rightarrow D := \{\text{start, intermediate, final}\}$. The main idea of the \textit{Map Reconstruction Theorem}\index{Theorem!Map Reconstruction} \cite[4.2]{Wol02} is taking the set of actions (plus ``missing value'') as attributes of the time part $\mathbf{T}$ with nominal scaling. Then given an LTSA $L$, there is an isomorphism from $L$ onto a \textit{state-LTSA} derived from a CTSOT, so that each path of $L$ (given by the transitions $R$) is mapped onto a life track.

\section{Temporal logics}
The atomic propositions $M$ of a Kripke structure and the transition relation $R$ define the semantics of a dynamic system. Based on it, a multitude of logics have been developed in order to reason about temporal properties of the system. Within the framework of the present thesis two classical approaches and a temporal extension of description logics are outlined, which have been developed and investigated in many directions more recently.

\subsection{Propositional tense logic}\label{sec:Prior}\index{Propositional tense logic}
A very important contribution to the modern logic of time -- including concepts and reasoning -- was made by A.N. Prior (1914 - 1969) based on philosophical traditions of antiquity and the Middle Ages. He started from J.M.E. McTaggarts (1866 - 1925) distinction between the A- and B-series conceptions of time\index{A- and B-concepts of time} (which lead McTaggart to a famous paradox and the refutation of the reality of time): The A-concepts \textit{past}, \textit{present} and \textit{future} are more fundamental for a proper understanding of time than the B-series conception of a set of instants organised by the earlier-later relation. Prior also rejected the latter static view of time, intending to substantiate the notion of freedom:
\begin{quote}
\textit{I believe that what we see as a progress of events is a progress of events, a 
coming to pass of one thing after another, and not just a timeless tapestry with everything stuck there for good and all. (A.N. Prior, cited after \cite[p.~69]{Ohr09}})
\end{quote}
Events that have become past are \agrave out of our reach" and unchangeable, whereas the future is to some extent open and depends on the decision of a free agent. Furthermore, Prior considered B-theory as a reduction of reality since the notion of the present, the Now, disappears.

In both views, well formed formula are composed by arbitrary propositional variables, $\wedge$, $\sim$ (negation), $F$ (``in the future'')\glossary{name={$F$}, description={Eventually}} and $P$ (``in the past'').
Yet like in one tradition of modern philosophy of language, Prior did not assume a sharp distinction between an object language and a metalanguage. Accordingly, there is no model,\index{Model} e.g. a Kripke structure (denoted by $(\text{TIME}, <, \nu)$) as a second level. It would be a tenseless metalanguage or ``timeless tapestry'' since it is based on a set of instants or durations. Instead of such a reification of instants, Prior introduced a special type of propositions, \textit{Instant-propositions} or \textit{World-state propositions} $t$. This makes it possible to define formula $T(t,\varphi)$ meaning that a tense-logical formula $\varphi$ is true at time $t$. Instant propositions are defined axiomatically in terms of the tense-logical language itself, together with a necessity operator $\Box$\glossary{name={$\Box$}, description={Necessity operator},sort=$.Nec$} and a possibility operator $\Diamond:= \,\sim \!\Box \sim$,\glossary{name={$\Diamond$}, description={Possibility operator},sort=$.Pos$}\glossary{name={$\sim$}, description={Negation},sort=$.Neg$} as well as standard quantification $\exists$ (and thus $\forall$). In this way, instants or times are treated as artificial constructs. They are replaced by the conjunction of a maximal consistent set of propositions that may be said to be true at $t$. Thus, Prior adapts the notion of possible worlds to time.\footnote{Compare John Wood, \textit{Course Modal Logic}, University of British Columbia, Spring 2007, Note 23 \texttt{[http://www.johnwoods.ca/Courses/Phil322-07/].}} His approach has also been followed by hybrid logics. \cite[p.~70f.]{Ohr09} 

He developed a theory of possibility and indeterminism\index{Process!nondeterministic} based on the notion of \textit{branching time}\index{Branching time}, which had been suggested in a letter from Saul Kripke in 1958. Branching time is representable by a tree, where the present is a node of \agrave rank 0", and the possible future states at the following moments have a higher rank, i.e.~depth. Finally, Prior incorporated this idea into the concept of time itself by introducing the notion of \textit{chronicles}\index{Chronicle} or \textit{histories}\index{History}, i.e. maximal linearly ordered subsets in $(\text{TIME}, <)$, where TIME is the set of instant-propositions.

Prior distinguished two models of branching time, inspired by William of Ockham (ca. 1285 - 1349) and Charles S. Peirce (1839 - 1914), respectively. In the Ockhamistic model, the operators $F$ (\agrave tomorrow"), $\Diamond F$ (\agrave possibly tomorrow") and $\Box F$ (\agrave necessarily tomorrow") are distinguishable. In the Peircean view, however, \agrave tomorrow" is identified with \agrave necessarily tomorrow", since Peirce emphasised the difference between future and past. There is no \agrave plain" or \agrave true future" and no factual, non-necessary statements concerning the future make sense. Consequently, for an arbitrary formula $q$ only in the Ockhamistic system $q \rightarrow HF q$ is a theorem, with $H := \,\sim \! P\! \sim$ (\agrave at all times in the past"). In both systems, there are no alternative pasts but a single chronicle in the past of an instant-proposition $t$. Hence, the theorem does not hold in the Peircean system but only for $F \neq \Box F$, since $\Box F$ refers to all possible chronicles. \cite[p.~72-77]{Ohr09}

Priors favorised A-theory is \agrave politically correct" within the context of contemporary philosophy, but it is too sophisticated within the framework of this study. Instead we remain with the usual model-theoretic, B-theoretic approach to time, since FCA is based on data, a fundamental concept of modern science. A data frame of observations representing -- or statistically interpreted as -- a deterministic time series is a \agrave timeless tapestry". For the observation time, future is considered as retrospective or \agrave plain", and the observations are extrapolated to future by postulating natural laws. Prior reminds us that this is a simplification. Sometimes the whole background of data analysis and theory building should be made conscious again, like abduction, induction, falsification, paradigm change and historical development of scientific notions. 
\begin{quote}
\textit{But even in this flux there is a pattern, and this pattern I try to trace with my tense-logic; and it is because this pattern exists that men have been able to construct their seemingly timeless frame of dates. Dates, like classes, are a wonderful
and tremendously useful invention, but they are an invention; the reality is
things acting.} (A.N. Prior, Bodleian Library, MS in box 6, 1 sheet, no title. Cited after \cite[p.~78]{Ohr09})\end{quote} 

We follow an intermediate approach: Within the FCA model for transitions that will be developed in Chapter \ref{ch:FCAModelling} we include nondeterminism, hence future is regarded as open. Furthermore, we decided for throughout quantification over paths according to the Peircean future operators (but also for an Ockhamistic model, the notion of a chronicle could be defined by a supplementary attribute of the state context (Definition \ref{def:stateCxt}). In this way, there is no \agrave true future", only possibility and necessity are considered. On the other side, we follow a data-driven approach but try to be cautious: Analysis results depend on specific experimental conditions, preprocessing, definition of thresholds (e.g. for gene expression up- or downregulation) or choice of algorithms. Hence, necessity should not be judged by given data only, but by all existing knowledge. Attribute exploration enforces the role of the expert, who ideally should interprete necessity before the background of all events possible according to the state of science. In a further step of reflection, the resulting temporal stem base can be considered as a clear and concise knowledge representation, which by determining the domain of interest also defines its limits as well as those of the present understanding of temporal reality. Against this background, a clear definition and validation of a temporal data frame implies the respect for the potentially infinite complexity of nature and life.

\subsection{Computation Tree Logic (CTL)}\label{sec:CTL}
\glossary{name={CTL}, description={Computation tree logic}}
\index{Computation Tree Logic (CTL)}
Current temporal logics include \textit{Interval Temporal Logic (ITL)} and \textit{$\mu$-calculus}, of which important subsets are \textit{Linear Temporal Logic (LTL)}\glossary{name={LTL}, description={Linear temporal logic}} and \textit{Computation Tree Logic (CTL)}. Thus, like all approaches considered here, CTL abstracts from duration values; the basic unity is an event corresponding to a state. CTL is able to describe properties of nondeterministic transition systems (with branching time)\index{Process!nondeterministic}\index{Branching time} and extends propositional or first order logic with the following path quantifiers and temporal operators defined formally in Section \ref{sec:tmpCxt}:
\begin{itemize}
 \item $A$: \agrave for \textbf{a}ll transition paths" (corresponding to Priors necessity operator $\Box$)\glossary{name={$\Box$}, description={Necessity operator},sort=$.Nec$} 
 \item $E$: \agrave for some (\textbf{e}xisting) transition path" (corresponding to $\Diamond:= \,\sim \!\Box \sim$)\glossary{name={$\Diamond$}, description={Possibility operator},sort=$.Pos$}
 \item $F$: ``eventually (\textbf{f}inally) in the future"\glossary{name={$F$}, description={Eventually}}
 \item $G$: \agrave always (\textbf{g}enerally) in the future''\glossary{name={$G$}, description={Always}}
 \item $X$: ``ne\textbf{x}t time"\glossary{name={$X$}, description={Next}}
 \item $U$: ``\textbf{u}ntil''\glossary{name={$U$}, description={Until}}
\end{itemize}

A \textit{safety}\index{Safety} property specifying that some situation described by a formula $\phi$ can never happen is expressed by the CTL formula $AG\neg \phi$,\glossary{name={$\neg$}, description={Negation},sort=$.Neg$} i.e. on all paths $\phi$ is always false. A \textit{liveness}\index{Liveness} property specifying that something good $\phi$ will eventually happen is expressed by the formula $AF\phi$.

\textit{State formulas} $\phi$ are evaluated on (arbirary) states, whereas \textit{path formulas} $\psi$ are evaluated on single paths. With a set of atomic propositions $AP$, CTL has the following grammar, including ordinary Boolean connectives \cite[4.1]{Cha04}:
\glossary{name={$\phi$}, description={State formula (CTL)}, sort=$Phi$}
\glossary{name={$\psi$}, description={Path formula (CTL)}, sort=$Psi$}
\begin{align*}
  \phi & := \alpha \in AP \mid E\psi \mid A\psi\\
  \psi & := \phi \mid X\psi \mid F\psi \mid G\psi \mid \psi U\psi
\end{align*}
While in CTL$^*$ arbitrary state and path formulas are admitted, in CTL a path formula $\psi$ has to be preceeded by a path quantifier $E$ or $A$. Hence, an Ockhamistic model of time (in Priors understanding) cannot be expressed in CTL.


\subsection{Description logics}\label{sec:DL}
\subsubsection{General definitions}
Description Logics (DLs) are a family of knowledge representation formalisms with a broad range of applications, such as data mining, natural language processing, semantic web or ontologies. Similar to FCA, a principal aim of a DL is to define conceptual hierarchies, and there are attempts aiming at the application of attribute exploration to construct DL knowledge bases \cite{Baa09b} \cite{Baa09a}.

\glossary{name={DL}, description={Description logics}}\index{Description logics (DL)}
\glossary{name={$N_C$}, description={Concept names (DL)}, sort=$NC$}
\glossary{name={$N_R$}, description={Role names (DL)}, sort=$NR$}
\textit{Concept descriptions} are built starting from a set $N_C$ of \textit{concept names} (unary predicates) and a set $N_R$ of \textit{role names}\index{Role (DL)} (binary predicates) with the aid of \textit{concept constructors} specific for the language. Important constructors are the conjunction $C \sqcap D$ and the disjunction $C \sqcup D$, where $C,D$ are concept names or more complex concept descriptions. If DL concepts are given as FCA concepts, constructors are identical to the infimum or supremum of two formal concepts. Besides negation, role restrictions are used, e.g. $\exists r.C\: (r \in N_R)$ meaning, e.g., the concept of having a French team member for a role $r$: (\texttt{FootballTeam, Member}) and a concept description $C$ \agrave French nationality". Finally, \textit{individual names} are assembled in a set $N_I$ referring to elements of the \textit{domain} $\Delta^I$ by means of which the semantics of a DL is given (Table \ref{tab:DLSemantics}).

\index{TBox}\index{ABox}
A DL knowledge base consists of a \textit{TBox} and a \textit{ABox}. The \textit{TBox} is a finite set of \textit{general concept inclusions (GCIs)}\glossary{name={GCI}, description={General concept inclusion (DL)}} $C \sqsubseteq D$ expressing a subconcept-superconcept relation. Concept definitions $C \equiv D$ abbreviate two GCIs holding for both directions. The $ABox$ assigns concepts and roles to individual names.

\begin{table}
\index{Semantics}\index{Syntax}
\centering
\glossary{name={$\bot$}, description={Complete attribute set}, sort=$.Complete$}
\begin{tabular}{lll}
\textit{\textbf{Standard DL}} &Syntax &Semantics\\ \hline \hline
\textit{Basic sets} &&\\
$\quad$ Concept &$C \in N_C$ &$C^\mathcal{I} \subseteq \Delta^\mathcal{I}$\\
$\quad$ Empty concept &$\bot$ &$\bot^\mathcal{I} = \emptyset \subset \Delta^\mathcal{I}$\\
$\quad$ Most general concept &$\top$ &$\top^\mathcal{I} = \Delta^\mathcal{I}$\\
$\quad$ Role &$r \in N_R$ &$r^\mathcal{I} \subseteq \Delta^\mathcal{I} \times \Delta^\mathcal{I}$\\
$\quad$ Individual name &$a \in N_I$ &$a^\mathcal{I} \in \Delta^\mathcal{I}$\\ \hline
\textit{Constructors} &&\\
$\quad$ Negation &$\neg C$ &$\Delta^\mathcal{I} \setminus C^\mathcal{I}$\\
$\quad$ Conjunction &$C \sqcap D$ &$C^\mathcal{I} \cap D^\mathcal{I}$\\
$\quad$ Disjunction &$C \sqcup D$ &$C^\mathcal{I} \cup D^\mathcal{I}$\\
$\quad$ Existential restriction &$\exists r.C$ &$\{x \in \Delta^\mathcal{I} \mid \exists y \in \Delta^\mathcal{I}\colon (x,y) \in r^\mathcal{I} \wedge y \in C^\mathcal{I}\}$\\ 
$\quad$ General restriction &$\forall r.C$ &$\{x \in \Delta^\mathcal{I} \mid \forall (x,y) \in r^\mathcal{I}\colon y \in C^\mathcal{I}\}$\\ \hline
\textit{TBox} &&\\
$\quad$ GCI &$C \sqsubseteq D$ &$C^\mathcal{I} \subseteq D^\mathcal{I}$\\
\textit{ABox} &&\\
$\quad$ Concept assertion &$C(a)$ &$a^\mathcal{I} \in C^\mathcal{I}$\\\glossary{name={$C(a)$}, description={Concept assertion (DL)}}
$\quad$ Role assertion &$r(a,b)$ &$(a^\mathcal{I}, b^\mathcal{I}) \in r^\mathcal{I}$\\ \hline
\glossary{name={$r(a,b)$}, description={Role assertion (DL)}, sort=$Rab$}
\end{tabular}
\caption{Semantics of DL concept descriptions, TBoxes and ABoxes in terms of an interpretation $\mathcal{I}:= (\Delta^\mathcal{I}, \cdot^\mathcal{I})$, where $\Delta^I$ is a nonempty set, and the interpretation function $\cdot^\mathcal{I}$ is defined as above (compare \cite[p.~4f.]{Baa09a}).}
\label{tab:DLSemantics}
\end{table}

\subsubsection{Temporal extensions of description logics}\label{sec:tempDL}
A basic possibility of applying a DL to dynamic processes is to interprete the domain as a set of states and to describe transitions by roles \texttt{nextState} or \texttt{reachableState} (see Section \ref{sec:compDL}, where a translation of our approach into the language of DL will be discussed.)

There are also several extensions of DL by operators from the temporal logic LTL together with the definition of an appropriate semantics, e.g. in \cite[6.2.4]{Baa03}. I will delineate briefly the more detailed proposal of the temporal extension $\mathcal{TDL}-Lite_{bool}$ \cite{Art07}.\glossary{name={$\mathcal{TDL}-Lite_{bool}$}, description={Temporally extended DL}, sort=$TDL$}

Besides the usual \textit{global roles} $T_i,\:i \in I$, time-dependent \textit{local roles} $P_i$\index{Role (DL)!local} are introduced, denoting a one step transition like \texttt{nextState}. The semantics of a role $P_i^-$ is the inverse relation on the domain (the state set) $\Delta$. Furthermore, the temporal operators $\bigcirc$ (\agrave at the next moment") and $\mathcal{U}$ (\agrave until") are introduced. The semantics of the non-standard operators is listed in Table \ref{tab:TDLSemantics} with an interpretation according to Definition \ref{def:TDLInt}. The complete syntax is defined as follows:
\begin{align*}
R :=\: &P_i \mid P_i^- \mid T_i \mid T_i^-\\
B :=\: &\bot \mid A_i \mid \, \geq q\, R\\
C :=\: &B \mid \neg C \mid C_1 \sqcap C_2 \mid \bigcirc C \mid C_1\, \mathcal{U}\, C_2.
\end{align*}

\begin{definition}\label{def:TDLInt}
\glossary{name={$\Delta$}, description={Domain (DL)}, sort=$Delta$}
Given a nonempty set (domain) $\Delta$, object names $a_i$, concept names $C_i$, local role names $P_i$ and global role names $T_i$, $i \in I$, the semantics of $\mathcal{TDL}-Lite_{bool}$ is given by an \textsc{Interpretation Function} $\mathcal{I}$:
\[\mathcal{I}(n) := (\Delta, a_0^{\mathcal{I}(n)},..., A_0^{\mathcal{I}(n)},..., P_0^{\mathcal{I}(n)},..., T_0^{\mathcal{I}(n)},...),\]
where $n \in \N,\: a_i^{\mathcal{I}(n)} \in \Delta,\: A_i^{\mathcal{I}(n)} \subseteq \Delta,\: P_i^{\mathcal{I}(n)} \subseteq \Delta \times \Delta,\: (i \in I)$, and $a_i^{\mathcal{I}(n)} = a_i^{\mathcal{I}(m)},\: T_i^{\mathcal{I}(n)} = T_i^{\mathcal{I}(m)}$ for all $n,m \in \N$.
\end{definition}
Finally, the usual unique name assumption is made: $a_i^{\mathcal{I}(n)} \neq a_j ^{\mathcal{I}(n)}$ for all $i \neq j$. The two operators $\bigcirc$ and $\mathcal{U}$ are sufficient to define other temporal operators: $F\, C \equiv \top\, \mathcal{U}\, C$ (\agrave some time in the future") and $G\, C \equiv \neg F\, \neg C$ (\agrave always in the future").
\begin{table}
\centering
\begin{tabular}{lll}
$\mathbf{TDL-Lite_{bool}}$ &Syntax &Semantics\\ \hline \hline
At the next moment &$\bigcirc$\,$C$ &$\bigcirc\, C^{\mathcal{I}(n)} := C^{\mathcal{I}(n+1)}$\\
Until &$C_1\, \mathcal{U}\, C_2$ &$(C_1\, \mathcal{U}\, C_2)^{\mathcal{I}(n)} := \bigcup_{k>n} \huge( C_2^{\mathcal{I}(k)} \cap\, \bigcap_{n<m<k} C_1^{\mathcal{I}(m)} \huge)$\\ 
Inverse role &$R^-$ &$(R^-)^{\mathcal{I}(n)} := \{(y,x) \mid (x,y)\in R^{\mathcal{I}(n)}\}$\\ 
Related states &$\geq q\, R$ &$(\geq q R)^{\mathcal{I}(n)} := \{x \in \Delta \mid |\{y \mid (x,y) \in R^{\mathcal{I}(n)}\}| \geq q\}$\\ \hline
\end{tabular}
\caption{Semantics of the supplementary operators in $TDL-Lite_{bool}$.}
\label{tab:TDLSemantics}
\end{table}


\section{Systems biology}\label{sec:sysBiol}
As models for regulatory networks, linear or nonlinear ordinary differential equations are often used. Partial differential equations additionally make it possible to model spatial behaviour, e.g. cell differentiation and movement. Those models are used primarily for simulation and prediction and offer possibilities of subsequent analyses, e.g. of stability or bifurcation.

\subsection{Discrete models}\label{sec:discreteModels}

Methods of symbolic computation allow for further and differentiated analyses by logical queries. They are closer to the thinking of human experts and can also be applied if quantitative data is sparse or only the qualitative behaviour is known \cite[Introduction]{Cha04} \cite{Hec09}. For instance, methods of software or hardware verification have been adapted to systems biology, like the \textit{$\pi$-calculus},\index{$\pi$-calculus} a process algebra for concurrent computation. Molecules are represented as processes in which they participate and interactions as communication channels. The $\pi$-calculus is used for simulation and verification of assertions like \agrave Will a signal reach a particular molecule{?}`` \cite{Reg04}.

CTL was used in \cite{Cha04} to analyse protein-protein and protein-DNA interaction networks. The approach developed in the following chapters is based on CTL and Boolean networks.

\subsection{Boolean networks}\label{sec:BooleNets}
Boolean networks (BN)\index{Boolean network}\glossary{name={BN}, description={Boolean network}}
 are often applied to the analysis of gene regulation. By abstraction only two expression levels \textit{off} and \textit{on} (0/1 or -/+) are considered. This is justified since there exist relatively fixed thresholds of activation for many genes \cite{Ros03}. Also in continuous models the dynamics are often approximated by a steep sigmoid function (e.g. $f(t) := \frac{1}{1+\operatorname{e}^{-t}})$. Moreover, a switch-like behaviour may be strengthened by a positive feedback of a transcription factor on its own expression \cite[p. 14797]{Kau03}. The classical approach of Boolean networks\index{Boolean network} \cite{Kau69} \cite{Kau93} is able to capture essential dynamic aspects of regulatory networks and scales up well to larger sets of genes. Boolean networks require time-series data as input (\textit{reverse engineering}) and generate such data as output (\textit{simulation}). They can be represented as directed graphs with nodes labelled by Boolean functions, which determine one of two attribute values 0 or 1 for each entity (e.g. gene) after one time step (output) given the values of the entities\index{Entity} at a given moment (input). Boolean networks are widely used in molecular biology for logical analysis and simulation of medium or large scale networks \cite{Kla06} \cite{Ste07}. For example, Kervizic et al. developed a method for the cholesterol regulatory pathway in 33 species which eliminates spurious cycles in a synchronous Boolean network model \cite{Ker08}. A formal definition within our conceptual framework will be given in \ref{def:BoolNet}.

\chapter{Modelling discrete temporal transitions by FCA}\label{ch:FCAModelling}

Our intention was to develop an FCA approach into which different types of process models may be translated. In the following will be demonstrated how the types of universal coalgebras presented in Section \ref{sec:coalg} are representable by (a family of) \textit{transition contexts} (Section \ref{sec:trCxt}). First, a basic \textit{state context} will be defined, then in Section \ref{sec:transCxt} a \textit{transitive context} which makes the information related to reachable states explicit. Finally, attributes from the temporal logic CTL are integrated into a \textit{temporal context}.

\section{Example: Installing a wireless device}\label{sec:runEx}
In order to illustrate the definitions, the method and possible applications, I introduce a very simple example. More realistic applications to systems biology will be described in Chapters \ref{ch:bSubtilis} and \ref{ch:geneRegNets}. 

A Linux (Ubuntu) help page aims at guiding a user through the process of installing a wireless card and establishing an internet connection. The definition of formal contexts and attribute exploration (which will be described in Chapter \ref{ch:useAttrEx}) may support a good structure of this page (e.g. by hyperlinks to subsequent steps) and can prevent to forget occurring cases. It might even be used to determine the process logic with the purpose of establishing an expert system. 

The formal context of Table \ref{tab:wlanStateCxt} relates states to attributes indicating which of four main steps of the installation process (including an alternative) are accomplished. In Definition \ref{def:stateCxt} this type of contexts will be called state context and referred to an LTSA.
\begin{table}
\begin{center}
\begin{tabular}{|l|c|c|c|c|}\hline
&\texttt{driver.linux} &\texttt{ndiswrapper} &\texttt{driver.windows} &\texttt{connection}\\ \hline \hline
$s_0$ &&&&\\ \hline
$s_{10}$ &$\times$&&&\\ \hline
$s_{11}$ &&$\times$&&\\ \hline
$s_{20}$ &$\times$&&&$\times$\\ \hline
$s_{21}$ &&$\times$&$\times$&\\ \hline
$s_{31}$ &&$\times$&$\times$&$\times$\\ \hline
\end{tabular}
\end{center} 
\caption{Formal context indicating in the columns (attributes) which of four main steps of a wireless card installation process are accomplished. The row names (objects) denote succeeding states. The indices suggest a branching after the initial state $s_0$ into procedures for a native Linux driver (states $s_{i0}$) or for a driver originally developed for Windows operating systems (states $s_{i1}$).}
\label{tab:wlanStateCxt}
\end{table}
\begin{itemize}
 \item \texttt{driver.linux}: Newer versions of Ubuntu provide full ``out of the box'' support for several wireless cards. In other cases, the driver has to be downloaded manually, unpacked to an appropriate directory and compiled.
 \item \texttt{driver.windows}: Often, no open source driver exists; then the Windows driver can be used.
 \item \texttt{ndiswrapper}: The Linux module \textit{ndiswrapper} has been developed with the purpose of using a Windows driver. Since it does not belong to the basic Ubuntu distributions, it must be installed first.
 \item \texttt{connection}: Finally, the individual connection data is entered (usually ESSID and password for the router). At this basic stage of process modelling, \texttt{connection} is the final state\index{State!final} and signifies overall success, i.e. an established internet connection.
\end{itemize}

\section{The state context $\Ks$ and some useful scalings}\label{sec:stateCxt}
In order to investigate a process, the occurring states have to be defined first. We do this by means of their attributes, i.e. by a formal context.
\index{State context}
\glossary{name={$\Ks$}, description={State context}, sort=$Ks$}
\glossary{name={$M$}, description={Attribute set of a state context}}
\glossary{name={$I$}, description={Relation of a state context}}
\begin{definition}\label{def:stateCxt} \cite[p.~147]{Rud01}
The formal context $\mathbb{K}_s:=(S,M,I)$ with the state set $S$, the attribute set $M$ and relation $I \subseteq S \times M$ from an LTSA $(S,M,I,A,R)$ is called the \textsc{State Context} of this LTSA. 
\end{definition}
For an example see Table \ref{tab:wlanStateCxt}. A non clarified state context may contain diverse states indiscernable by the attributes, related to different time points or observations. Then, the information regarding time granules may be coded in the object names (compare Table \ref{tab:3GenesCxt}). As for the attributes, we are focusing primarily on the logic of the state space, i.e.~of $\mathfrak{B}(\mathbb{K}_{C})$ in TCA. In (biological) regulatory networks, one is more interested in what happens before or after a certain class of states, and less in exact time points. However, a coarse granularity of time could be useful to describe, e.g., early and late activation of gene expression. For this purpose, our framework may be easily applied to the situation space $\mathfrak{B}(\mathbb{K}_{TC})$ by introducing a supplementary many-valued attribute ``time point'' or ``time interval''.

If the state context is clarified, a state is attribute defined (i.e. unambigously identifiable by its attributes). If further nominal scaling (p. \pageref{nomScaling})\index{Scaling!nominal} is applied, a unique value is assigned to each attribute or variable. Then, a state may be identified with a function \mbox{$s \in F^E := \{E \rightarrow F\}$}
\glossary{name={$F^E$}, description={Set of mappings $\{E \rightarrow F\}$}} with
\begin{itemize}
\item The universe $E$. The elements of $E$ will be called \textit{entities}.\index{Entity}
They represent the objects of the world which we are interested in\glossary{name={$E$}, description={Entities, universe}} (installation steps, measuring devices, genes, etc.).
\item The set $F$ (\textit{fluents}) denotes changing properties of the entities.\glossary{name={$F$}, description={Fluents}}\index{Fluent} It is the union of the scale values, for all $e \in E$.
This descriptive term is adopted from the \textit{fluent calculus}, an agent based modelling and reasoning method \cite{Thi05}.
\end{itemize}
With these restrictions, Definition \ref{def:stateCxt} is equivalent to 
\begin{definition}\label{def:stateCxt1}
Given sets $S$ (states), $E$ (entities), $F$ (fluents) and a function \mbox{$\gamma: S \rightarrow F^E$}, a \textsc{state context} is a formal context
$(S,M,I)$ with $M \subseteq E \times F$. Its
relation $I$ is given as $s\, I\,(e,f) \Leftrightarrow \gamma(s)(e) = f$, for all $s \in
S,\: e \in E$ and $f \in F$.
\end{definition}
Thus, in the language of automata theory and Kripke structures, the output function\index{Output function} $\gamma: S \rightarrow D$ maps to the data set $D:= F^E \subseteq \mathfrak{P}(E \times F)$. A state is a function name and a row of the context defines the mapping.

Besides nominal or dichotomic (for $F :=\{0,1\}$) scaling as in Definition \ref{def:stateCxt1},  different scalings\index{Scaling} may be useful, if the state context is derived from a many-valued context $(S,E,F,J)$. \textit{Biordinal scaling} (\cite[Definition 31]{GW99}, Table \ref{tab:biordScaling}) differentiates low and high measured values into several classes according to thresholds. Simultaneously a coarser and finer ``clustering'' of observed values may be expressed, as well as imprecise knowledge: Intermediate scale values can be represented without loss of information for the extreme values. This is biologically relevant, if for instance a transcription factor activates or inhibits different genes at different expression levels.

A similar scale \cite[Figure 3]{Gan05} may also be appropriate in the case of possible imprecise measuring or if no precise threshold of effectiveness (\agrave high") is known. In addition to \agrave low" (e.g. $\leq 300$) and \agrave high" ($\geq 600$), this scale contains the attributes \agrave not low" ($> 300$) and \agrave not high" ($<600$) expressing intermediate values. Of course, a scale can have even more discretisation steps (scale attributes).
\begin{table}
\begin{center}a)
\begin{tabular}[t]{|l||c|c|}\hline
&ETS1 &SMAD4\\ \hline \hline
$s_0$ &280&305\\ \hline
$s_1$ &345&567\\ \hline
$s_2$ &628&410\\ \hline
\end{tabular}
$\qquad$
b)
\begin{tabular}[t]{|l||c|c|c|c|}\hline
 &$\leq 300$ &$\leq 450$ &$> 450$&$\geq 600$ \\ \hline \hline
280 &$\times$&$\times$&&\\ \hline
305 &&$\times$&&\\ \hline
345 &&$\times$&&\\ \hline
410 &&$\times$&&\\ \hline
567 &&&$\times$&\\ \hline
628 &&&$\times$&$\times$\\ \hline
\end{tabular}\\[5mm]

c)
\begin{tabular}{|l||c|c|c|c|c|c|c|c|}\hline
&\begin{sideways}ETS1$\leq 300$\end{sideways}
&\begin{sideways}ETS1$\leq 450$\end{sideways}
&\begin{sideways}ETS1$\geq 450$\end{sideways}
&\begin{sideways}ETS1$\geq 600$\end{sideways}
&\begin{sideways}SMAD4$\leq 300$\end{sideways}
&\begin{sideways}SMAD4$\leq 450$\end{sideways}
&\begin{sideways}SMAD4$> 450$\end{sideways}
&\begin{sideways}SMAD4$\geq 600$\end{sideways}\\ \hline \hline
$s_0$ &$\times$&$\times$&&&&$\times$&&\\ \hline
$s_1$ &&$\times$&&&&&$\times$&\\ \hline
$s_2$ &&&$\times$&$\times$&&$\times$&&\\ \hline
\end{tabular}
\end{center}\label{tab:biordScaling}
\caption{A small part of a single gene expression time series for cells stimulated with TGF$\beta$1; the complete data set will be analysed in Chapter \ref{ch:geneRegNets}. The states $s_0, s_1$ and $s_2$ represent mRNA measurements of the transcription factors ETS1 and SMAD4 at three time points. \textit{a)}: original data, considered as a many-valued context (Definition \ref{def:MVcontext}).\index{Formal context!many-valued} \textit{b)}: discretising biordinal scale (Definition \ref{def:Scaling}).\index{Scaling!biordinal} \textit{c)}: derived one-valued context.}
\end{table}

\section{The transition context $\Kt$}\label{sec:trCxt}
In order to express dynamics, we need a supplementary structure: a relation $R \subseteq S \times S$\glossary{name={$R$},description={Transition relation}} indicating temporal transitions between the states. The output function $\gamma\colon S \rightarrow \mathfrak{P}(M)$ is representable by a state context $(S, M, I)$ (each row defines $s \mapsto \gamma(s)$, for $s \in S$). Since $R$ is in one-to-one correspondence to a transition function\index{Transition function} $\delta\colon S \rightarrow \mathfrak{P}(S)$, we have a Kripke structure. It is expressed as a unique mathematical structure (more integrated than a tuple of sets and maps like in Definition \ref{def:automaton}), following an approach in \cite {Rud01}. In this work, an \textit{action context} of an LTSA has been introduced, representing the Kripke structure for each action $a \in A$ by a unique formal context and allowing attribute exploration of dynamic properties. Moreover, it follows from the definitions in Chapter \ref{ch:logic} that automata and LTSA are representable by a family of action contexts ${(\Kt^a)}_{a \in A}$, which here are called \textit{transition contexts} (Theorem \ref{theorem:coAlg}). We concentrate on single transition contexts where the LTSA relation $R \subseteq S \times \{a\} \times S$ is identifiable with $R \subseteq S \times S$.

\index{Transition context}\glossary{name={$\Kt$}, description={Transition context}, sort=$Kt$}
\glossary{name={$\nabla$}, description={Relation of a transition context}, sort=$Nabla$}
\glossary{name={$s^{in}$}, description={Input state}, sort=$S^in$}\index{State!input}
\glossary{name={$s^{out}$}, description={Output state}, sort=$S^out$}\index{State!output}
\begin{definition}\label{def:trCxt}
Given a state context $\mathbb{K}_s = (S,M,I)$ and a relation \mbox{$R \subseteq S\times S$}, the \textsc{transition context} $\mathbb{K}_t$ of $\mathbb{K}_s$ with respect to $R$ is the context $(R, M\times\{in,out\},\nabla)$ with relation $\nabla$:
\[(s^{in},s^{out}) \nabla (m,i) :\Leftrightarrow s^i I\, m \qquad \text{for all } m \in M,\: i \in \{in,out\} \text{ and } (s^{in},s^{out}) \in R.\]
\end{definition}
\begin{table}
\centering
\begin{tabular}{|l||c|c|c|c||c|c|c|c|}
\hline
&\begin{sideways}\texttt{driver.linux.in}\end{sideways}
&\begin{sideways}\texttt{ndiswrapper.in}\end{sideways}
&\begin{sideways}\texttt{driver.windows.in}\end{sideways}
&\begin{sideways}\texttt{connection.in}\end{sideways}
&\begin{sideways}\texttt{driver.linux.out}\end{sideways}
&\begin{sideways}\texttt{ndiswrapper.out}\end{sideways}
&\begin{sideways}\texttt{driver.windows.out}\end{sideways}
&\begin{sideways}\texttt{connection.out}\end{sideways}\\
\hline \hline
$(s_0,s_{10})$ &&&&&$\times$&&&\\ \hline 
$(s_0,s_{11})$ &&&&&&$\times$&&\\ \hline
$(s_{10},s_{20})$ &$\times$&&&&$\times$&&&$\times$\\ \hline
$(s_{11},s_{21})$ &&$\times$&&&&$\times$&$\times$&\\ \hline
$(s_{21},s_{31})$ &&$\times$&$\times$&&&$\times$&$\times$&$\times$\\ \hline
$(s_{20},s_{20})$ &$\times$&&&$\times$&$\times$&&&$\times$\\ \hline
$(s_{31},s_{31})$ &&$\times$&$\times$&$\times$&&$\times$&$\times$&$\times$\\ \hline
counterEx$_1$ &&&&&&&$\times$&\\ \hline
counterEx$_2$ &&&$\times$&&&$\times$&$\times$&\\ \hline
\end{tabular}
\caption{Transition context related to four main steps of the Linux installation process for a wireless card. The objects are pairs of succeeding states from Table \ref{tab:wlanStateCxt}. The four left columns denote attributes of the input, the right columns of the output states. The last two rows are counterexamples introduced during attribute exploration (cf. Section \ref{sec:exEx}).}\label{tab:wlanTrCxt}
\end{table}
Thus, a transition context is a subcontext of the semiproduct\index{Semiproduct of formal contexts}\label{subCxtScaling} $\mathbb{K}_s\: \begin{sideways}$\bowtie$\end{sideways}\; \mathbb{K}_s$ (Definition \ref{def:semiproduct}) of a state context with itself. It may be regarded as the context derived from the many-valued context $(R, \{\text{in,out}\}, S, J)$:\label{scaleKsKt}
\begin{center}
\begin{tabular}{|c|c|c|}\hline
&in &out\\ \hline
$t_0$ &$s_0$ &$s_1$\\ \hline
...&&\\ \hline
$t_n$ &$s_i$ &$s_j$\\ \hline
\end{tabular}
\end{center}
by scaling both attributes with $\mathbb{K}_s$ ($t_0, ..., t_n,\: n \in \N_0$ are transitions, $i,j \in \{0,1,...,{|S|-1}\}$). Hence, a transition context is derived by replacing the attributes by the rows of $\Ks$ for the input and output state of the respective transition.

Transitions may reflect observations repeated at different time points, or they may be generated by a dynamic model. In this respect, we focus on BN (Section \ref{sec:BooleNets}). Dichotomic scaling\index{Scaling!dichotomic} will be applied, if $0$ is regarded as explicit attribute and implications with this attribute are meaningful. For instance, both low and high expression values of different genes or even of the same gene can have effects.
\begin{definition}\label{def:BoolNet}
Let $E$ be an arbitrary set of entities and $F:=\{0,1\}$ a set of fluents. A transition function\index{Transition function} $\delta\colon F^E \rightarrow F^E$ is called a \textsc{Boolean Network}.
\index{Boolean network}\glossary{name={BN}, description={Boolean network}}
\end{definition}
Given $E$ and $F$, $F^E$ is the set of all possible state descriptions (compare Definition \ref{def:stateCxt1}). Together with an injective (i.e., $s \in S$ is attribute defined) output function\index{Output function} $\gamma\colon S \rightarrow F^E$, a BN defines a transition function $\delta'\colon S \rightarrow S$ by $\delta' (s) := (\gamma^{-1} \circ \delta \circ \gamma)(s)\:(s \in S)$, hence a deterministic\index{Process!deterministic} Kripke structure.\index{Kripke structure} $\delta'$ is well-defined, if the state set is chosen large enough, i.e. if all state descriptions generated by $\delta$ correspond again to a state $s \in S$: $\forall s\in S\colon (\delta \circ \gamma)(s) \in \gamma(S)$. 

For $|E|=n$, a BN is an $n$-ary Boolean function and $F^E \cong \{0,1\}^n$. For ease of notation, $s \in S$ is identified with $\gamma(s) \in F^E$. With states $s \in \{0,1\}^n$ and coordinate functions $\delta_j \colon \{0,1\}^n \rightarrow \{0,1\},\: j=1, \dots n$, the transition function $\delta$ is given by
\[\delta(s) := 
\begin{pmatrix}
\delta_1(s)\\
\vdots\\
\delta_n(s)
\end{pmatrix} \]

A BN is representable by a directed graph where only the edges $(i,j)$ from influencing entities to an output node are drawn: $\exists\, s \in \{0,1\}^n\colon \delta_j(s \mid s_i = 0) \neq \delta_j(s \mid s_i = 1)$.\footnote{The notation $(s \mid s_i = 0)$ means: In the tuple $s=(s_0,\dots, s_{n-1}) \in \{0,1\}^n$ the entry $s_i, i \in \{0, \dots, n\}$ is replaced by 0.} The nodes are labelled by the coordinate functions $\delta_j$.

A BN generates a dynamic simulation, i.e. a process, by repeated application of $\delta'$ to a set of start states $S^{start} \subseteq S$. After each discrete time step, all component functions may be updated simultaneously or with specific time delays (synchronous or asynchronous BN).
Boolean networks may be generalised in order to include nondeterminism. Then, different output states\index{State!output} are generated from a single input state (compare \cite{Wol07}, Section \ref{sec:sporulation} and Table \ref{BoolEq}).


\section{The transitive context $\mathbb{K}_{tt}$}\label{sec:transCxt}\index{Transitive context}\glossary{name={$\Ktt$}, description={Transitive context}, sort=$Ktt$}
It appears promising to consider the transitive closure $t(R) = \bigcup_{n\in \mathbb{N}} R^n$,\glossary{name={$t(R)$},description={Transitive closure of the relation $R$}}
i.e. $(s_0,s_1) \in t(R)$ for two elements $s_0$ and $s_1$ of $S$, if
there exist $\alpha_0, \alpha_1, ..., \alpha_n \in S$ with $\alpha_0=s_0, \alpha_n=s_1$ and $(\alpha_r, \alpha_{r+1}) \in R$ for all $0\leq r < n$. That means, the state $s_1$ emerges from $s_0$ by some transition sequence of arbitrary finite length. A transitive context contains explicit information regarding reachable states.
\begin{definition}\label{def:transCxt}
The \textsc{Transitive context} $\mathbb{K}_{tt}$ of a given transition context $\Kt := (R, M\times\{in,out\},\nabla)$ is the formal context with object set $t(R)$, the transitive closure of $R$. $\nabla$ is extended accordingly: 
\[\mathbb{K}_{tt} := (t(R), M\times\{in,out\},\nabla).\]
\end{definition}

\begin{table}
\centering
\begin{tabular}{|l||c|c|c|c||c|c|c|c|}
\hline
&\begin{sideways}\texttt{driver.linux.in}\end{sideways}
&\begin{sideways}\texttt{ndiswrapper.in}\end{sideways}
&\begin{sideways}\texttt{driver.windows.in}\end{sideways}
&\begin{sideways}\texttt{connection.in}\end{sideways}
&\begin{sideways}\texttt{driver.linux.out}\end{sideways}
&\begin{sideways}\texttt{ndiswrapper.out}\end{sideways}
&\begin{sideways}\texttt{driver.windows.out}\end{sideways}
&\begin{sideways}\texttt{connection.out}\end{sideways}\\
\hline \hline
$(s_0,s_{10})$ &&&&&$\times$&&&\\ \hline 
$(s_0,s_{20})$ &&&&&$\times$&&&$\times$\\ \hline
$(s_0,s_{11})$ &&&&&&$\times$&&\\ \hline
$(s_0,s_{21})$ &&&&&&$\times$&$\times$&\\ \hline
$(s_0,s_{31})$ &&&&&&$\times$&$\times$&$\times$\\ \hline
$(s_{10},s_{20})$ &$\times$&&&&$\times$&&&$\times$\\ \hline
$(s_{11},s_{21})$ &&$\times$&&&&$\times$&$\times$&\\ \hline
...&&&&&&&&\\ \hline
\end{tabular}
\caption{Transitive context derived from the transition context of Table \ref{tab:wlanTrCxt}. Now the objects are transitions, which relate an input to an output state succeeding after an arbitrary number of time steps.}
\label{tab:wlanTransCxt}
\end{table}


\section{The temporal context $\mathbb{K}_{tmp}$}\label{sec:tmpCxt}
A transition context represents a Kripke structure by which the semantics\index{Semantics} of a temporal logic is given. It thus generates a new \textit{temporal context}: the attributes of the underlying state context are extended by the set of atomar propositions formed with the original attributes and three operators from temporal logic. Using the corresponding transitive context for definitions will show to be more convenient in important cases like deterministic processes.

The language CTL is chosen since it is quite universal and includes nondeterminism. A major restriction, however, is that CTL does not provide operators related to the past. First, we do not want to enlarge the number of operators within this basic study and therefore will even consider a subset of CTL operators explicitly. More importantly, the principal aim of this thesis are applications related to user guidance, process control or predictions. Finally, in the transition and transitive contexts, output-input implications related to the past hold. Hence, in principle past operators could be defined analogously to the following introduction of future operators. In this case, the asymmetry of time should be kept in mind, usually assumed in temporal logics and highlighted by Prior. Since there is no \agrave retrospective nondeterminism", only the deterministic versions of the subsequent definitions should be adapted. In a nondeterministic transition context, this is not quite natural and needs some technical efforts.

In order to emphasise the structural difference, in CTL, of the path and temporal operators, I use Priors (as well as the modal logic) notation $\Diamond$ for the possibility operator (corresponding to $E$) and $\Box$ for the necessity operator (corresponding to $A$). As for the temporal operators, we focus on $F$ (``eventually''), $G$ (``always'') and $\neg F$ (``never''). 

\begin{definition}\label{def:path}\glossary{name={$Seq_R$}, description={Sequences generated by a relation $R$}}
Let $R \subseteq S \times S$ be a relation. Then a \textsc{Path} is a finite sequence $\pi := (s_0,\ldots,s_n) \in S^{n+1}\: (n \in \N)$ or an infinite sequence $\pi := (s_i)_{i \in\N_0}$, so that $(s_i,s_{i+1}) \in R$ for all $0 \leq i < n$ or $i \in\N_0$, respectively. If $\pi$ is finite, it is required to be maximal: $\nexists\, s_{n+1} \in S\colon (s_n,s_{n+1})\in R$. The set of all paths $\pi$ generated by $R$ is denoted by $Seq_R$.
\end{definition}

In CTL usually infinite paths are required, hence a total relation: $\forall {s \in S}\: \exists\, {s'\in S} \colon (s,s') \in R$ \cite[Section 4.1]{Cha04}. If there are final states $s_F$\index{State!final} like for an automaton, assuming $(s_F,s_F) \in R$ makes the relation total. In applications, however, observations or predictions are often incomplete. Then we accept that the relation $R$ is not total. 

\begin{definition}\label{def:tmpCxt}\index{Temporal context}\glossary{name={$\Ktmp$}, description={Temporal context}, sort=$Ktmp$}
\glossary{name={$T$}, description={Set of temporal attributes}}
\glossary{name={$\Diamond$}, description={Possibility operator}, sort=$.Pos$}
\glossary{name={$\Box$}, description={Necessity operator}, sort=$.Nec$}
\glossary{name={$F$}, description={Eventually}}
\glossary{name={$G$}, description={Always}}
\glossary{name={$\neg F$}, description={Never}, sort=$.Never$}
Given a state context $\mathbb{K}_s=(S,M,I)$ and a relation $R \subseteq S \times S$, a \textsc{temporal context} is defined as $\mathbb{K}_{tmp} :=(S^{in}, M \cup\, T, I^{in} \cup I_T)$. The state set is restricted to the set of input states\index{State!input} $S^{in} := \{s \in S \mid \exists s' \in S \colon (s,s') \in R\}$, $I$ to $I^{in}$ correspondingly, whereas the attribute set is extended by $T := \{\Diamond Fm \mid\: m \in M\}\, \cup\, \{\Box Fm \mid\: m \in M\}\, \cup\, \{\Diamond Gm \mid\: m \in M\}\, \cup\, \{\Box Gm \mid\: m \in M\}\, \cup\, \{\Diamond \neg Fm \mid\: m \in M\}\, \cup\, \{\Box \neg Fm \mid\: m \in M\}$. Let $s \in S^{in}$ and  $Seq_R^s := \{\pi \in Seq_R \mid s_0 = s\}$. The relation $I_T$ is defined as follows:
\begin{align}
s\: I_T\:\Diamond Fm :\Leftrightarrow\: &\exists\, \pi \in Seq^s_R\: \exists\, i \in \mathbb{N}:\:  (s_i, m) \in I\\
s\: I_T\:\Box Gm :\Leftrightarrow\: &\forall\, \pi \in Seq^s_R\: \forall\, i \in \mathbb{N}:\: (s_i, m) \in I\label{eq:allAlways}\\
s\: I_T\:\Box \neg Fm :\Leftrightarrow\: &\forall\, \pi \in Seq^s_R\: \forall\, i \in \mathbb{N}:\:  (s_i, m) \notin I\label{def:nev}\\
s\: I_T\:\Box Fm :\Leftrightarrow\: &\forall\, \pi \in Seq^s_R\: 
\exists\, i \in \mathbb{N}:\:  (s_i, m) \in I\label{eq:allEv}\\
s\: I_T\:\Diamond Gm :\Leftrightarrow\: &\exists\, \pi \in Seq^s_R\: \forall\, i \in \mathbb{N}:\:  (s_i, m) \in I\\
s\: I_T\:\Diamond \neg Fm :\Leftrightarrow\: &\exists\, \pi \in Seq^s_R\: \forall\, i \in \mathbb{N}:\: (s_i, m) \notin I
\end{align}
For $B \subseteq M$, set $\Diamond FB$ := $\{\Diamond Fm_1,..., \Diamond Fm_n \mid m_1,...,m_n \in B\}$, and so forth.
\end{definition}
\glossary{name={$\mathbb{K}_T$}, description={Part of $\Ktmp$ with attributes $T$}, sort=$KT$}
$\mathbb{K}_{tmp}$ is the apposition\index{Formal context!apposition} $\mathbb{K}^{in}_s \mid \mathbb{K}_T$ of a state context $\mathbb{K}^{in}_s := (S^{in}, M, I^{in})$ and $\mathbb{K}_T:= (S^{in}, T, I_T)$. It should be kept in mind that sets of temporal attributes may relate to different paths. We understand $0 \notin \N$,\glossary{name={$\mathbb{N}$}, description={Set of integers, $0 \notin \N$}, sort=$N$} hence the definitions refer to states subsequent to $s$ and have a clear dynamic meaning. This is in accordance to Definition \ref{def:path}, which does not admit sequences $(s_0)$.


\begin{table}
\begin{center}
\begin{tabular}{|l|c|c|c|c|c|c|c|c|c|c|c|c|c|c|c|c|c|}\hline
&\begin{sideways}\texttt{driver.linux}\end{sideways} &\begin{sideways}\texttt{ndiswrapper}\end{sideways} &\begin{sideways}\texttt{driver.windows}\end{sideways} &\begin{sideways}\texttt{connection}\end{sideways}
&\begin{sideways}$\Diamond F\,$\texttt{driver.linux}\end{sideways} 
&\begin{sideways}$\Box\, G\,$\texttt{driver.linux}\end{sideways}
&\begin{sideways}$\Box\, \neg F\,$\texttt{driver.linux}\end{sideways}
&\begin{sideways}$\Box\, F\,$\texttt{driver.linux}\end{sideways}
&\begin{sideways}$\Diamond\, G\,$\texttt{driver.linux}\end{sideways}
&\begin{sideways}$\Diamond\, \neg F\,$\texttt{driver.linux}\end{sideways}
&\begin{sideways}...\end{sideways}
&\begin{sideways}$\Diamond\, F\,$\texttt{connection}\end{sideways} 
&\begin{sideways}$\Box\, G\,$\texttt{connection}\end{sideways}
&\begin{sideways}$\Box\, \neg F\,$\texttt{connection}\end{sideways}
&\begin{sideways}$\Box\, F\,$\texttt{connnection}\end{sideways}
&\begin{sideways}$\Diamond\, G\,$\texttt{connection}\end{sideways}
&\begin{sideways}$\Diamond\, \neg F\,$\texttt{connnection}\end{sideways}
\\ \hline \hline
$s_0$ &&&&&$\times$&&&&$\times$&$\times$&&$\times$&&&$\times$&&\\ \hline
$s_{10}$ &$\times$&&&&$\times$&$\times$&&$\times$&$\times$&&&$\times$&$\times$&&$\times$&$\times$&\\ \hline
$s_{11}$ &&$\times$&&&&&$\times$&&&$\times$&&$\times$&&&$\times$&&\\ \hline
$s_{20}$ &$\times$&&&$\times$&$\times$&$\times$&&$\times$&$\times$&&&$\times$&$\times$&&$\times$&$\times$&\\ \hline
$s_{21}$ &&$\times$&$\times$&&&&$\times$&&&$\times$&&$\times$&$\times$&&$\times$&$\times$&\\ \hline
$s_{31}$ &&$\times$&$\times$&$\times$&&&$\times$&&&$\times$&&$\times$&$\times$&&$\times$&$\times$&\\ \hline
\end{tabular}
\end{center} 
\caption{Temporal context (Definition \ref{def:tmpCxt}) extending the state context of Table \ref{tab:wlanStateCxt} by attributes from temporal logic. $\Diamond\colon \text{for one path}$, $\Box\colon \text{for any path}$, $F\colon \text{eventually}$, $G\colon \text{always}$, $\neg F\colon \text{never}$. More specifically, the last three predicates will be used for $\Diamond F$, $\Box\, G$ or $\Box\, \neg F$, respectively.}
\label{tab:wlanTmpCxt}
\end{table}
As we admit finite paths and do not presuppose a total relation $R$, the objects of $\Ktmp$ are restricted to input states in order to have an exact correspondence to $\Ktt$. In Chapter \ref{ch:bg}, inferences between implications of the two formal contexts will be investigated. 

The temporal attributes may also be defined by $\Diamond$, $F$ and $\neg$ only. For all $s \in S$ and all $m\in M$ holds: 
\[s\, I_T\, \square \neg F m \Leftrightarrow s\, I_T\, \neg \Diamond F m\]
Thus, one attribute can be expressed by the absence of the other. Both are needed in the case of dichotomically scaling\index{Scaling!dichotomic} $\Diamond Fm$, i.e. if one is interested in implications with the attribute $\neg \Diamond Fm$, or vice versa. The same holds for $\Box Fm \Leftrightarrow \neg \Diamond \neg Fm$ and $\Box Gm \Leftrightarrow \neg \Diamond F \neg m.$

Since a path is defined by $R$, at least partial knowledge of $\mathbb{K}_t$ is required to decide if the temporal attributes hold. $\Kt$ cannot be reconstructed unambigously from a transitive context $\mathbb{K}_{tt}$ (whereas the inverse holds, of course). However, in important cases knowledge of the transitive context is sufficient. It contains information regarding reachable states, while the paths do not have to be known. 

\begin{remark}\label{rem:semanticsKtt}
\textnormal{The first three temporal attributes may be decided more conveniently  by the transitive context $\mathbb{K}_{tt} = (t(R), M\times\{in,out\},\nabla)$ generated by $\mathbb{K}_t$. For all $s \in S^{in}$ and all $m \in M$ holds:
\begin{align*}
s\: I_T\:\Diamond Fm \Leftrightarrow\: &\exists\, (s^{in}, s^{out}) \in
t(R)\colon s= s^{in} \wedge ((s^{in}, s^{out}), m^{out}) \in \nabla\\
s\: I_T\:\Box Gm \Leftrightarrow\: &\forall \,(s^{in}, s^{out}) \in t(R) \text{ with } s= s^{in}\colon ((s^{in}, s^{out}), m^{out}) \in \nabla\\
s\: I_T\: \Box \neg F m \Leftrightarrow\: &\forall\, (s^{in}, s^{out}) \in t(R)\text{ with } s= s^{in}\colon  ((s^{in}, s^{out}), m^{out}) \notin \nabla.
\end{align*}
As $s \in S^{in}$ is presupposed, $s\, I_T\,\Box Gm$ implies $s\, I_T\,\Diamond Fm$.}
\end{remark}
In the following, \agrave eventually'',\index{Eventually}\glossary{name={$\operatorname{ev}$}, description={Eventually, $\Diamond F$}, sort=$Ev$} \agrave always''\index{Always}\glossary{name={$\operatorname{alw}$}, description={Always, $\Box\, G$}, sort=$Alw$} and \agrave never''\index{Never}\glossary{name={$\operatorname{nev}$}, description={Never, $\Box\, \neg F$}, sort=$Nev$} will have the specific semantics given by the three properties, which only makes a difference for nondeterminism. In the deterministic case,\index{Process!deterministic} a single path exists starting from an arbitrary state $s_0 \in S^{in}$. Therefore, also the other attributes are definable by $\mathbb{K}_{tt}$, with $\Box Fm = \Diamond Fm, \Diamond Gm = \Box Gm$ and $\Diamond \neg Fm = \Box \neg Fm$.

$\Box G \neg m$ expresses a safety property\index{Safety}, $\Box Fm$ a liveness property\index{Liveness} (see Section \ref{sec:CTL}). In the following applications we focus on the properties definable by the transitive context, thus on $\Diamond Fm$ instead of $\Box Fm$. The property is adequate to stochastic data like in biology and interesting as negation of safety.

The remaining operators from CTL -- $X$ (\agrave next'')\index{Next} and $U$ (\agrave until'')\index{Until} -- are definable as follows, for $m\in M$, $B \subseteq M$\label{nextOp} and for all $s \in S^{in}$:
\begin{align*}
s\: I_T\:\Diamond\, Xm :\Leftrightarrow\: &\exists\, \pi \in Seq^s_R\colon  (s_1, m) \in I\\ \glossary{name={$X$}, description={Next}}
\Leftrightarrow\: &\exists\, (s_{in}, s_{out}) \in R\colon (s_{in} = s) \wedge  (s^{in}, s^{out})\, \nabla \, m^{out}\\
s\: I_T\:\Diamond\, m\, U B :\Leftrightarrow\: &\exists\, \pi \in Seq^s_R\: \exists\, j \in \mathbb{N}_0\colon (s_j \in B^I) \wedge (\forall\, 0 \leq i < j\colon (s_i, m) \in I).\glossary{name={$U$}, description={Until}}
\end{align*}
With the condition $i<j$, $m\, U B$ is satisfied if the start state of $\pi$ has all attributes in $B$, i.e. if those attributes are contained in the object intent $s'$. The definitions for the necessity operator $\Box$ are analogous.

All CTL operators may be constructed with the aid of $\Diamond$, $X$, $U$ and $\neg$, e.g. $\Diamond\, B = \Diamond\, \top\,U B$ (compare \cite[p. 2]{Art07}). However, the two operators will not be examined below. $X$ is expressed by a transition context and $U$ enlarges the set of temporal attributes $T$ excessively, since it connects two attribute sets $\{m\}, B \subseteq M$.

\chapter{Using attribute exploration of the defined formal contexts}\label{ch:useAttrEx}
\chaptermark{Using attribute exploration}

Attribute exploration of discrete temporal transitions aims at unfolding the dynamics of a process by investigating its rules. In the last chapter, four data structures were defined: A state context $\Ks$ describes by observable attributes the different states occuring during a temporal development. It is enlarged to a context $\Ktmp$ by attributes constructed with operators from temporal logic. Their semantics\index{Semantics} is given by a transition context $\Kt$. For deterministic processes the related transitive context $\Ktt$ is sufficient. It expresses knowledge about reachable states and their attributes.

In order to give an intuition of the methodic possibilities, I start with the wireless card example of Section \ref{sec:runEx}. In this example, the proposed implications are decidable by a human basically experienced with a Linux operating system. In Section \ref{sec:knowData}, more structured interactions between theory and observations will be presented. One purpose is to relate observations to existing knowledge by a scientist. Inversely, model predictions can be validated empirically. I propose a procedure which will be applied in the main biological example of Chapter \ref{ch:geneRegNets}. In Chapter \ref{ch:bSubtilis}, the dynamics of a biological model from literature are analysed by computing the stem base of the related transitive context. In Section \ref{sec:appAutomata}, the universal applicability of our approach to three important classes of process models is summarised in a theorem: LTSA, Kripke structures and automata are representable by a family of transition contexts.

Given a set of observations related to the input-output-behaviour, the attribute exploration algorithm asks if those observations are complete. \texttt{Next Closure}\index{Attribute exploration!Next Closure algorithm}, the core of the algorithm, finds the next pseudo-closed set with respect to the \textit{lectic order}\index{Lectic order} of the attributes \cite[p. 66f., 85]{GW99}; regarding the supplementary criterion of set inclusion it is the smallest one. If the corresponding implication is rejected, the counterexample represents a minimal missing observation in this sense. Furthermore, the number of newly introduced observations is generally small, in correspondence to the minimality of the stem base. 

\section{Example}\label{sec:exEx}
The stem base of the state context (Table \ref{tab:wlanStateCxt}) is simple:
\begin{verbatim}
1 < 2 > driver.windows ==> ndiswrapper;
2 < 1 > ndiswrapper, connection ==> driver.windows;
3 < 0 > driver.linux, ndiswrapper ==> driver.windows, connection;
\end{verbatim} 
The first two implications express that a Windows driver and the wrapper module have to be installed together, implicating temporal priority of the latter. For an open source driver, no further statement is possible. The attribute \texttt{driver.linux} may only occur together with \texttt{connection}, but there is no dependency in any direction. The third implication signifies the exclusion of the Windows and Linux procedures: all remaining attributes occur in the conclusion, and the extent of the premise as well as of the complete attribute set is empty. For implications, in contrast to less strict association rules, the cardinality of this extent is the support\index{Support} of the rule (see p. \pageref{support}). It is indicated in brackets \texttt{<$\,\cdot$\,>}. This type of implications is mostly noted with the $\bot$ symbol: \texttt{driver.linux, ndiswrapper ==> $\bot$}.

\glossary{name={$<\,\cdot\,>$}, description={Support of an implication}}
\glossary{name={$\bot$}, description={Complete attribute set}, sort=$.Complete$}The transition context describes immediately succeeding installation steps $(s^{in}, s^{out}) \in R$ by their respective attributes, i.e. by single performed procedures. Accordingly, static implications of $\Ks$ are found again between attributes referring to $s^{in}$ or $s^{out}$ solely. New information regarding dynamics is obtained by \agrave mixed" implications with input and output attributes. The exploration of the transition context (Table \ref{tab:wlanTrCxt}) leads to two counterexamples:
\begin{itemize}
 \item \texttt{driver.windows.out $\rightarrow$ ndiswrapper.in, ndiswrapper.out}: The proposed implication raises the question if the order can also be changed: A user might download a Windows driver first and then s/he learns that s/he needs the ndiswrapper module. This procedure is possible, since ndiswrapper copies the driver files into its own directory anyway. Therefore, the first counterexample in Table \ref{tab:wlanTrCxt} is introduced.
 \item It is important to correct the later implication \texttt{ndiswrapper.out, driver.windows.out $\rightarrow$ ndiswrapper.in} and to add counterexample 2 with \texttt{driver.windows.in} instead of \texttt{ndiswrapper.in}.
\end{itemize}

The stem base of the transition context is the following:
\begin{verbatim}
1 < 3 > ndiswrapper.in ==> ndiswrapper.out, driver.windows.out;
2 < 3 > driver.windows.in ==> ndiswrapper.out, driver.windows.out;
3 < 2 > driver.linux.out, connection.out ==> driver.linux.in;
4 < 2 > ndiswrapper.out, connection.out ==> ndiswrapper.in,
driver.windows.in, driver.windows.out;
5 < 2 > driver.windows.out, connection.out ==> ndiswrapper.in,
driver.windows.in, ndiswrapper.out;
6 < 2 > driver.linux.in ==> driver.linux.out, connection.out;
7 < 2 > connection.in ==> connection.out;
8 < 2 > ndiswrapper.in, driver.windows.in, ndiswrapper.out,
driver.windows.out ==> connection.out;
9 < 0 > driver.linux.out, ndiswrapper.out ==> driver.linux.in,
ndiswrapper.in, driver.windows.in, connection.in, driver.windows.out,
connection.out;
10 < 0 > driver.linux.out, driver.windows.out ==> driver.linux.in, 
ndiswrapper.in, driver.windows.in, connection.in, ndiswrapper.out,
connection.out;
\end{verbatim} 
The first two implications describe the alternative paths of the Windows driver installation, the third to fifth preconditions for entering the connection data, combined in the latter two with an (output) state implication. Implication 6 and 8 define the last step in the Linux and Windows driver installation process, respectively. Implication 7 marks the state with \texttt{connection} attribute as the final state, or a \textit{steady state}\index{State!steady} in engineering, physical or biological systems. The last two implications again express the exclusion of the alternative paths for a Linux and Windows driver.

A state $s_{12}$ with single attribute \texttt{driver.windows} could be found as counterexample during the exploration of the state context already (Table \ref{tab:wlanStateCxt}), if one considers the complete dynamics. Then, the counterexamples of the transition context will be introduced before its exploration as the transitions $(s_0,s_{12})$ and $(s_{12}, s_{21})$ with the new state. Consequently, no other counterexamples need to be added during the exploration of $\Kt$. With $s_{12}$, the stem base of $\Ks$ is:
\begin{verbatim}
1 < 1 > ndiswrapper, connection ==> driver.windows;
2 < 1 > driver.windows, connection ==> ndiswrapper;
3 < 0 > driver.linux, ndiswrapper ==> driver.windows, connection;
4 < 0 > driver.linux, driver.windows ==> ndiswrapper, connection;\end{verbatim}
\label{ex:pseudoclosed}Implication 4 is new, since now \texttt{\{driver.windows\}} is closed as object intent: only the driver may be installed. Then, all subsets of the premise (including $\emptyset$)\footnote{All objects (states) have the empty attribute set in common, but no other attribute; there is no implication $\emptyset \rightarrow ...$} are closed, and the condition of a pseudo-intent\index{Pseudo-intent} is trivially fulfilled. In the first version of the stem base, however, \texttt{\{driver.linux, driver.windows\}} does not contain the closure of the pseudo-intent \texttt{\{driver.windows\}}, indicated by implication 1. Hence, the attribute set is not pseudo-closed and there is no corresponding implication in the previous stem base.

If the enlarged stem base of $\Ks$ is entered as background knowledge prior to an exploration of $\Kt$ (compare Section \ref{sec:hierarchy}), pure input or output implications can be decided automatically, like $3 \text{ (in }\mathbb{K}_s) \models 9 \text{ (in }\mathbb{K}_t)$, $4\text{ (in }\mathbb{K}_s) \models 10\text{ (in }\mathbb{K}_t)$, or \texttt{ndiswrapper.out, connection.out ==> driver.windows.out} (output part of implication 4).

Without the implications derivable from those of the (corrected) state context, the stem base of the transitive context (Table \ref{tab:wlanTransCxt}) is:\label{baseTransCxt}
\begin{verbatim}
1 < 4 > ndiswrapper.in ==> ndiswrapper.out, driver.windows.out;
2 < 4 > driver.windows.in ==> ndiswrapper.out, driver.windows.out;
3 < 2 > driver.linux.in ==> driver.linux.out, connection.out;
4 < 2 > connection.in ==> connection.out;
5 < 2 > ndiswrapper.in, driver.windows.in, ndiswrapper.out,
driver.windows.out ==> connection.out;
6 < 1 > connection.in, driver.linux.out, connection.out
==> driver.linux.in;
7 < 1 > connection.in, ndiswrapper.out, driver.windows.out,
connection.out ==> ndiswrapper.in, driver.windows.in;
\end{verbatim}
It expresses the following new facts:
\begin{itemize}
\item Implications 1 to 4 are identical to the implications 1, 2, 6 and 7 of $\Kt$, but have different semantics: The attributes of a state remain those of all subsequent states, which signifies that a success cannot be destroyed by a further action. This seems to be realistic with the assumption of a sufficiently experienced user. For the same input and output attribute this meaning is already implicit in the transition context, e.g. \texttt{driver.linux.in} $\rightarrow$ \texttt{driver.linux.out}.
\item Implication 6 connects input and output states in a more complicated manner (\texttt{connection.out} follows from implication 4): If after an input state with attribute \texttt{connection} the Linux driver is installed, it must have been installed already at this input state. Thus, the temporal order of the two installation steps is fixed. Implication 6 replaces 3 in $\Kt$. $(s_{10}, s_{20})^\prime$ = \{\texttt{driver.linux.in, driver.linux.out, connection.out}\} is no more generated by the former premise \{\texttt{driver.linux.out, connection.out}\}, but this attribute set is closed as object intent of the new transition $(s_0, s_{20})$ with the start state $s_0$ where $s_0'=\emptyset$: 
\[(s_0, s_{20})^\prime = \text{\{\texttt{driver.linux.out, connection.out}\}} \varsubsetneq (s_{10}, s_{20})^\prime \varsubsetneq (s_{20}, s_{20})^\prime\]
With the supplementary condition \texttt{connection.in} implication 6 refers to $(s_{20}, s_{20})$ only and describes a property of the steady state\index{State!steady} $s_{20}$. Implication 7 is analogous for the Windows driver path.
\end{itemize}

The first of 18 implications in the stem base of the temporal context are, with \texttt{ev(entually)}$:= \Diamond F$, \texttt{alw(ays)}$:= \Box\, G$ and \texttt{nev(er)}$:= \Box\, \neg F$:
\begin{verbatim}
1 < 14 > { } ==> ev(connection);
2 < 12 > ev(driver.windows) ==> ev(ndiswrapper);
3 < 12 > ev(ndiswrapper) ==> ev(driver.windows);
4 < 2 > connection ==> alw(connection);
5 < 4 > driver.windows ==> ev(ndiswrapper), ev(driver.windows),
alw(ndiswrapper), alw(driver.windows), nev(driver.linux);
6 < 4 > ndiswrapper ==> ev(ndiswrapper), ev(driver.windows),
alw(ndiswrapper), alw(driver.windows), nev(driver.linux);
...
\end{verbatim}
\label{stemBaseKtmp}
Implications 2 and 3 indicate that the two steps \texttt{driver.windows} and \texttt{ndiswrapper} have to be taken on the Windows path in order to establish the connection, similarly implications 5 and 6. Those mean again: Once ndiswrapper or a windows driver is installed, it remains installed; also the Windows and Linux paths are exclusive. The first and fourth implication express a hopeful view: An internet connection will be established some day (and by some way) -- and then it will hold forever... 

Since users sometimes make bad experiences with the compatibility of hardware and software or with their internet provider, more realistic counterexamples could be entered into the context. However, they only make sense together with the introduction of new attributes, e.g. \texttt{line} for the existence or breakdown of the external internet connection. In this case, the previously explored implications cannot be used further. After all, even if several implications are removed, only a defined part of the stem base has to be computed again \cite[p.14]{Baa09a}. A better solution is to interpret \texttt{connection} as only entering the personal connection data and to define a new attribute \texttt{success}, corresponding to the final state\index{State!final} of an automaton. Then, no new transitions $\notin R$ with \texttt{connection.in} but not \texttt{connection.out} are introduced, contradicting the implication \texttt{connection.in $\rightarrow$ connection.out} in $\Ktt$, which corresponds to implication 4 in the stem base of $\Ktmp$.

In order to refine the process description, new attributes may be introduced, for example device recognition with \texttt{sudo lshw -C network} at the beginning, resulting in the alternative actions \texttt{device.on.out} (turning the wireless device on by a hardware switch or even by Windows), \texttt{ndiswrapper.out}, \texttt{driver.linux.out} or \texttt{connection.out}.\footnote{\texttt{https://help.ubuntu.com/9.04/internet/C/troubleshooting-wireless.html\# troubleshooting-wireless-device}} Furthermore, existing attributes may be splitted, for instance in \texttt{ndiswrapper.graphical} and \texttt{ndiswrapper.commandline}.\footnote{\texttt{https://help.ubuntu.com/community/WifiDocs/Driver/Ndiswrapper\#Installing Windows driver}} 


In general, newly added attributes
should be chosen carefully to avoid a change of the meaning of the first attributes. Then, the implications of all contexts hold in the enlarged contexts. For the reader familiar with FCA, the following remark expresses this demand more formally.
\begin{remark}\label{rem:subsets}
\textnormal{For a transition context $\Kt := (R, (M_1 \cup M_2) \times \{in,out\}, \nabla)$ (or a transitive context $\Ktt$), implications of the original attribute set $M_1$ are preserved, if the set of transitions $R$ is not changed and the restrictions of the incidence relation $\nabla$ of the enlarged context $\Kt$ to $R \times (M_1  \times \{in,out\})$ and $R \times (M_2 \times \{in,out\})$ yield the incidence relations of the partial contexts $\Kt^1 := (R, M_1  \times \{in,out\}, \nabla_1)$ and $\Kt^2 := (R, M_2 \times \{in,out\}, \nabla_2)$. Then, a supremum-preserving order embedding of the lattice $\mathfrak{B}(\Kt)$ into the direct product of $\mathfrak{B} (\Kt^1)$ and $\mathfrak{B} (\Kt^2)$ is defined. 
By reason of the supremum condition, an implication of $\Kt^1$ (and $\Kt^2$) is also an implication of $\Kt$, since the conclusion states that the respective attribute concepts and all infima are greater or equal to the concept generated by the premise. 
The concept lattice of $\Kt$ can be visualised by a \textit{nested line diagram}\index{Nested line diagram}, where copies of $\Kt^2$ and the embedding are represented in each node of $\Kt^1$, as a kind of zooming into more differentiated object descriptions. \cite[p. 77ff. and Theorem~7]{GW99}}
\end{remark}
An order embedding is also given, if an original context $\Kt^1:= (R', M_1  \times \{in,out\}, \nabla_1),\: R' \subseteq R$ is enlarged by all transitions in $R \setminus R'$ without changing the set of object intents, i.e. if clarifying the enlarged context results in $\Kt^1$. Then, new attributes from $\Kt^2$ are introduced that discriminate the supplementary from the previous transitions.

\section{Integration of knowledge and data}\label{sec:knowData}
The defined mathematical structures may be used in various ways. In Chapter \ref{ch:bSubtilis}, a transitive context is generated from a BN. Then, the computed stem base is searched for implications which make the temporal behaviour explicit and give new insight. Furthermore, experimental time series can be evaluated by comparison with existing knowledge, i.e. implications are generalised or rejected  supposing outliers or by reason of special conditions. Inversely, for the ECM application in Chapter \ref{ch:geneRegNets} we developed a procedure starting from knowledge:
\begin{enumerate}\label{protocol}
\item\label{item:Kobs} Discretise a set of time series of (gene expression) measurements and
transform it to an observed transition context $\mathbb{K}_t^{obs}$.\glossary{name={$\mathbb{K}_t^{obs}$}, description={Observed transition context}, sort=$Ktobs$}
\item\label{item:BoolNet} For a set of interesting genes translate interactions from biological
literature and databases into a Boolean network.\index{Boolean network} This step could be supported by text mining software.
\item\label{item:K} Construct the transition context $\mathbb{K}_t$ by a simulation
 starting from a set of  states,
 e.g. the initial states of $\mathbb{K}_t^{obs}$ or all states (for small networks).
\item\label{item:K_t} Derive the respective transitive contexts $\mathbb{K}_{tt}$ and
$\mathbb{K}^{obs}_{tt}$.
\item \label{item:compObs}Perform attribute exploration of $\mathbb{K}_{tt}$. Decide about an
implication $A\rightarrow B$, $A,B \subseteq M$, by checking its validity in $\mathbb{K}_{tt}^{obs}$
and/or by searching for supplementary knowledge. Possibly provide a counterexample from
$\mathbb{K}^{obs}_{tt}$.
\item\label{item:query} Answer logical queries from the modified context $\mathbb{K}_{tt}$ and from its stem base.
\end{enumerate}
In step \ref{item:compObs} automatic decision criteria could be thresholds of \textit{support} $q= |(A
\cup B)'|$\index{Support}\label{support} and \textit{confidence} $p= \frac{|(A \cup B)'|}{|A'|}$\index{Confidence} for an association rule\index{Association rule} in $\mathbb{K}^{obs}_{tt}$ (which coincides with the respective implication for $p=1$). A weak criterion is to reject only implications with support 0 (but if no
object in $\mathbb{K}^{obs}_{tt}$ has all attributes from A, the implication is not violated). In
\cite{Wol07} and the main part of Chapter \ref{ch:geneRegNets}, a strong criterion is applied: implications of $\mathbb{K}_{tt}$ have to be valid also in the observed context ($p=1$). This is equivalent to an exploration of the union of the two contexts. Moreover, in Section \ref{sec:expertExploration} attribute exploration is performed where the human expert can use support and confidence as two among several decision criteria.

I implemented the steps \ref{item:Kobs}, \ref{item:K}
and \ref{item:K_t} in the scripting language \texttt{R} designed for data analysis and statistical purposes \cite{R11}. For step \ref{item:compObs}, I used the Java tool 
\texttt{Concept Explorer} \cite{ConExp}. The output was translated with \texttt{R} into a \texttt{Prolog}\index{Prolog} knowledge base, which can be queried according to step \ref{item:query}.

\section{Transition contexts and automata}\label{sec:appAutomata}
As introduced in Section \ref{sec:trCxt}, a transition context represents a Kripke structure and a single context from the action context family of an LTSA. In Chapter \ref{ch:logic} the correspondence of an LTSA and an automaton was mentioned. The following theorem summarises these relationships:
\begin{theorem}\label{theorem:coAlg}
With sets $S$ (states), $M$ (attributes) and $A$ (actions), let $\mathcal{A} := (S, \alpha_S)$ be a universal coalgebra of type $\Omega$, where
\[\alpha_S\colon S \rightarrow \Omega(S) := \mathfrak{P}(M) \times (\mathfrak{P}(S))^A.\] 
The component maps of $\alpha_S$ are denoted by $\gamma\colon S \rightarrow \mathfrak{P}(M)$ and $\delta\colon S \rightarrow (\mathfrak{P}(S))^A$, thus for all $s \in S$
\[\alpha_S(s) = (\gamma(s), \delta(s)),\]
where $\delta(s)\colon A \rightarrow \mathfrak{P}(S)$.

Then $\mathcal{A}$ is representable by a family of transition contexts $(\Kt^a)_{a \in A}$, if and only if for all $s \in S$ one of the following conditions holds:
\begin{enumerate}
\item $\exists\, a \in A\colon \delta(s)(a) \neq \emptyset$.
\item $\exists\, a \in A\:\exists\, s' \in S\colon s \in \delta(s')(a).$
\end{enumerate}
\end{theorem}
\begin{proof}
According to Proposition \ref{prop:LTSAAutomaton}, an LTSA is a universal coalgebra of type $\Omega$, and every such coalgebra can be represented as an appropriate LTSA. 

Its transition relation $R \subseteq S \times A \times S$ (Definition \ref{def:LTSA}) corresponds to the component map $\delta\colon S \rightarrow (\mathfrak{P}(S))^A$. It is representable by a family $(R_a)_{a\in A}$, where $R_a := \{(x,y) \mid (x,a,y) \in R\} \subseteq S \times S$. Each relation $R_a$ is the object set of a transition context (or action context) $\Kt^a$, so $R$ can be represented as a family $(\Kt^a)_{a\in A}$.

The supplementary conditions are equivalent to
\[\forall s \in S\: \exists\, a \in A\: \exists\, s' \in S: (s,a,s') \in R \vee (s',a,s) \in R.\]
Hence, each state is either an input or an output state for some transition context from $(\Kt^a)_{a \in A}$. Then the state context, i.e. the map $\gamma\colon S \rightarrow \mathfrak{P}(M)$ can be reconstructed. The condition is necessary, since transition contexts only contain information regarding states occuring in a relational pair.
\end{proof}

Besides LTSAs, Kripke structures -- with $A:=\{\text{update}\}$ -- are such types of coalgebras (Section \ref{sec:Kripke}). By Proposition \ref{prop:LTSAAutomaton} an automaton is identifiable with an LTSA, hence it can also be represented by a family of transition contexts.

It is more convenient to have explicit information regarding $\gamma\colon S \rightarrow \mathfrak{P}(M)$ by a state context $\Ks$. In this case, $\mathcal{A}$ will be represented by a pair $(\Ks, (\Kt)_{a \in A})$ and the supplementary conditions are not required. 

In principle, formal contexts may have an infinite attribute set $M$, like LTSAs or automata (set of data $D$). However, this has restricted practical bearing, since for example there are termination problems of attribute exploration. Nevertheless, for attributes defined as DL concepts, S. Rudolph \cite{Rud06} as well as F. Baader and F. Distel \cite{Baa08} proved by different approaches that the algorithm terminates, if a finite model\index{Model} exists (compare p. \pageref{infAttr}). It may be given as a transition context.

It could be theoretically interesting to define a category of transition contexts and to investigate if it is equivalent to a category of Kripke structures. If the morphisms are chosen accordingly, this should be obvious. However, there are some technical problems, since within a Kripke structure the state context is given explicitly, whereas states may occur only in the first or second component of the transition relation $R$, as indicated in the proof of the theorem.

In order to show categorial equivalence, it is sufficient to prove that there is a full, faithful and essentially surjective functor in one direction. Regarding automata, surely only one or two of these properties are given. Constructing the respective functors could permit to transfer mathematical results from one category (or theory) to the other. Yet, for our purposes it is adequate to state the above structural paralleles.

\begin{figure}
 \centering
 \includegraphics[width=15cm]{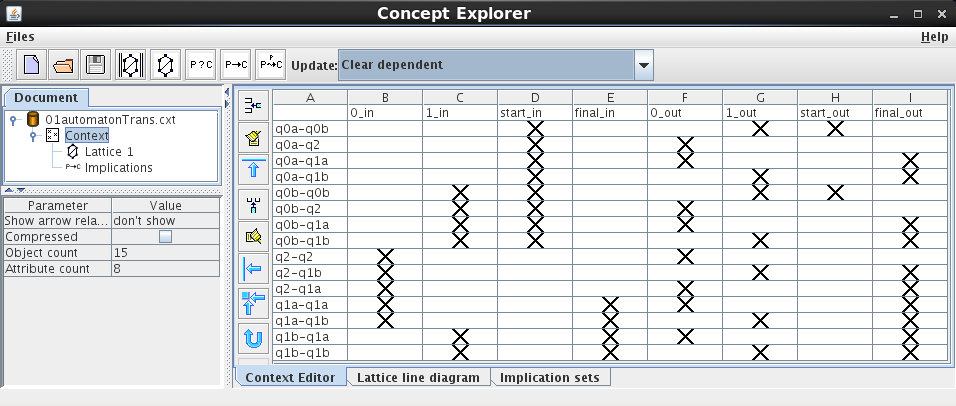}
 \caption{Transitive context of the automaton example \ref{ex:01automaton} shown in Figure \ref{fig:01automaton} (\texttt{Concept Explorer} screenshot).}
 \label{fig:01automatonTrans}
\end{figure}
Finally an example is given how attribute exploration might be used for the analysis of an automaton. The stem base of the transitive context (Figure \ref{fig:01automatonTrans}) according to example \ref{ex:01automaton} is rather simple:
\begin{small}\begin{verbatim}
1 < 4 > final_in ==> final_out;
2 < 2 > start_out ==> start_in 1_out;
3 < 2 > 0_in 1_out ==> final_out;
4 < 0 > start_in 1_out start_out final_out ==> 0_in 1_in final_in 0_out;
5 < 0 > 0_in 1_in ==> start_in final_in 0_out 1_out start_out final_out;
6 < 0 > 0_in start_in ==> 1_in final_in 0_out 1_out start_out final_out;
7 < 0 > 0_out 1_out ==> 0_in 1_in start_in final_in start_out final_out;
8 < 0 > start_in final_in final_out ==> 0_in 1_in 0_out 1_out start_out;
\end{verbatim}\end{small}
Implication 1 signifies that a string is recognised as soon as the first 01-substring is entered, 3 expresses the final condition, 2 that an input 1 is necessary to remain in a start state. In implications 4-8 all attributes occur and the support is 0. They express that states must not be start and final at the same time, 0 and 1 are exclusive, and 0 is not possible in the start state, but an empty input or 1. These implications could have been entered as background knowledge, but an attribute exploration helps not to forget one of the implications, e.g. 6. For this example they are rather obvious for a human, but in a computer program for automatic reasoning every piece of information has to be made explicit.

\chapter{Inference rules for the integration of background knowledge}\label{ch:bg}
\chaptermark{Inference rules for background knowledge}
This chapter presents a method to derive inference rules integrating already acquired knowledge into the successive exploration of the four formal contexts. This method uses attribute exploration again on a higher level by means of a test context defined in Section \ref{sec:ruleEx}.

A reasonable order is to first explore a state context $\Ks$, then the related transitive, transition and temporal contexts $\Ktt$, $\Kt$ and $\Ktmp$ (Section \ref{sec:hierarchy}). We concentrate on inference rules for implications of $\Ktt$, which defines the semantics of three temporal attributes, or of all in the case of a deterministic process. Nevertheless, the stem base of $\Ktmp$ is not completely derivable from the stem base of $\Ktt$; a subsequent exploration of $\Ktmp$ creates new information (Section \ref{sec:incompleteDet}). Inference rules are proven for two important classes of implications:
\begin{itemize}
\item Section \ref{sec:rules3Attr}: The premises are one or two \agrave homogeneous" attribute sets, e.g. $X^{in} \wedge\, Y^{out}$ (in $\Ktt$) or $\ev{Y}$ (in $\Ktmp$), $X,Y \subseteq M$. Due to computational limits, the test context had to be constructed with $|M|=3$, but most of the discovered inference rules are general.    
\item Section \ref{sec:rulesBooleAttr}: The many-valued attributes $e \in E$ take Boolean values $f \in F=\{0,1\}$ (compare Definition \ref{def:stateCxt1}). The implications have one homogeneous attribute set as premise and conclusion, respectively. 
\end{itemize}

\section{A hierarchy of formal contexts}\label{sec:hierarchy}
The four defined formal contexts are closely related. $\mathbb{K}_t$ is a subcontext of $\mathbb{K}_{tt}$, since transitions according to the transitive closure of $R$ are added. Here subcontext means that the attribute set is equal and the object set is a subset of the object set of the larger context (of course equality is possible). The same holds for the sets of intents. $\mathbb{K}_t^{in/out}$ and $\mathbb{K}_{tt}^{in/out}$, the input and output parts of a transition and a transitive context, are equal after object clarifying and identification of a transition $(s^{in}, s^{out})$ with the state $s^{in}$ and $s^{out}$, respectively. Modified that way and after identification of an attribute $a^{in}/a^{out}$ with $a$, $\mathbb{K}_t^{in/out}$ and $\mathbb{K}_{tt}^{in/out}$ are subcontexts of the scale $\mathbb{K}_s$.  

$\operatorname{Imp} \mathbb{K}$ denotes the set of attribute implications that hold in a context $\mathbb{K}$.\glossary{name={$\operatorname{Imp}$}, description={Implications valid in a formal context (for its set of intents)},sort=$Imp$} Since implications valid for a formal context are less restricted by a smaller set of intents, the inverse subset relations hold for the respective implication sets:
\begin{equation}
\operatorname{Imp} \mathbb{K}_s \subseteq \operatorname{Imp} \mathbb{K}_{tt}^{in/out} = \operatorname{Imp}\mathbb{K}_t^{in/out} 
\end{equation}
\begin{equation}\label{eq:impKttKt}
\operatorname{Imp} \mathbb{K}_{tt} \subseteq \operatorname{Imp} \mathbb{K}_t
\end{equation}
However, implications between sets of input and output attributes of a transition and a transitive context have different semantics. They relate to single transitions and sequences of transitions, respectively. 

The subset relations can be proven more formally before the background of model theory.

\subsection{Excursus: Model theory and Galois connections}
The two derivation operators $'$ of a formal context (Definiton \ref{def:context}) define maps $\sigma\colon \mathfrak{P}(G) \rightarrow \mathfrak{P}(M)$ and $\tau\colon \mathfrak{P}(M) \rightarrow \mathfrak{P}(G)$ between power sets ordered by set inclusion. They constitute a \textit{Galois connection}\index{Galois connection} between $G$ and $M$.
\begin{definition}\cite[Definition 3.1.3]{Ihr03}\label{def:galois}
A \textsc{Galois connection} between the sets $G$ and $M$ is a pair of maps 
\[\sigma\colon \mathfrak{P}(G) \rightarrow \mathfrak{P}(M) \qquad\text{and}\qquad \tau\colon \mathfrak{P}(M) \rightarrow \mathfrak{P}(G),\]
so that for all $X,X_1,X_2\subseteq G$ and all $Y,Y_1,Y_2 \subseteq M$ the subsequent conditions hold:
\begin{enumerate}
 \item $X_1 \subseteq X_2 \Rightarrow \sigma(X_2) \subseteq \sigma(X_1)$
 \item $Y_1 \subseteq Y_2 \Rightarrow \tau(Y_2) \subseteq \tau(Y_1)$
 \item $X \subseteq \tau \sigma(X) \text{ and } Y \subseteq \sigma \tau(Y)$.
\end{enumerate}
\end{definition}

Now this definition will be applied to the formal context
\[(\mathfrak{P}(M), \mathfrak{P}(M)\times \mathfrak{P}(M), \models).\]
The attributes in $\mathfrak{P}(M)\times \mathfrak{P}(M)$ are interpreted as implications $\alpha\colon A \rightarrow B,\: A,B \subseteq M$. The incidence means that $X \in \mathfrak{P}(M)$ respects the implication $\alpha$, i.e.\glossary{name={$\models$}, description={Semantic inference}, sort=$.Models$}\index{Semantic inference} $X \models \alpha \Leftrightarrow A \not\subseteq X \text{ or } B \subseteq X$ \cite[p. 110-112]{Gan05}. $X$ is called a \textit{model}\index{Model} of $\alpha$. 

For subsets $\mathcal{F} \subseteq \mathfrak{P}(M)\times \mathfrak{P}(M)$ and $\mathcal{S} \subseteq \mathfrak{P}(M)$ let
\[\operatorname{Mod} \mathcal{F} := \{X \subseteq M \mid X \models \alpha \text{ for all } \alpha \in \mathcal{F}\}
\glossary{name={$\operatorname{Mod}$}, description={Models of a set of implications},sort=$Mod$}\]
be the set of models of $\mathcal{F}$ and
\[\operatorname{Imp} \mathcal{S} := \{\alpha \in \mathfrak{P}(M)\times \mathfrak{P}(M) \mid X \models \alpha \text{ for all } X \in \mathcal{S}\}.
\glossary{name={$\operatorname{Imp}$}, description={Implications valid in a formal context (for its set of intents)},sort=$Imp$}\]
The operators $\operatorname{Imp}$ and $\operatorname{Mod}$ define a Galois connection between subsets of $M$ and implications over $M$, i.e. between $\mathfrak{P}(M)$ and $\mathfrak{P}(M)\times \mathfrak{P}(M)$ (compare \cite{Gan99}). 

For every Galois connection the compositions $\tau \sigma$ and $\sigma \tau$ define a closure operator\index{Closure operator} on $G$ and $M$, respectively (p. \pageref{def:concept}). Here, $G:= \mathfrak{P}(M)$, $M:= \mathfrak{P}(M) \times \mathfrak{P}(M)$, $\sigma:= \operatorname{Imp}$ and $\tau:=\operatorname{Mod}$. The extents and intents of $(\mathfrak{P}(M), \mathfrak{P}(M)\times \mathfrak{P}(M), \models)$ are Galois closed sets. This closure means that the extents are precisely the closure systems\index{Closure system} on $M$ (they are closed under arbitrary intersections), the intents are the respective implicational theories \cite[p. 110]{Gan05}. Thus, $\operatorname{Mod} \operatorname{Imp}\,\mathcal{S}$ is a closure system on $M$. It can be represented as the system of intents of a formal context $\mathbb{K}=(G,M,I)$. Then, $\operatorname{Imp} \mathcal{S}$ is the set of implications holding in $\mathbb{K}$, or the implicational theory generated by the stem base of $\mathbb{K}$.

Returning to $\Ks$ and $\Kt^{in/out}$ identified with a subcontext of $\Ks$, let $\mathcal{S}_1$ be the set of object intents of $\Ks$ and $\mathcal{S}_2$ the set of object intents of $\Kt^{in/out}$. Then $\operatorname{Imp} \mathcal{S}_1$ is the implicational theory of $\Ks$, $\operatorname{Imp} \mathcal{S}_2$ that of $\Kt^{in/out}$. $\operatorname{Mod} \operatorname{Imp} \mathcal{S}_i$ are the respective systems of all intents. Since $\mathcal{S}_2 \subseteq \mathcal{S}_1$, we obtain by virtue of the Galois connection between $\mathfrak{P}(M)$ and $\mathfrak{P}(M) \times \mathfrak{P}(M)$ (property 1.):
\[\operatorname{Imp}\Ks = \operatorname{Imp}\mathcal{S}_1 \subseteq \operatorname{Imp}\mathcal{S}_2 = \operatorname{Imp}\Kt^{in/out}.\] 
The inclusion $\eqref{eq:impKttKt}$ is proven analogously.

\subsection[Background knowledge for the exploration of $\Ks, \Ktt, \Kt$ and $\Ktmp$]{Background knowledge for the successive exploration of $\Ks, \Ktt, \Kt$ and $\Ktmp$}
By entering background knowledge (not necessarily implications) prior to an attribute exploration, the algorithm may be shortened considerably \cite{Gan05}.
Thus, a well structured analysis of a dynamic system should first explore the set of possible states, then the long-term dynamics expressed by the transitive context. If a more fine-grained investigation is desired, then also the transiton context may be explored. Finally, the temporal context should be computed out of $\mathbb{K}_{t}$, respectively $\mathbb{K}_{tt}$, if the transitive context is sufficient to determine the semantics of the interesting attributes (compare Remark \ref{rem:semanticsKtt}). Then the stem base of $\mathbb{K}_{s}$ serves as background knowledge for the exploration of $\mathbb{K}_{tt}$ and the resulting stem base can be used for the exploration of $\mathbb{K}_{t}$. 

Given a state context $\mathbb{K}_s$ as scale for a transition as well as a transitive context (p.~\pageref{scaleKsKt}), exploration means determining the attribute logic with respect to a set of background sequents \glossary{name={$\bigwedge$}, description={Infimum of an ordered set}, sort={$.Inf}}\glossary{name={$\bigvee$}, description={Supremum of an ordered set}, sort={$.Sup}}$\bigwedge A \rightarrow \bigvee B,\: A,B \subseteq M,$ valid in $\mathbb{K}_s$ \cite[104f.]{Gan05}. A formal context is determined by its stem base up to object reduction and by a set of sequents up to object clarification. The stem base of $\mathbb{K}_s$ excludes only object intents but does not determine which intents are object intents. It can be used as background knowledge instead of computing the sequents (possibly additionally to an already performed attribute exploration), if one is only interested in the implicational logic of the transition context and does not need to fix occurring states positively. However, if partially defined states should be excluded, in general sequents of the structure $\emptyset \multimap B, B \subseteq M$ are necessary. They can express that for each attribute a scale value exists, for instance $\top \rightarrow m.0 \vee m.1$ for a dichotomic attribute $m \in M$.

As will be shown in Section \ref{sec:incompleteDet}, generally the stem base of $\mathbb{K}_{tt}$ does not determine $\mathbb{K}_{tmp}$ uniquely. However, I will prove semantic inference rules in order to integrate respective background knowledge into the exploration of a temporal context. Since $\mathbb{K}_{tmp}$ extends the attribute set of $\mathbb{K}_{s}$ and restricts its object set to the set of input states, the implications of $\Ks$ remain valid but are possibly \agrave slightly" too restrictive. Nevertheless, they can serve as background knowledge for the exploration of $\Ktmp$. 

Finally, $\operatorname{Imp} \mathbb{K}_{tt}$ is partly determined by $\operatorname{Imp}\mathbb{K}_t$, e.g.
\[A_{in} \rightarrow B_{out},\: B_{in} \rightarrow C_{out} \:(\text{in }\mathbb{K}_t) \models A_{in} \rightarrow C_{out}\: (\text{in }\mathbb{K}_{tt}).\]
Hence, it could also make sense to first explore $\mathbb{K}_{t}$ and to integrate the information into the exploration of $\mathbb{K}_{tt}$ by inference rules. I will not investigate this rare case further. 

\section[Transitive and temporal contexts]{Transitive and temporal contexts: A calculus for implications between temporal atomic propositions}
\label{sec:calculus}
\sectionmark{Transitive and temporal contexts}

I searched for first order logic background formula in order to use the results of an attribute exploration of $\Ktt$ for the exploration of $\Ktmp$. Then the implications of the latter context are derivable from this background knowledge and a reduced set of new implications. The developed method also generates rules between implications of $\Ktmp$ alone. Therefore, during an exploration of $\Ktmp$ implications can be decided automatically based on the stem base of $\Ktt$ and on already accepted implications. The second type of inferences exploits implications between temporal attributes following from their definitions, like $\Box G m \rightarrow \Diamond F m$.

Prior to searching inference rules systematically by means of a test context, the subsequent proposition summarises inference rules between implications of $\Ktt$ and $\Ktmp$ that follow immediately from the definition of the temporal attributes by $\Ktt$ (Remark \ref{rem:semanticsKtt}):
\begin{proposition}\label{prop:KtKs}
Let $\mathbb{K}_{tmp}=(S^{in}, M \cup T, I^{in} \cup I_T)$ be a temporal context, $m,m_1$ and $m_2 \in M$ and $B \subseteq M$. Suppose the relation $t(R)
\subseteq S \times S$ is the object set of the related transitive context $\mathbb{K}_{tt} = (t(R),
M\times\{in,out\}, \nabla)$. Then the following entailments between implications of both
contexts are valid:
\begin{align}
B^{in} \rightarrow m^{out}\: (\text{in }\mathbb{K}_{tt}) &\equiv B \rightarrow \Box\, G m \: (\text{in }\mathbb{K}_{tmp})\label{prop:1}\\
\label{prop:2}B^{in} \rightarrow m^{out}\: (\text{in }\mathbb{K}_{tt}) &\models B \rightarrow \Diamond F m \: (\text{in }\mathbb{K}_{tmp})\\
\label{prop:3}B^{out} \rightarrow m^{out} \: (\text{in }\mathbb{K}_{tt}) &\models \Box\, G B \rightarrow \Box\, G m \: (\text{in }\mathbb{K}_{tmp})\\
\label{prop:4}m_1^{out} \rightarrow m_2^{out} \: (\text{in }\mathbb{K}_{tt}) &\models \Diamond F m_1 \rightarrow \Diamond F m_2 \: (\text{in}\: \mathbb{K}_{tmp})
\end{align}
If there is no transition $(s^{in},s^{out}) \in t(R)$ with extents $(s^{in})^I = (s^{out})^I = M$, the inference holds:
\begin{equation}\label{prop:5}
B^{in}\cup m^{out} \rightarrow \bot\: (\text{in }\mathbb{K}_{tt}) \models B \rightarrow
\Box\,\neg F m \: (\text{in }\mathbb{K}_{tmp})
\end{equation}
\end{proposition}
\begin{proof}
We remind the reader that for relations $R \subseteq A \times B$ and $a \in A$ the notation $[a]R$ means $[a]R := \{b \in B \mid (a,b) \in R\}$.\glossary{name={$[s]R$}, description={Set of output states for $s$ given $R$},sort=$SR$}
First we show item~(\ref{prop:3}).
\begin{equation*}
 \begin{aligned}
 &B^{out} \rightarrow m^{out}\: (\text{in }\mathbb{K}_{tt})\\
 \Leftrightarrow\: &\forall\,
 (s^{in},s^{out}) \in t(R)\colon s^{out} \in B^I \Rightarrow s^{out} I m\\
 \Rightarrow\: &\forall\, s^{in} \in S^{in}\colon \text{\large(}\forall\, (s^{in},s^{out}) \in t(R)\colon s^{out} \in B^I \text{\large)} \Rightarrow \text{\large(}\forall\, (s^{in},{s^{out}})
 \in t(R)\colon s^{out} I m \text{\large)}\\
 \Leftrightarrow\: &\Box\, G B \rightarrow \Box\, G m\: (\text{in }\mathbb{K}_{tmp}).
 \end{aligned} 
\end{equation*}

Now we show item~(\ref{prop:4}).
 \begin{align*}
 &m_1^{out} \rightarrow m_2^{out}\: (\text{in }\mathbb{K}_{tt})\\
 \Leftrightarrow\: &\forall\, (s^{in},s^{out}) \in t(R)\colon
 (s^{in},s^{out})\nabla m_1^{out} \Rightarrow (s^{in},s^{out}) \nabla m_2^{out}\\
 \Leftrightarrow\: &\forall\, (s^{in},s^{out}) \in t(R)\colon
 s^{out} I m_1 \Rightarrow s^{out} I m_2\: (\text{in $\Ks$})\\
 \Leftrightarrow\: &\forall\, s^{in} \in S^{in},\: \forall\,
 s^{out} \in [s^{in}]t(R)\colon s^{out} \in m_1^I \Rightarrow s^{out} I m_2\\
 \overset{{s^{out}}':= s^{out}}{\Rightarrow}\: &\forall\,
 s^{in} \in S^{in}\colon \text{\large(}\exists\, s^{out} \in [s^{in}]t(R)\colon s^{out} \in m_1^I \text{\large)} \Rightarrow \text{\large(}\exists\,
 {s^{out}}' \in [s^{in}]t(R)\colon {s^{out}}' I m_2 \text{\large)}\\
\Leftrightarrow\: & \Diamond F m_1 \rightarrow \Diamond F m_2\: (\text{in }\mathbb{K}_{tmp}).
\end{align*}

The remaining equivalences are proven analogously. The supplementary condition for $\eqref{prop:5}$ ensures that $\bot^I = M^I = \emptyset$. Hence, the presupposed implication signifies exclusion of $B^{in}$ and $m^{out}$, or $B^I=\emptyset$. In the latter case the inferred conclusion is trivially valid. 
\end{proof}

The restriction of $\Ktmp$ to an object set $S^{in}$ (Definition \ref{def:tmpCxt}) is necessary for the proofs. Regarding (\ref{prop:2}), consider e.g.\ $S= \{0,1\}$, $s =0$, $R = \{(1,1)\}$ and $M=\{m,n\}$, $B=\{n\}$ and $I = \{(0,n),(1,m)\}$. Then $B^{in} = \{n^{in}\} \not\subseteq (1,1)^{\nabla}$, so the left-hand side of the entailment is trivially true. However, $B \subseteq 0^I$ but $(0,\Diamond Fm)\notin I_T$, since $Seq^s_R =\emptyset$ for $s=0$.

\subsection{Incomplete determination of $\mathbb{K}_{tmp}$ by the stem base of $\mathbb{K}_{tt}$}\label{sec:incompleteDet}
For $\Diamond F$, $\Box\, G$ and $\Box\, \neg F$, the complete dynamics of the investigated system, thus $\mathbb{K}_{tmp}$, is determined by the respective transitive (or transition) context, but there is no one-to-one correspondence between the stem bases of $\mathbb{K}_{tmp}$ and $\mathbb{K}_{tt}$. Reducible transitions may exist that prevent a temporal implication, for nondeterministic processes or a data set consisting of different deterministic time series. Omitting them does not change the stem base of $\Ktt$. For an example, see Table \ref{counterex}.
\begin{table}\centering
\begin{tabular}{|l||c|c|c||c|c|c|}
\hline\textbf{Transition}
&$a^{in}$ &$b^{in}$ &$c^{in}$ &$a^{out}$ &$b^{out}$ &$c^{out}$\\
\hline \hline
$(s_0^{in}, s_1^{out})$ &$\times$&$\times$&&&$\times$&$\times$\\ \hline
$(s_0^{in}, s_2^{out})$ &$\times$&$\times$&&$\times$&&$\times$\\ \hline
$(s_1^{in}, s_2^{out})$ &&$\times$&$\times$&$\times$&&$\times$\\ \hline
$(s_3^{in}, s_4^{out})$ &&$\times$&&&&$\times$\\
\hline
\end{tabular}
\caption{The transitive context $\mathbb{K}_{tt}$ for two deterministic time series $s_0 - s_1 - s_2$  and $s_3 - s_4$.}
\label{counterex}
\end{table}

The stem base of this transitive context is:
\begin{center}
\begin{minipage}[t]{65mm}
\begin{enumerate}
 \item $\top \rightarrow b^{in}, c^{out}$
 \item $b^{in}, c^{in}, c^{out} \rightarrow a^{out}$
 \item $b^{in}, b^{out}, c^{out} \rightarrow a^{in}$
\end{enumerate}
\end{minipage}
\begin{minipage}[t]{65mm}
\begin{enumerate}\setcounter{enumi}{3} 
 \item $a^{in}, b^{in}, c^{in}, a^{out}, c^{out} \rightarrow b^{out}$ 
 \item $a^{in}, b^{in}, a^{out}, b^{out}, c^{out} \rightarrow c^{in}$
\end{enumerate}
\end{minipage}
\end{center}
$(s_3^{in}, s_4^{out})$ is a reducible object, because for example:
\begin{align*}
(s_3^{in}, s_4^{out})' &= \{b^{in}, c^{out}\} = \{a^{in}, b^{in}, b^{out}, c^{out}\} \cap \{b^{in}, c^{in}, a^{out}, c^{out}\}\\
&= (s_0^{in}, s_1^{out})' \cap (s_1^{in}, s_2^{out})'
\end{align*}
Therefore, the context without this transition has the same stem base. However, only without this observation $\top \rightarrow \Diamond F a$ holds in $\mathbb{K}_{tmp}$.

The following proposition concerns a linear order $R$ and specifies a special case where the temporal context is determined by the stem base of the transitive context: If it is known that only transitions with the first (or only with the last) transition are reducible, not all of them can be reduced and the order relation can be reconstructed.
\begin{proposition}
Let $\mathbb{K}_{tt} = (t(R), M\times\{in,out\},\nabla)$ be a transitive context with a linear order $R \subseteq S \times S$, $R \neq \emptyset$. Then the removal from $R$ of all transitions with the first (last) state changes the stem base of $\mathbb{K}_{tt}$. In general however, the following changes might not affect the stem base of $\Ktt$, but result in a different stem base of the temporal context $\mathbb{K}_{tmp}$ defined by $\mathbb{K}_{tt}$:
\begin{enumerate}
 \item The removal of a state and all related transitions that is not the first or last state.
 \item A switch of the order relation of two states.
\end{enumerate}
\end{proposition}
\begin{proof}
The implication logic determines $\mathbb{K}_{tt}$ up to object reduction.
If $\mathbb{K}_{tt}$ consists of a single transition, it is reducible only if $(s_0^{in}, s_1^{out}) = (M \times\{in,out\})' = \emptyset'$. But since the resulting concept lattice is isomorphic to the one element lattice, $\mathfrak{B}(\emptyset, \emptyset, \emptyset) \cong \mathfrak{B}((s_0^{in}, s_1^{out}), \emptyset, \emptyset)$, this case could be considered as attribute reduction and therefore excluded (attribute reduction results in a different stem base).

A state $s_k$ can be removed from the linear order, if all transitions $(s_i, s_k)$ and $(s_k, s_j)$ are reducible. Then all $s_i$ have to be intersections of input, all $s_j$ intersections of output state intents, which cannot be generelly excluded. 

However, if $k=0$, only the output parts of the respective transitions are available for intersection. With the order relation of intersection, these intents build a lattice, and each nonempty lattice has $\bigvee$-irreducible elements. 
Hence, a corresponding transition $(s_0,s_j)$ is irreducible, and the linear order may be reconstructed. The same argument holds for a final output state $s_k$: There exist irreducible $(s_i, s_k)$.

Changing the order of two states may result in the same stem base of $\mathbb{K}_{tt}$, but generate a temporal context with different implicational logic: In the transitive context derivable from the transition context of Table \ref{tab:counterExLin}, the transition $(s_5,s_6)$ is reducible. Changing only the order of $s_5$ and $s_6$ results in the same transitive context, with the exception of transition $(s_6,s_5)$. It is also reducible to the same context. However, the corresponding temporal contexts are different. The implication $\Diamond F a \rightarrow \Diamond F b$ only holds in the original, $\Diamond F b \rightarrow \Diamond F a$ only in the modified temporal context.
\end{proof}

\begin{table}\centering
\begin{tabular}{|l||c|c|c|c||c|c|c|c|}
\hline\textbf{Transition}
&$a^{in}$ &$b^{in}$ &$c^{in}$ &$d^{in}$ &$a^{out}$ &$b^{out}$ &$c^{out}$ &$d^{out}$\\
\hline \hline
$(s_0^{in}, s_1^{out})$ &$\times$&&&$\times$&&$\times$&&$\times$\\ \hline
$(s_1^{in}, s_2^{out})$ &&$\times$&&$\times$&$\times$&&$\times$&\\ \hline
$(s_2^{in}, s_3^{out})$ &$\times$&&$\times$&&&$\times$&$\times$&\\ \hline
$(s_3^{in}, s_4^{out})$ &&$\times$&$\times$&&$\times$&$\times$&&$\times$\\ \hline
$(s_4^{in}, s_5^{out})$ &$\times$&$\times$&&$\times$&$\times$&&&\\ \hline
$(s_5^{in}, s_6^{out})$ &$\times$&&&&&$\times$&&\\ \hline
$(s_6^{in}, s_7^{out})$ &&$\times$&&&&&&\\
\hline
\end{tabular}
\caption{Transition context for a linearly ordered state set. If the order of $s_5$ and $s_6$ is exchanged, the stem base of the respective reduced $\mathbb{K}_{tt}$ remains the same, but a different $\mathbb{K}_{tmp}$ is generated.}
\label{tab:counterExLin}
\end{table}
Nevertheless, note that there exist cases 1. and 2. where $\Ktmp$ is reconstructable from the stem base of $\Ktt$. For instance, if the order is changed at an early point, this may result in the same temporal context. In Table \ref{tab:counterExLin2}, the object $(s_0, s_1)$ is reducible. Thus, supposing a linear order the relation of $s_0$ and  $s_1$ is not reflected by the stem base. In general, the inverse relation may induce different temporal attributes. In this example, however, the context with  $(s_1, s_0)$ as first transition has the same stem base, but also the same temporal context $\mathbb{K}_{tmp}$ (only the attributes $b^{in}$ and $b^{out}$ are changed in the first row). 
\begin{table}\centering
\begin{tabular}{|l||c|c|c||c|c|c|}
\hline\textbf{Transition}
&$a^{in}$ &$b^{in}$ &$c^{in}$ &$a^{out}$ &$b^{out}$ &$c^{out}$\\
\hline \hline
$(s_0^{in}, s_1^{out})$ &&$\times$&&&&\\ \hline
$(s_0^{in}, s_2^{out})$ &&$\times$&&$\times$&$\times$&\\ \hline
$(s_0^{in}, s_3^{out})$ &&$\times$&&&$\times$&$\times$\\ \hline
$(s_0^{in}, s_4^{out})$ &&$\times$&&&&$\times$\\ \hline
$(s_1^{in}, s_2^{out})$ &&&&$\times$&$\times$&\\ \hline
$(s_1^{in}, s_3^{out})$ &&&&&$\times$&$\times$\\ \hline
$(s_1^{in}, s_4^{out})$ &&&&&&$\times$\\ \hline
$(s_2^{in}, s_3^{out})$ &$\times$&$\times$&&&$\times$&$\times$\\ \hline
$(s_2^{in}, s_4^{out})$ &$\times$&$\times$&&&&$\times$\\ \hline
$(s_3^{in}, s_4^{out})$ &&$\times$&$\times$&&&$\times$\\
\hline
\end{tabular}
\caption{The transitive context for a linear order $s_0 \leq s_1 \leq s_2 \leq s_3 \leq s_4$. $(s_0,s_1)$ is reducible, as well as $(s_1,s_0)$ in a context with this transition instead, but both formal contexts generate the same temporal context.}
\label{tab:counterExLin2}
\end{table}
\begin{remark}
\textnormal{Implications of the form $B \rightarrow \Box \neg F m \text{ in }\mathbb{K}_{tmp}$ may be derived from implications $B^{in}\cup m^{out} \rightarrow \bot\text{ in }\mathbb{K}_{tt}$ (compare Proposition \ref{prop:KtKs}). However, implications with negated attributes in the premise can only be decided using the stem base of $\mathbb{K}_{tt}$ in the case of dichotomic scaling.\index{Scaling!dichotomic} Else, there is no explicit information regarding the non-occurrence of an attribute. In the example of Table \ref{tab:counterExLin}, the only implication in the stem base of $\mathbb{K}_{tt}$ with $a^{out}, b^{out}$ in the premise is $a^{out}, b^{out} \rightarrow d^{out}$. With $b^{out}, c^{out}, d^{out} \rightarrow \bot$, $\Box\, G a, \Box\, G b \rightarrow \Box\, \neg F c$ is derivable, but there is no inference for deriving from $\Ktt$ the valid implication $\Box\, F a, \Box\, F b, \Box\, \neg F c \rightarrow b$, which is supported by $\{s_3\}$ in $\Ktmp$. This can be easily seen, since there is no $\mathbb{K}_{tt}$ implication with $b^{in}$ in the conclusion. The same holds for the possibility operator $\Diamond$, since it is equivalent to the necessity operator in the case of deterministic time series.}
\end{remark}

\subsection{The test context}\label{sec:ruleEx}
In order to get a complete overview on valid \agrave pure'' or ``mixed" entailments between implications of $\mathbb{K}_{tt}$ and $\mathbb{K}_{tmp}$, we performed attribute exploration
of the following \textit{test context}: Given fixed sets $S$ and $M$, the attributes are implication forms $IF$\glossary{name={$IF$},description={Set of implication forms}, sort=$If$} in variables $V:=\{X,Y,Z\}$, the objects are all possible $\mathbb{K}_{tt}$ respectively the corresponding $\mathbb{K}_{tmp}$ (Table
\ref{tab:testCxt}) together with a variable assignment, and the incidence relation means: ``With the respective values of the set variables, the implication holds in the context.'' More precisely, the objects are appositions\index{Formal context!apposition} ($\mathbb{K}_{tt}|\mathbb{K}'_{tmp})$, where $\mathbb{K}'_{tmp}$ proceeds from $\mathbb{K}_{tmp}$ by replacing an object $s^{in}$ of the latter by all transitions
$(s^{in}, s^{out}) \in t(R)$. This state context $\mathbb{K}_{tmp}'$ is not object
clarified, i.e. redundant. The attributes  are given by $(s^{in}, s^{out})' =
s^{in}$$'$. According to Definition \ref{def:tmpCxt} of a temporal context, the last state of a finite path is omitted.


The investigated temporal operators are restricted to $\operatorname{ev}:=\Diamond F$, $\operatorname{alw}:=\Box G$ and $\operatorname{nev}:= \Box \neg F$ decidable by $\Ktt$, where the same quantifier $\exists$ or $\forall$ applies to paths and transitions (Definition \ref{def:tmpCxt}). 

Thus, we performed attribute exploration of the following test context:\index{Test context}
\glossary{name={$\mathbb{K}_{test}$}, description={Test context}, sort=$Ktest$}
\glossary{name={$\alpha$}, description={Variable assignment},sort=$Alpha$}
\[\mathbb{K}_{test} = (\{\alpha: \{X,Y,Z\} \rightarrow \mathfrak{P}(M)\setminus \emptyset\} / S_M \times \{\mathbb{K}_{tt}|\mathbb{K}_{tmp}'\}, IF, \models)\]
Variable assignments $\operatorname{mod} S_M$ (the symmetric group on $M$) are sufficient, because also $\{\mathbb{K}_{tt}|\mathbb{K}_{tmp}'\}$ is symmetric in the attributes $M$.
\begin{table}\centering
\begin{tabular}{|c||c|c|c||c|c|c||c|c|c|c|c|c|c|c|c|}
\hline
&\begin{sideways}$a^{in}$\end{sideways} &\begin{sideways}$b^{in}$\end{sideways} &\begin{sideways}$c^{in}$\end{sideways} &\begin{sideways}$a^{out}$\end{sideways} &\begin{sideways}$b^{out}$\end{sideways} &\begin{sideways}$c^{out}$\end{sideways}
&\begin{sideways}$\ev a$\end{sideways} &\begin{sideways}$\ev b$\end{sideways} &\begin{sideways}$\ev c$\end{sideways} &\begin{sideways}$\alw a\:$\end{sideways} &\begin{sideways}$\alw b$\end{sideways} &\begin{sideways}$\alw c$\end{sideways} &\begin{sideways}$\nev a$\end{sideways} &\begin{sideways}$\nev b$\end{sideways} &\begin{sideways}$\nev c$\end{sideways}\\  
\hline \hline
$(s_0^{in}, s_1^{out})$ &$\times$&&&$\times$&$\times$&&$\times$&$\times$&$\times$&&&&&&\\ \hline
$(s_0^{in}, s_2^{out})$ &$\times$&&&$\times$&$\times$&$\times$&$\times$&$\times$&$\times$&&&&&&\\ \hline
$(s_0^{in}, s_3^{out})$ &$\times$&&&&&$\times$&$\times$&$\times$&$\times$&&&&&&\\\hline 
$(s_1^{in}, s_2^{out})$ &$\times$&$\times$&&$\times$&$\times$&$\times$&$\times$&$\times$&$\times$&&&$\times$&&&\\ \hline
$(s_1^{in}, s_3^{out})$ &$\times$&$\times$&&&&$\times$&$\times$&$\times$&$\times$&&&$\times$&&&\\ \hline
$(s_2^{in}, s_3^{out})$ &$\times$&$\times$&$\times$&&&$\times$&&&$\times$&&&$\times$&$\times$&$\times$&\\
\hline
\end{tabular}
\caption{A single object of the test context: the apposition of the transitive and temporal contexts corresponding to the time series $a - ab - abc - c.$}
\label{tab:testCxt}
\end{table}

Since it was not feasible to investigate all possible implications, we chose the following classes with at most two premises (with different variables $X$, $Y$ or the constant $\top$) and a single conclusion. Many implications with more than two different attribute sets in the premise can be derived by the second and third Armstrong rule.
\begin{center}
\begin{tabular}{|l|l|l|}\hline
&$\mathbb{K}_{tt}$ &$\mathbb{K}_{tmp}$\\ \hline
1 or 2 premises& $X^{in}, Y^{out}, \top$ & $X,\, \Diamond FY,\, \Box GY,\, \Box \neg F Y, \top$\\
1 conclusion &$Z^{in}, Z^{out}, \bot$ &$Z,\, \Diamond FZ,\, \Box\, GZ,\, \Box \neg F Z, \bot$\\
\hline
\end{tabular} 
\end{center}
The sets were supposed to be nonempty. Thus, $\top:=\emptyset$\glossary{name={$\top$}, description={Empty attribute set},sort=$.Empty$} was considered explicitly, in order to avoid redundancies like $\Diamond F \emptyset = \Box G \emptyset$. Tautologic or unsatisfiable implication forms were not considered, like $\Box G Y \rightarrow \Diamond F Y$ or $\Diamond FY,\: \Box \neg F Y \rightarrow \bot$. 

\section{Inference rules for a three element attribute set $M$}\label{sec:rules3Attr}
Even with computational optimisation, it took 10 days, partly on two desktop computers, to generate the test context with $|M|=3$. Since computation time is exponential in $|M|$, it is not feasible to compute the test context for larger attribute sets. For $M=\{a,b,c\}$, we generated all deterministic time series\index{Process!deterministic} (linear orders) including cycles.

If the 81 inference rules of the stem base are proven, the hypothesis is confirmed that this attribute set is general enough. However, this was only possible for rules 1 to 56. For rule 57, 59 and 65, a supplementary condition is necessary restricting the allowed transitions (Propositions \ref{prop:nevEv} and \ref{prop:Yout1stStep}). Since during the computation of a stem base the generated rules depend on the previous, I proved only several further rules resulting in a large, but not complete set of inference rules for transitive and temporal context. Further research is needed: Expert centered attribute exploration should be performed by introducing counterexamples to rules that are not generalisable and by proving the remaining or new implications. Then the Theorem of Duquenne-Guigues \ref{theorem:dg}\index{Theorem!of Duquenne-Guigues} ensures the completeness of the stem base rules for drawing all rule-like conclusions between the implications investigated as attributes of the test context. 

A further purpose of the following proofs is to determine presuppostions as generally as possible. It turned out that most rules are also valid for nondeterminstic processes.

As indicated by the used sofware \texttt{ConImp} \cite{Bur03}, we have to prove only the parts of the premise and conclusion not following from previous implications. Several implications follow trivially, since $\bot$ means \agrave every attribute occurs", for instance $\alw Y \rightarrow \bot \: \models \alw Y \rightarrow Z^{in}$. 

In a first step, general propositions are proven for example inferences. The attached file \texttt{testCxt.pro} -- the output of the test context exploration -- contains all 81 inference rules. Their indices are listed in Table \ref{tab:allInfRules} together with the propositions necessary for the proofs (sometimes easy consequences are needed, or analogous applications of the proofs). 

For the subsequent propositions, the following presuppositions are made: $\Ktt = (t(R),M\times \{in,out\},\nabla)$ with $R \subseteq S \times S$ and $\Ktmp = (S^{in}, M \cup T, I^{in} \cup I_T)$ are formal contexts generated by the same process. $B$, $X$, $Y$ and $Z \subseteq M$ are their basic attribute sets. If the process is not qualified, it may be deterministic or nondeterministic.\index{Process!nondeterministic} $i,j, k \in \mathbb{N}_0$ are indices from sets according to the number of transitions of a path. Implications with attribute sets indexed by $.^{in}$ or $.^{out}$ refer to a transitive context and the left part of the apposition $\Ktt \mid \Ktmp'$, the remaining attributes to $\Ktmp'$ and therefore to the temporal context $\Ktmp$ generated by $\Ktt$.

\begin{proposition}\label{prop:2ndArmstrong}
For all $\Ktt$ and $\Ktmp$ generated by the same process, the following inference rules hold by the second Armstrong rule.\index{Armstrong rules} 
\begin{align*}
11:\: & \operatorname{alw} Y \rightarrow \bot  &\models\: &X \wedge \alw Y \rightarrow \bot,\: \ev Y \wedge \alw Y \rightarrow \bot,\\ 
&&&X \wedge \nev Y \rightarrow \nev Z\\
5: \: &\top \rightarrow Z^{in} &\models\: & X^{in} \rightarrow Z^{in},\: Y^{out} \rightarrow Z^{in},\: \ev Y \rightarrow Z,\: X^{in} \wedge Y^{out} \rightarrow Z^{in}, ...
\end{align*}
\end{proposition}
\begin{proof}
Rule 11 follows by the second Armstrong rule, i.e. an expansion of the premise $P$, with attributes $A_i \subseteq M \cup (M \times \{in,out\}) \cup T$ (where attributes of $\Ktmp$ are in $M$ or $T$):
\begin{equation}\label{eq:2ndArmstrong}
\begin{aligned}
P: \underline{\qquad A_2\:} &\underline{\:\rightarrow A_3}\\
C: A_1 \wedge A_2 &\rightarrow A_3
\end{aligned}
\end{equation}
In rule 5, $A_2 = \emptyset$. With the presupposed implication of $\Ktt$, all implications of $\Ktt$ with  the conclusion $Z^{in}$ and all implications of $\Ktmp$ with  the conclusion $Z$ can be inferred, since the extents are equal: $Z^{in} = Z \subseteq M \Rightarrow (Z^{in})^{I^{in}} = Z^{I^{in}}$ (remember that in $\Ktmp$ the object set is restricted to the set of input states $S^{in}$, consequently the relation $I$). By further applications of the second Armstrong rule, implications with arbitrarily enlarged premises are shown to be valid.
\end{proof}

\begin{proposition}\label{prop:definitions}
For all $\Ktt$ and $\Ktmp$ generated by the same process, the following inference rules hold by definition of the temporal operators $\operatorname{ev},\:\operatorname{alw}$ and $\operatorname{nev}$:
\begin{align*}
1.\: &26\colon\: &\operatorname{ev} Y \wedge \alw Y \rightarrow \bot 
\qquad &\models &&\alw Y \rightarrow \bot\\
&27\colon\:  &\ev Y \wedge \alw Y \rightarrow \nev Z\qquad &\models 
&&\alw Y \rightarrow \nev Z\\
2.\: &41:\: &X \wedge \ev Y \rightarrow \bot\qquad &\models &&X^{in} 
\wedge Y^{out} \rightarrow \bot\\
&44:\: &X \wedge \ev Y \rightarrow \ev Z\qquad &\models &&X \wedge \alw 
Y \rightarrow \ev Z\\
&46:\: &X \wedge \ev Y \rightarrow Z^{in}\qquad &\models &&X^{in} \wedge 
Y^{out} \rightarrow Z^{in},\:  X \wedge \alw Y \rightarrow Z\\
3.\: &32:\: &X \wedge \nev Y \rightarrow \alw Z\qquad &\models &&X 
\wedge \nev Y \rightarrow \ev Z\\
&37:\: &X \wedge \alw Y \rightarrow \alw Z\qquad &\models &&X \wedge 
\alw Y \rightarrow \ev Z\\
  &43:\: &X \wedge \ev Y \rightarrow \alw Z\qquad &\models &&X^{in} 
\wedge Y^{out} \rightarrow Z^{out},\: X \wedge \alw Y \rightarrow \alw Z.
\end{align*}
\end{proposition}
\begin{proof}
Let $D$ be an implication valid in the test context that follows immediately from the definitions. Examples are the subsequent implications of $\Ktmp$ (\ref{eq:alwEv}) or with attributes of $\Ktmp$ and $\Ktt$:
\begin{align}
\alw B &\rightarrow \ev B\label{eq:alwEv}\\
\alw B &\rightarrow B^{out}\label{eq:alwBout}\\
B^{out} &\rightarrow \ev B\label{eq:BoutEv}.
\end{align}
We distinguish the following inference modi, with attribute sets $A_i \subseteq M \cup (M \times \{in,out\}) \cup T$. They are special cases of the third Armstrong rule that expresses generalised transitivity.\index{Armstrong rules}
\begin{enumerate}
\item
\begin{equation}\label{eq:premPremise}
\begin{aligned}
P: A_1 \wedge A_2 &\rightarrow A_3\\
D: \underline{\qquad\:\, A_2 } &\underline{\:\rightarrow A_1}\\
C:  \qquad\:\, A_2 &\rightarrow A_3
\end{aligned}
\end{equation}
With $A_1:= \ev Y,\: A_2 := \alw Y$ and $A_3 := \bot$, $P \models C$ is inference rule 26. Rule 27 and others follow with different attribute sets $A_3$ (see Table \ref{tab:allInfRules}, where also the used implications $D$ are indicated in the column \agrave Definition").

\item In 41, 44 and 46, a part of the premise is replaced by a stronger assumption. Here D is either $\alw Y \rightarrow \ev Y$ or $Y^{out} \rightarrow \ev Y$:
\begin{equation}\label{eq:premConclusion}
\begin{aligned}
P: A_1 \wedge A_2 &\rightarrow A_3\\
D: \underline{\qquad\; A_4 } &\underline{\:\rightarrow A_2}\\
C:  A_1 \wedge A_4 &\rightarrow A_3
\end{aligned}
\end{equation}

\item In 32 and 37, $D\colon \alw Y \rightarrow \ev Y$ is applied to the conclusion of the presupposed rule (ordinary transitive inference):
\begin{equation}\label{eq:conPremise}
\begin{aligned}
P: A_1 \wedge A_2 &\rightarrow A_3\\
D: \underline{\qquad\; A_3 } &\underline{\:\rightarrow A_4}\\
C:  A_1 \wedge A_2 &\rightarrow A_4
\end{aligned}
\end{equation}
In 43, $D\colon \alw Z \rightarrow Z^{out}$ is used, additionally $\alw Y \rightarrow \ev Y$ and $Y^{out} \rightarrow \ev Y$ according to inference pattern (\ref{eq:premConclusion}).
\end{enumerate}
\end{proof}

If the inference rules are implemented as an extension of an attribute exploration software, it has to be decided if the rules according to this proposition are an immediate input to a reasoner. Alternatively, implications as \eqref{eq:alwEv}, \eqref{eq:alwBout} and  \eqref{eq:BoutEv} could be assembled, together with an implementation of the Armstrong rules for attribute exploration itself. Proposition \ref{prop:2ndArmstrong} may have little  practical relevance: If an implication $A_2 \rightarrow A_3$ is accepted, $A_1 \cup A_2$ does not contain the closure $A_3$ of the pseudo-intent $A_2$. Therefore, it is no pseudo-intent and cannot be proposed as implication of the stem base. At least the proposition could be useful for \agrave mixed'' implications: rule 5 translates the implication $\top \rightarrow Z^{in}$ of $\Ktt$ into implications of $\Ktmp$. However, I will not discuss such issues of practical applicability further. 

\begin{proposition}\label{prop:ZoutAlw}
Let $\mathbb{K}_{tmp}$ be a temporal context and $\Ktt$ the related transitive context.
Then the entailment between implications of $\mathbb{K}_{tmp}$ and $\Ktt$ is valid:
\[53\colon X^{in} \rightarrow Z^{out} \models X \rightarrow \alw Z.\]
\end{proposition}
\begin{proof}
The rule follows from Proposition \ref{prop:KtKs} \eqref{prop:1} applied to every $z \in Z$.
\end{proof}
If the inferences of Proposition \ref{prop:KtKs} are entered as background knowledge for the exploration of the test context (with variants for the set variables $X,Y$ and $Z$), rule 53 and a few others can be decided automatically.

\begin{proposition}\label{prop:rule1}
For all $\Ktt$ and $\Ktmp$ generated by the same process, the following inference rules hold by exclusion of $\ev/\alw$ and $\nev$:
\begin{align}
&\label{eq:exclusion}\;1:\:\ \top \models \: \ev Y \wedge \nev Y \rightarrow \bot,\: \alw Y \wedge \nev Y \rightarrow \bot &\\
&38:\: X \wedge \alw Y \rightarrow \ev Z,\: X \wedge \alw Y \rightarrow \nev Z \models\: X \wedge \alw Y \rightarrow \bot &
\end{align}
\end{proposition}
\begin{proof}
By definition, $\ev y$ and $\nev y$ cannot be both attributes of a state $s \in S^{in}$, for all $y \in Y$. By pairwise exclusion, the extent $(\ev Y \cup\, \nev Y)^{I_T}$ is empty; this means $\ev Y \wedge\, \nev Y \rightarrow \bot$. The second implication of rule 1 follows with $\eqref{eq:alwEv}$ and the third Armstrong rule \eqref{eq:premConclusion}. 38 is inferred by $\eqref{eq:conPremise}$ with the attribute set $A_3 = \ev Z\, \cup\, \nev Z$.
\end{proof}

A surprising rule expressing more complex temporal relationships is the following.
\begin{proposition}\label{prop:inEvAlw}
Let $\Ktmp$ be a temporal context of a deterministic process, i.e. $\forall s \in S^{in}\colon |[s]R| := |\{s^{out} \in S \mid \exists (s, s^{out}) \in R\}| = 1$. Then the subsequent entailments between implications of $\Ktmp$ are valid:
\begin{align*}
39\colon \alw Y \rightarrow Z,\: \alw Y \rightarrow \ev Z &\models \alw Y \rightarrow \alw Z\\
34\colon \nev Y \rightarrow Z,\: \nev Y \rightarrow \ev Z &\models \nev Y \rightarrow \alw Z.
\end{align*}
\end{proposition}
\begin{proof}
We consider an arbitrary transitive context $\mathbb{K}_{tt} = (t(R), M\times\{in,out\}, \nabla)$ generating $\mathbb{K}_{tmp}$ and differentiate between two cases for a state $s_0^{in} \in S^{in}$:
\begin{enumerate}
\item $\alw Y \nsubseteq (s_0^{in})^{I_T}$.\\
$\Rightarrow$ The implications and the entailment trivially hold.
\item $\alw Y \subseteq (s_0^{in})^{I_T}$. Then\\
\begin{tabular}{lll}
&&$\alw Y$ $\rightarrow Z,\: \alw Y \rightarrow \ev Z$ is supported by $s_0^{in}$.\\
$\Leftrightarrow$ &&$\forall\, (s_0^{in}, s^{out}) \in t(R),\: \forall\, y \in Y\colon (s^{out}, y) \in I \qquad\qquad (*)$\\
&$\wedge$ &$\forall z \in Z\colon (s_0^{in}, z) \in I$\\
&$\wedge$ &$ \forall z \in Z\: \exists\, (s_0^{in}, s^{out}) \in t(R)\colon (s^{out},z) \in I$.\\
$\Rightarrow$ &&$\forall\, (s_0^{in}, s^{out}) \in t(R)\colon \alw Y \subseteq (s_0^{out})^{I_T}$,\\
&&if $s^{out} \in \Ktmp$, i.e. $\exists\, s^{out'}\colon (s^{out},s^{out'}) \in t(R)$.\\
$\Rightarrow$ &&$\forall\, (s_0^{in}, s^{out}) \in t(R),\: \forall z \in Z\colon (s^{out},z) \in I$:\\
&&If $s^{out}$ is a final state, $(s^{out},z) \in I$ follows with $\alw Y \rightarrow \ev Z$\\
&& for an immediately antecedent state and a deterministic process.\\
$\Leftrightarrow$ &&$\alw Y \rightarrow \alw Z$ holds for $s_0^{in}$. 
\end{tabular}
\end{enumerate} 
Like $\alw$, $\nev$ applies to every state subsequent to $s_0^{in}$, if defined and $\nev Y \subseteq (s_0^{in})^{I_T}$. Therefore, the entailment for $\nev Y$ is proven in the same way with $(s^{out},y)\notin I$ in statement $(*)$.
\end{proof}

\begin{table}\centering
\begin{tabular}{|c|l|l|}\hline
Rule &Inference modus &Definition\\ \hline
1. &Proposition \ref{prop:rule1}, (\ref{eq:premConclusion}) &(\ref{eq:alwEv})\\ \hline
2. &(\ref{eq:2ndArmstrong}) &\\ \hline
3. &(\ref{eq:conPremise}), (\ref{eq:2ndArmstrong})  &(\ref{eq:alwEv}), (\ref{eq:alwBout}), (\ref{eq:BoutEv})\\ \hline
4. &(\ref{eq:2ndArmstrong}) &\\ \hline
5. &(\ref{eq:2ndArmstrong}) &\\ \hline
6. &trivial, (\ref{eq:2ndArmstrong}) &\\ \hline
7. &(\ref{eq:2ndArmstrong}) &\\ \hline
8. &(\ref{eq:conPremise}), (\ref{eq:2ndArmstrong}) &(\ref{eq:alwEv})\\ \hline
9. &(\ref{eq:2ndArmstrong}) &\\ \hline
10. &(\ref{eq:2ndArmstrong}) &\\ \hline
11. &trivial, (\ref{eq:2ndArmstrong}) &\\ \hline
12. &(\ref{eq:2ndArmstrong}) &\\ \hline
13. &(\ref{eq:conPremise}), (\ref{eq:2ndArmstrong}) &(\ref{eq:alwEv})\\ \hline
14. &(\ref{eq:2ndArmstrong}) &\\ \hline
15. &(\ref{eq:2ndArmstrong}) &\\ \hline
16. &trivial, (\ref{eq:2ndArmstrong}) &\\ \hline
17. &(\ref{eq:premConclusion}), (\ref{eq:2ndArmstrong}) &(\ref{eq:alwEv})\\ \hline
18. &(\ref{eq:premConclusion}), (\ref{eq:conPremise}), (\ref{eq:2ndArmstrong}) &(\ref{eq:alwEv}), (\ref{eq:alwBout}), (\ref{eq:BoutEv})\\ \hline
19. &(\ref{eq:premConclusion}), (\ref{eq:2ndArmstrong}) &(\ref{eq:alwEv})\\ \hline
20. &(\ref{eq:premConclusion}), (\ref{eq:2ndArmstrong})  &(\ref{eq:alwEv}), (\ref{eq:BoutEv})\\\hline
21. &trivial, (\ref{eq:2ndArmstrong}) &\\ \hline
22. &(\ref{eq:2ndArmstrong}) &\\ \hline
23. &(\ref{eq:conPremise}), (\ref{eq:2ndArmstrong}) &(\ref{eq:alwEv}), (\ref{eq:alwBout})\\ \hline
24. &(\ref{eq:2ndArmstrong}) &\\ \hline
25. &(\ref{eq:2ndArmstrong}) &\\ \hline
26. &(\ref{eq:premPremise}) &(\ref{eq:alwEv})\\ \hline
27. &(\ref{eq:premPremise}) &(\ref{eq:alwEv})\\ \hline
28. &(\ref{eq:premPremise}) &(\ref{eq:alwEv})\\ \hline
29. &(\ref{eq:premPremise}) &(\ref{eq:alwEv})\\ \hline
30. &(\ref{eq:premPremise}) &(\ref{eq:alwEv})\\ \hline
31. &trivial &\\ \hline
32. &(\ref{eq:conPremise}) &(\ref{eq:alwEv})\\ \hline
33. &(\ref{eq:conPremise}), (\ref{eq:exclusion}) & \\ \hline
34. &Proposition \ref{prop:inEvAlw} & \\ \hline
35. &(\ref{eq:conPremise}), (\ref{eq:exclusion}) & \\ \hline
36. &trivial &\\ \hline
37. &(\ref{eq:conPremise}) &(\ref{eq:alwEv})\\ \hline
38. &Proposition \ref{prop:rule1}& \\ \hline
39. &Proposition \ref{prop:inEvAlw} &(\ref{eq:alwEv})\\ \hline
40. &(\ref{eq:conPremise}), (\ref{eq:exclusion}) & \\ \hline
41. &trivial, (\ref{eq:premConclusion})  &(\ref{eq:alwEv}), (\ref{eq:BoutEv})\\ \hline
42. &(\ref{eq:premConclusion}) &(\ref{eq:alwEv})\\ \hline
43. &(\ref{eq:premConclusion}), (\ref{eq:conPremise}) &(\ref{eq:alwEv}), (\ref{eq:alwBout}), (\ref{eq:BoutEv})\\ \hline
44. &(\ref{eq:premConclusion}) &(\ref{eq:alwEv})\\ \hline
45. &(\ref{eq:conPremise}), (\ref{eq:exclusion}) & \\ \hline
46. &(\ref{eq:premConclusion}) &(\ref{eq:alwEv}), (\ref{eq:BoutEv})\\ \hline
47. &Proposition \ref{prop:ZoutAlw} &\\ \hline
\end{tabular}
\end{table}
\begin{table}\centering
\begin{tabular}{|c|l|l|}\hline
Rule &Inference modus &Definition\\ \hline
48. &identical &\\ \hline
49. &trivial, (\ref{eq:2ndArmstrong}) & \\ \hline
50. &(\ref{eq:premConclusion}), (\ref{eq:2ndArmstrong}), Proposition \ref{prop:ZoutAlw} &(\ref{eq:alwBout})\\ \hline
51. &(\ref{eq:2ndArmstrong}) & \\ \hline
52. &identical & \\ \hline
53. &Proposition \ref{prop:ZoutAlw} &\\ \hline
54. &identical &\\ \hline 
55. &trivial &\\ \hline
56. &(\ref{eq:premConclusion}), Proposition \ref{prop:ZoutAlw} &(\ref{eq:alwBout})\\ \hline
57. &Proposition \ref{prop:nevEv} with supplementary condition &\\ \hline
58. &(\ref{eq:premConclusion}) &(\ref{eq:alwBout})\\ \hline
59. &\ref{prop:Yout1stStep} with supplementary condition &\\ \hline
65. &\ref{prop:Yout1stStep} with supplementary condition &\\ \hline
68. & \ref{prop:rule68} &\\ \hline
\end{tabular}\label{tab:allInfRules}
\caption{Inference rules for implications of $\Ktt$ and $\Ktmp$. The second column lists the inference modi or propositions required for the proofs, the third column implications following from the definitions of a temporal context.}
\end{table}

\begin{proposition}\label{prop:nevEv} 
For all $\Ktt$ and $\Ktmp$ generated by the same process, the following inference rule holds, if $\forall (s^{in},s^{out}) \in R, \forall y \in Y\colon y \in (s^{out})^I \Rightarrow Y \subseteq (s^{out})^I$, especially for $|Y|=1\!:$
\[57\colon X^{in} \wedge Y^{out} \rightarrow Z^{out},\: X \wedge \operatorname{ev} Y \rightarrow \operatorname{nev}Z \models X \wedge \operatorname{ev} Y \rightarrow \bot.\]
\end{proposition}
\begin{proof}


Suppose there is a state $s_0^{in} \in S^{in}$ with attributes $X$ and $\operatorname{ev} Y$. Then a transition $(s_0^{in},s^{out})$ exists in $\Ktt$ with $Y \subseteq (s^{out})^I$. For this step the condition of the proposition is important: A single transition with the property has to exist, the $y \in Y$ must not be attributes of subsequent states or of states belonging to different paths. By the first presupposed implication, then also $Z \subseteq (s^{out})^I$. This contradicts the second implication. Hence, there is no state with attributes $X$ and $\operatorname{ev} Y$ and the second implication trivially holds, i.e. $X \wedge \ev Y \rightarrow \bot$.
\end{proof}

Without the condition for the transitions $(s^{in},s^{out}) \in R$, the inference is valid:
\[X^{in} \wedge Y^{out} \rightarrow Z^{out},\: X \wedge \operatorname{ev} Y \rightarrow \operatorname{nev}Z \models X^{in} \wedge Y^{out} \rightarrow \bot.\]

This is the first inference rule which could not be generally proven during the attribute exploration of the test context restricted to a three element attribute set. The supplementary condition mostly restricts $Y$ to a one element set. This contradicts the aim of proving rules for implications with set variables, so the condition was not presupposed generally. As mentioned in the introduction of this section, an alternative would be an expert centered exploration, where counterexamples are introduced. However, this is out of the scope of this thesis. Nevertheless, I proved three further rules which express interesting dynamic dependencies.

\begin{proposition}\label{prop:Yout1stStep}
For all $\Ktt$ and $\Ktmp$ generated by the same process, the following inference rules holds, if $\forall\, (s^{in},s^{out}) \in R, \forall\, y \in Y\colon y \in (s^{out})^I \Rightarrow Y \subseteq (s^{out})^I$, especially for $|Y|=1\!:$
\[65: Y^{out} \rightarrow Z^{in},\: \ev Y \rightarrow \nev Z \models \ev Y \rightarrow Z \]
\[59: X^{in} \wedge Y^{out} \rightarrow Z^{in},\: X \wedge \ev Y \rightarrow \nev Z
\models\: X \wedge  \ev Y \rightarrow Z.\]
\end{proposition}
\begin{proof}
If $Y^{out} \rightarrow Z^{in}$ is supported by any, then by the immediate preceding transition $(s_k,s_{k+1})$ and by the transitions $(s_{k-i},s_{k+1})$ with all previous input states. All underlying paths have the structure:
\[ Z - \cdots- Z - Y - \cdots\]
Since $\ev Y$ is an attribute of the first state of a path, $\ev Y \rightarrow \nev Z$ requires the structure:
\[ Z - Y\overline{Z} - \overline{Z} - \cdots\]
Thus, only the first state has the attribute $\ev Y$ and also the attributes in $Z=Z^{in}$ as demonstrated.\\
The condition makes sure that there is no other path with a state $s \in (\ev Y)^{I_T}$, but there is no subsequent state $s' \in Y^I$. If there is no transition with $Y^{out}$ at all, the condition implies that $\ev Y$ does not apply to any state, and the inferred implication trivially holds.

Rule 59 is proven analogously.
\end{proof}

\begin{proposition}\label{prop:rule68}
For all $\Ktt$ and $\Ktmp$ generated by the same deterministic process, the inference rule holds:
\[68: \top \rightarrow Z^{in},\: X^{in} \wedge\, \alw Y \rightarrow \ev Z \models X^{in} \wedge\, \ev Y \rightarrow \ev Z.\]
\end{proposition}
\begin{proof}
By $\top \rightarrow Z^{in}$, all but the final states have attributes $Z$. The same holds for the corresponding output states of a transition. For the last input state of a path and a deterministic process, $\ev Y$ and $\alw Y$ have the same meaning, thus the inferred implication holds, too.
\end{proof}

\section{Inference rules for Boolean attributes}\label{sec:rulesBooleAttr}
In order to get a complete overview on valid entailments for Boolean attribute values, we performed manual attribute exploration of a test context similar to Section \ref{sec:ruleEx} with attributes $IF$ as listed below. Formal contexts are considered where the attributes are dichotomically scaled, i.e. $\Ktmp$ is an extension of a state context according to Definition \ref{def:stateCxt1}: $\Ks = (S,M,I)$ and $M\subseteq E \times F$. Now $F:=\{0,1\}$. The attributes of $\Ktt$ and $\Ktmp$ are constructed as usual from $M$. In the subsequent implications, the sets $B_0,B_1,C_0,C_1$ and $C$ are nonempty subsets of $M$. The indices express which of the fluents is assigned to the entities: $m=(e,f) \in B_0 \text{ or } C_0 \Rightarrow \gamma(s)(e) =0$ and $m \in B_1 \text{ or } C_1 \Rightarrow \gamma(s)(e) =1$.\label{imp} We suppose that all states and transitions are completely defined.
\begin{center}
\begin{minipage}[t]{60mm}
\begin{enumerate}
	\item $B^{in}\rightarrow C^{in}$
  \item $B^{in}\rightarrow C_0^{out} \equiv B^{in}\rightarrow$ $\nev C_1$
  \item $B^{in}\rightarrow C_1^{out} \equiv B^{in}\rightarrow$ $\alw C_1$
  \item $B^{in}\rightarrow$ $\ev C_1$
  \item $B_0^{out}\rightarrow C^{in}$
  \item $B_1^{out}\rightarrow C^{in}$
  \item $\ev B_1$ $\rightarrow C$
  \item $\alw B_1$ $\rightarrow C$
  \item $\nev B_1$ $\rightarrow C$
  \item $B_0^{out}\rightarrow C_0^{out}$
  \item $B_0^{out}\rightarrow C_1^{out}$
\end{enumerate}
\end{minipage}
\begin{minipage}[t]{45mm}
 \begin{enumerate}\setcounter{enumi}{11}
 \item $B_1^{out}\rightarrow C_0^{out}$
  \item $B_1^{out}\rightarrow C_1^{out}$
  \item $\ev B_1$ $\rightarrow$ $\ev C_1$
  \item $\ev B_1$ $\rightarrow$ $\alw C_1$
  \item $\ev B_1$ $\rightarrow$ $\nev C_1$
  \item $\alw B_1$ $\rightarrow$ $\ev C_1$
  \item $\alw B_1$ $\rightarrow$ $\alw C_1$
  \item $\alw B_1$ $\rightarrow$ $\nev C_1$
  \item $\nev B_1$ $\rightarrow$ $\ev C_1$
  \item $\nev B_1$ $\rightarrow$ $\alw C_1$
  \item $\nev B_1$ $\rightarrow$ $\nev C_1$
\end{enumerate}
\end{minipage}
\end{center}

The equivalences in 2. and 3. follow from Proposition \ref{prop:KtKs} (\ref{prop:1}). Since the
implications comprising input attributes are independent from those related only to output
attributes, attribute exploration was performed separately for the first 9 and the remaining
13 implications. Results for the second part are shown here.

The exploration with \texttt{ConImp}\index{ConImp} started from a hypothetical context as single object of the test context, where no implications are valid. Before, 25 known entailments were added as background rules (BR) like
those of Proposition \ref{prop:KtKs} or other rules following from the definitions like $\alw B_1 \rightarrow \ev B_1$. A counterexample had to be chosen carefully, since an object not having its maximal attribute set might preclude a valid entailment (compare p. \pageref{maxCounterEx}). The exploration resulted in the following stem base of only 14 entailments.
Most of them are background rules (they are accepted automatically during the exploration), but not
all of these are needed in order to derive all valid entailments between the chosen implications.
This demonstrates the effectivity and minimality of the algorithm. Entailments 5., 6., 7. and 10. were new findings.
\begin{enumerate}
  \item $\nev B_1$ $\rightarrow$ $\alw C_1$
  $\models$ $\nev B_1$ $\rightarrow$ $\ev C_1$ (BR 1)
  \item $\nev B_1$ $\rightarrow$ $\ev C_1$, $\nev B_1$ $\rightarrow$ $\nev C_1$
  $\models \bot$ (BR 11)
  \item $\alw B_1$ $\rightarrow$ $\alw C_1$
  $\models$ $\alw B_1$ $\rightarrow$ $\ev C_1$ (BR 2)
  \item $\alw B_1$ $\rightarrow$ $\ev C_1$, $\alw B_1$ $\rightarrow$ $\nev C_1$
  $\models \bot$ (BR 14)
  \item $\ev B_1$ $\rightarrow$ $\nev C_1$, $\nev B_1$ $\rightarrow$ $\nev C_1$
  $\models$ $B_0^{out}$ $\rightarrow$ $C_0^{out}$, $B_1^{out}$ $\rightarrow$ $C_0^{out}$
  \item $\ev B_1$ $\rightarrow$ $\nev C_1$, $\alw B_1$ $\rightarrow$ $\nev C_1$
  $\models$ $B_1^{out}$ $\rightarrow$ $C_0^{out}$
  \item $\ev B_1$ $\rightarrow$ $\nev C_1$, $\alw B_1$ $\rightarrow$ $\ev C_1$
  $\models \bot$
  \item $\ev B_1$ $\rightarrow$ $\alw C_1$ $\models$
  $\ev B_1$ $\rightarrow$ $\ev C_1$, $\alw B_1$ $\rightarrow$ $\ev C_1$,
  $\alw B_1$ $\rightarrow$ $\alw C_1$ (BR 3)
  \item $\ev B_1$ $\rightarrow$ $\ev C_1$
  $\models$ $\alw B_1$ $\rightarrow$ $\ev C_1$ (BR 4)
  \item $\ev B_1$ $\rightarrow$ $\ev C_1$, $\nev B_1$ $\rightarrow$ $\ev C_1$
  $\models$ $B_0^{out}$ $\rightarrow$ $C_1^{out}$,
  $B_1^{out}$ $\rightarrow$ $C_1^{out}$, $\ev B_1$ $\rightarrow$ $\alw C_1$
  \item $B_1^{out}$ $\rightarrow$ $C_1^{out}$  $\models$ $\ev B_1$ $\rightarrow$ $\ev C_1$,
  $\alw B_1$ $\rightarrow$ $\ev C_1$, $\alw B_1$ $\rightarrow$ $\alw C_1$\\
  (BR 4, BR 5 $\Leftarrow$ Proposition \ref{prop:KtKs} (\ref{prop:3}) (\ref{prop:4}))
  \item $B_1^{out}$ $\rightarrow$ $C_0^{out}$
  $\models$ $\alw B_1$ $\rightarrow$ $\nev C_1$ (BR 9 $\Leftarrow$ Proposition \ref{prop:KtKs} (\ref{prop:3}))
  \item $B_0^{out}$ $\rightarrow$ $C_1^{out}$ $\models$
  $\nev B_1$ $\rightarrow$ $\ev C_1$, $\nev B_1$ $\rightarrow$ $\alw C_1$\\
  (BR 1, 10 $\Leftarrow$ Proposition \ref{prop:KtKs} (\ref{prop:3}))
  \item $B_0^{out}$ $\rightarrow$ $C_0^{out}$
  $\models$ $\nev B_1$ $\rightarrow$ $\nev C_1$ (BR 6 $\Leftarrow$ Proposition \ref{prop:KtKs} (\ref{prop:3}))
\end{enumerate}

It remains to prove the rules of this stem base, which is straightforward from the definitions. I am giving some hints.

BR 1, 2, 3 and 4 are based on $\alw A \rightarrow \ev A$, $A \subseteq M$, and BR 11 and 14 on $\nev C_1 \wedge \ev C_1 \rightarrow \bot$. Hence, in BR 14 $\alw B_1 \rightarrow \bot$ follows in the underlying contexts. This implication has not been considered explicitly, but it is presupposed that other implications with $\alw B_1$ do not hold in this case. Their common extent is empty in the test context and $\emptyset' = \bot$ follows.

7.: By the third Armstrong rule \eqref{eq:premConclusion} and $\alw B_1 \rightarrow \ev B_1$, $\alw B_1 \rightarrow (\nev C_1 \wedge \ev C_1)$ follows from the presupposed implications. This is precisely the premise of rule 4 (BR 14).

10.: Inversely, in all possible cases the states / transitions have the attribute $\ev C_1$ and therefore also $\alw C_1$ and $C_1^{out}$. This means explicitly: $\top \rightarrow \ev C_1$, $\top \rightarrow  \alw C_1$, $\top \rightarrow C_1^{out}$. 5. is a parallel rule concerning $\nev C_1$. Rules 7. and 10. suggest that implications with empty premise $\top$ or with conclusion $\bot$ should be considered explicitly. 

11.: This rule has to be modified slightly for the first inference, according to a restriction of Proposition \ref{prop:KtKs} after this exploration had been finished:
\[m_1^{out} \rightarrow C_1^{out} \models \ev m_1 \rightarrow \ev C_1\: (m_1 = (e,1) \in M).\]

As a conclusion it can be stated that we have derived another set of rules (for Boolean attributes), which can shorten attribute exploration and narrow the decisions of an expert to really new implications. This rule set represents a sound and complete entailment calculus for the selected class of implications for transition and state contexts under the condition of Proposition \ref{prop:nevEv}:
\[ \forall (s^{in},s^{out}) \in R, \forall  m \in B_1\colon m \in (s^{out})^I \Rightarrow B_1 \subseteq (s^{out})^I,\]
in particular if $|B_1|=1$.



\section{Overview of the implemented R scripts}\label{sec:scripts}
\texttt{Concept Explorer} \cite{ConExp}\index{Concept Explorer} and \texttt{ConImp} \cite{Bur03} \index{ConImp} were used for attribute exploration. R scripts were programmed for the following tasks:
\begin{itemize}
 \item First, thanks to Mike Behrisch for generating all deterministic time series (including cycles) for 1 to 3 attributes, by his C++ program \texttt{AllinOne1.1.cpp}. The output files are stored as \texttt{inferenceRules/ergebnis2-3/ergebnis1.xt, .../ergebnis2.txt, .../ergebnis.3.txt}.
 \item \texttt{generateTestContext3.2.r} uses:
    \begin{itemize}
       \item The file \texttt{generateKttKtmpFct1.0.r} contains a function generating an apposition of $\Ktt$ (or $\Kt$) and $\Ktmp'$.
	\item The function \texttt{testImpFct2.0.r} decides about he validity of an implication in an arbitrary one-valued formal context.
  \end{itemize}
 Felix Steinbeck optimised the code so that the script terminated in yet acceptable time.
  \item \texttt{startPartContext.r}: R script for the serial computation of the test context.
	\item The output files were object reduced with \texttt{ConExp} and concatenated to \texttt{testContext3AttrImpForms\_noEmptyset3.2\_complete.txt}, a tabulator separated text file. \agrave noEmptyset" indicates that the empty premise or conclusion were considered explicitly and only nonempty sets were assigned to the variables $X, Y$ and $Z$.
  \item \texttt{testCxt.cxt}: The test context converted by \texttt{ConExp} to the Burmeister format readable by \texttt{ConImp}.
  \item \texttt{testContext3AttrImpForms\_noEmptyset3.2\_imp.txt}: Stem base (inference rules)
    as computed by \texttt{ConExp}.
  \item \texttt{testCxt.pro} (text file): Stem base as computed by \texttt{ConImp}, with short form of rules highlighting attributes (implicaton forms) not following from previous rules.
  \item \texttt{tCxtImp.duq}: Stem base stored in \texttt{ConImp} format.
\end{itemize}
For details, see the comments within the R files on the attached CD.

%

\chapter[Gene regulatory networks I: Analysis of a Boolean network\\from literature]{Gene regulatory networks I: Analysis of a Boolean network from literature}\label{ch:bSubtilis}
\section{Gene regulatory networks}\label{sec:geneRegProc}
\chaptermark{Analysis of a Boolean network from literature}
\begin{figure} 
\centering 
\includegraphics[width=8cm]{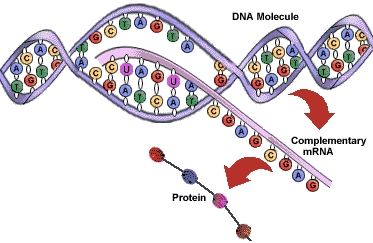} 
\caption{Gene expression: the two phases of transcription and translation  [www.scientificpsychic.com/fitness/aminoacids1.html].}\label{fig:geneExpr} 
\end{figure} 
To understand normal and destructive cellular reactions, systems biology has developed models describing processes at the molecular level. A fundamental process is \textit{gene expression},\index{Process!gene regulatory} which is a sequence of two phases, \textit{transcription} and \textit{translation} (Figure \ref{fig:geneExpr}): i) during 
\textit{transcription}, messenger ribonucleic acid (mRNA) is produced according to the genetic DNA template (i.e., a DNA sequence coding for a single protein); and ii) during \textit{translation}, proteins are built from amino acids using these mRNA templates. \cite{Alb08}

Proteins are the main regulators of living organisms; they activate almost every chemical reaction 
within an organism as enzymes, build new structures during cell division or transduce biochemical 
signals from the cell surface to the cytoplasm and the nucleus, e.g., by phosphorylation of target 
proteins. These signals can activate another class of proteins, the \textit{transcription factors}, which 
bind to the DNA and thus are able to initiate, enhance or repress transcription. \cite{Alb08} 

Mathematical and computational models may assist biologists in further research activities by generating predictions and hypotheses that can be experimentally tested. Network models,\index{Gene regulatory network} generated on the basis of extracted information and/or experimental data, facilitate the analysis of interactions among different key molecules and provide new insight into complex biological pathways and interactions (for an overview of methods see \cite{Schl07} and \cite{Hec09}).

In the present chapter, the proposed method is applied to the analysis of a gene regulatory network assembled from literature and database information in \cite{deJ04} and transformed to a Petri net\label{Petri net}\index{Petri net} as well as a Boolean network in \cite{Ste07}.

\section{Sporulation in \textit{Bacillus subtilis}}\label{sec:sporulation}
%
\textit{B. subtilis} is a gram positive soil bacterium. Under extreme environmental stress, it
produces a single endospore, which can survive ultraviolet or gamma radiation, acid, hours of
boiling or long periods of starvation. The bacterium leaves the vegetative growth phase in favour
of a dramatically changed and highly energy consuming behaviour, and it dies at the end of the
sporulation process. This corresponds to a switch between two clearly distinguished genetic
programs, which are complex but comparatively well understood.

By literature and database search, de Jong et al.~\cite{deJ04} identified 12 main regulators,
constructed a model of piecewise linear differential equations and obtained realistic simulation
results. An exogenous signal (starvation) triggers the phosphorylation of the transcription factor
Spo0A to Spo0AP by the kinase KinA. This process is reversible by the phosphatase Spo0E. Spo0AP is
necessary to transcribe SigF, which regulates genes initiating sporulation and therefore is an
output of the model. The interplay with other transcription factors AbrB, Hpr, SigA, SigF, SigH and
SinR is graphically represented in \cite[Figure 3]{deJ04}. SinI inactivates SinR by binding to it.
SigA and Signal are considered as an input to the model and are always on. Table \ref{BoolEq} lists
the Boolean equations building the model in \cite{Ste07}. They exhibit a certain degree of nondeterminism,\index{Process!nondeterministic} since the functions for the \textit{off} fluents sometimes are
not the negation of the \textit{on} functions. This accounts for incomplete or inconsistent
knowledge. In the case of state transitions determined by $k$ conflicting function pairs, $2^k$ output states were generated.\index{State!output}

\begin{table}\index{Boolean network}
\centering
\begin{tabular}{|lcl|}
\hline
&\tiny{ }&\\[-3mm]
AbrB &= &SigA $\overline{\text{AbrB}}$ $\overline{\text{Spo0AP}}$\\
$\overline{\text{AbrB}}$ &= &$\overline{\text{SigA}}$ + AbrB + Spo0AP\\
SigF &= &(SigH Spo0AP $\overline{\text{SinR}})$ + (SigH Spo0AP SinI)\\
$\overline{\text{SigF}}$ &= &(SinR $\overline{\text{SinI}}$) + $\overline{\text{SigH}}$ + $\overline{\text{Spo0AP}}$\\
KinA &= &SigH $\overline{\text{Spo0AP}}$\\
$\overline{\text{KinA}}$ &= &$\overline{\text{SigH}}$ + Spo0AP\\
Spo0A &= &(SigH $\overline{\text{Spo0AP}}$) + (SigA $\overline{\text{Spo0AP}}$)\\
$\overline{\text{Spo0A}}$ &= &($\overline{\text{SigA}}$ SinR $\overline{\text{SinI}}$) +
($\overline{\text{SigH}}$ $\overline{\text{SigA}}$ ) + Spo0AP\\
Spo0AP &= &Signal Spo0A $\overline{\text{Spo0E}}$ KinA\\
$\overline{\text{Spo0AP}}$ &= &$\overline{\text{Signal}}$ + $\overline{\text{Spo0A}}$ + Spo0E + $\overline{\text{KinA}}$\\
Spo0E &= &SigA $\overline{\text{AbrB}}$\\
$\overline{\text{Spo0E}}$ &= &$\overline{\text{SigA}}$ + AbrB\\
SigH &= &SigA $\overline{\text{AbrB}}$\\
$\overline{\text{SigH}}$ &= &$\overline{\text{SigA}}$ + AbrB\\
Hpr &= &SigA AbrB $\overline{\text{Spo0AP}}$\\
$\overline{\text{Hpr}}$ &= &$\overline{\text{SigA}}$ + $\overline{\text{AbrB}}$ + Spo0AP\\
SinR &= &(SigA $\overline{\text{AbrB}}$ $\overline{\text{Hpr}}$ $\overline{\text{SinR}}$ $\overline{\text{SinI}}$ Spo0AP) +\\
&&(SigA $\overline{\text{AbrB}}$ $\overline{\text{Hpr}}$ SinR SinI Spo0AP)\\
$\overline{\text{SinR}}$ &= &$\overline{\text{SigA}}$ + AbrB + Hpr + (SinR $\overline{\text{SinI}}$) + ($\overline{\text{SinR}}$ SinI) + $\overline{\text{Spo0AP}}$)\\
SinI &= &SinR\\
$\overline{\text{SinI}}$ &= &$\overline{\text{SinR}}$\\
SigA &= &TRUE (input to the model)\\
Signal &= &TRUE or FALSE (constant, depending on the initial state)\\
\hline
\end{tabular}
\caption{Boolean rules for the nutritional stress response regulatory network, derived in
\cite{Ste07} from \cite{deJ04}. $\overline{x} \hat{=} \neg x,\: x+y\, \hat{=}\, x \vee y,\: xy\, \hat{=}\,x \wedge y$.}
\label{BoolEq}
\glossary{name={$\overline{x}$},description={Negation of the proposition $x$},sort=$.Neg$}
\end{table}

Since the validation of the model by data and experimental literature has been done before in \cite{deJ04}, we analysed pure knowledge-based simulations. For that reason, in step \ref{item:compObs} of the protocol (p.~\pageref{protocol}), the stem base is computed automatically without further confirmation by an expert. The developed \texttt{R} \cite{R11} scripts for simulation, for generating the transitive contexts and for converting a stem base into \texttt{Prolog}\index{Prolog} format are indicated in Section \ref{sec:RgeneNetsII}.

\section{Simulation starting from a state typical for the vegetative growth phase}\label{simNoStress}
\sectionmark{Simulation of the vegetative growth phase}
We performed supplementary analyses of the transitions starting from a typical state without the
starvation signal \cite[Table 4]{Ste07}. 
\begin{table}\centering
\begin{tabular}{|c||c|c|c|c|c|c|c||c|c|c|c|c|c|c|}
\hline
\textbf{Transition} &\begin{sideways}KinA$^{in}$\end{sideways}
&\begin{sideways}Spo0A$^{in}$\end{sideways}
&\begin{sideways}Spo0AP$^{in}$\end{sideways}
&\begin{sideways}AbrB$^{in}$\end{sideways}
&\begin{sideways}Spo0E$^{in}$\end{sideways}
&\begin{sideways}SigH$^{in}$\end{sideways}
&\begin{sideways}Hpr$^{in}$\end{sideways}
&\begin{sideways}KinA$^{out}$\end{sideways}
&\begin{sideways}Spo0A$^{out}$\end{sideways}
&\begin{sideways}Spo0AP$^{out}$\end{sideways}
&\begin{sideways}AbrB$^{out}$\end{sideways}
&\begin{sideways}Spo0E$^{out}$\end{sideways}
&\begin{sideways}SigH$^{out}$\end{sideways}
&\begin{sideways}Hpr$^{out}$\end{sideways}\\
\hline \hline
$(s_0^{in}, s_1^{out})$ &-&+&-&-&-&-&+&-&+&-&+&+&+&-\\ \hline
$(s_0^{in}, s_2^{out})$ &-&+&-&-&-&-&+&+&+&-&-&-&-&+\\ \hline
$(s_1^{in}, s_1^{out})$ &-&+&-&+&+&+&-&-&+&-&+&+&+&-\\ \hline
$(s_1^{in}, s_2^{out})$ &-&+&-&+&+&+&-&+&+&-&-&-&-&+\\ \hline
$(s_2^{in}, s_1^{out})$ &+&+&-&-&-&-&+&-&+&-&+&+&+&-\\ \hline
$(s_2^{in}, s_2^{out})$ &+&+&-&-&-&-&+&+&+&-&-&-&-&+\\
\hline
\end{tabular}
\caption{The (many-valued) transitive context corresponding to a simulation starting from a \textit{B.
subtilis} state without nutritional stress \cite[Table 4]{Ste07}. +:~on, -:~off. Only the attributes are listed that are changing during the simulation as well as Spo0A and Spo0AP.}
\label{tab:subtilisTransCxt}
\end{table}

The concept lattice for the resulting transitive context (Table \ref{tab:subtilisTransCxt}, with only a part of the attribute set only) is shown in Figure \ref{fig:noStress}. 
\begin{figure}
  \centering
  \includegraphics[width=14cm]{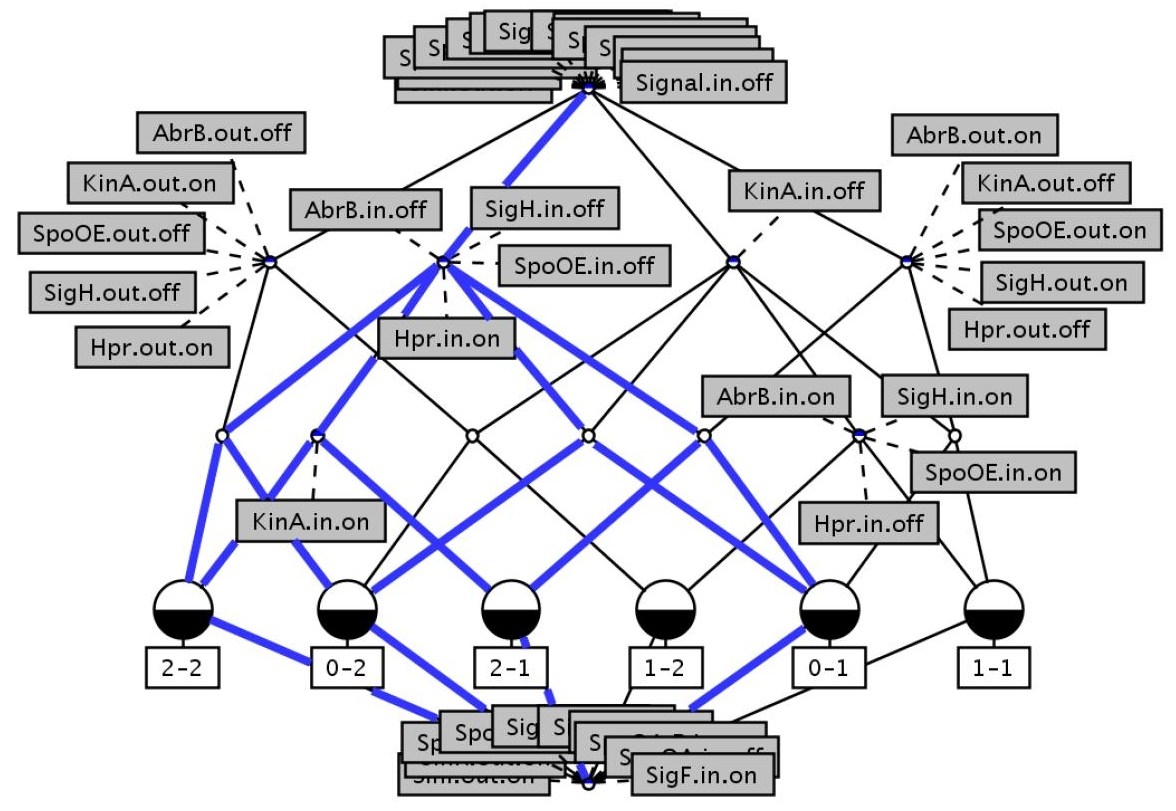}
  \caption{Concept lattice (computed and drawn with \texttt{Concept Explorer}\index{Concept Explorer} \cite{ConExp}) representing a simulation without nutritional stress.
\textit{Signal}: starvation; \textit{AbrB, Hpr, SigA, SigF, SigH, SinR, Spo0A} (phosporylated form
\textit{Spo0AP}): transcription factors -- \textit{SigF} initiates sporulation; \textit{KinA}: kinase; \textit{Spo0E}: phosphatase;
\textit{SinI} inactivates \textit{SinR} by binding to it. $i$-$j$ indicates a transition
($s_i^{in}, s_j^{out})$. \textit{Bold / blue lines}: Filter (superconcepts)\index{Filter (lattice)} and ideal
(subconcepts)\index{Ideal (lattice)} of the concept (\{($s_0^{in}, s_1^{out}), (s_0^{in},
s_2^{out}), (s_2^{in}, s_1^{out}), (s_2^{in}, s_2^{out})$\},
\{AbrB.in.off, SigH.in.off, Spo0E.in.off, Hpr.in.on\})} \label{fig:noStress}
\end{figure}
The larger circles at the bottom represent object concepts; their
extents are the four single transitions with the input state at $t=0$ or $t=2$, and the intents are
all attributes above a concept. Thus for instance, the two latter transitions have the attribute
KinA.in.on in common, designating the respective concept. Moreover, they are distinguished
unambigously from other sets of transitions by this  attribute -- the concept is generated
by \agrave KinA.in.on".

Implications\label{impSubtilisGrowth} of the stem base can be read from the lattice. For instance there are implications
between the generators of a concept:
\begin{equation}
<4> \text{ AbrB.in.off $\rightarrow$ SigH.in.off, Spo0E.in.off, Hpr.in.on}
\end{equation}
Analogous implications hold for the attributes of the conclusion, and there are implications
between attributes of sub- and superconcepts. $< 4>$ indicates that the rule is supported by four
transitions.

The bottom concept has an empty extent. Its intent is the set of attributes never occuring during
this simulation comprising only three time steps. The top concept does not have an empty intent -- as it is often the case --,
but it consists of 10 attributes common to all 6 transitions. The corresponding rule has an empty
body ($\top$):
\begin{equation}
\begin{aligned}
< 6 >\: \top \rightarrow\: &\text{Signal.in.off, SigA.in.on, SigF.in.off, Spo0A.in.on, Spo0AP.in.off,}\\
&\text{SinR.in.off, SinI.in.off, Signal.out.off, SigA.out.on, SigF.out.off, }\\
&\text{Spo0A.out.on, Spo0AP.out.off, SinR.out.off, SinI.out.off}
\end{aligned}
\end{equation}

The transitive context generated by a simulation following nutritional stress has about
20 transitions, 500 concepts and 50 implications. In such a case it is more convenient to query the
implicational knowledge base. But also for the visualisation of large concept hierarchies, there
exist more sophisticated tools like the ToscanaJ suite \cite{Tos}.

\section{Analysis of all possible transitions}\label{subseq:allTrans}
In order to analyse the dynamics of the \textit{B. subtilis} network exhaustively, I generated
4224 transitions from all possible $2^{12}=4096$ initial states (thus the rules are nearly
deterministic). There were 11.700 transitions in the transitive context, from which the
stem base was computed containing 524 implications with support $>$ 0, but $11.023.494 \approx 2^{24}$ concepts.

It was not feasible to provide biological evidence for a larger part of the implications, within
the scope of this methodological study. This could be done by literature search, especially
automatic text mining, by new specific experiments or by comparison with observed time series (\cite[3.2]{Wol07}, Chapter \ref{ch:geneRegNets}). Instead examples will be given for classes of implications that can be validated or falsified during attribute
exploration in specific ways. I start with the examples of \cite[4.3]{Ste07}.
\begin{itemize}
  \item \agrave For example, we know that in the absence of nutritional stress, sporulation should never be initiated \cite{deJ04}. We can use \textit{model checking}\index{Model!checking} to show this holds in our model by proving that no reachable state exists with SigF = 1 starting from any initial state in which Signal = 0,
SigF = 0 and Spo0AP = 0." \cite[p. 341]{Ste07} This is equivalent to the rule following from the stem
base:
\begin{equation}\label{imp:SigF.off}
  \text{Signal.in.off, SigF.in.off, Spo0AP.in.off} \rightarrow \text{SigF.out.off}
\end{equation}
  \item SigF.out.on $\rightarrow$ KinA.out.off, Spo0A.out.off, Hpr.out.off, AbrB.out.off:\\
Spo0AP is reported to activate the production of SigF but also to repress its own production (mutual
exclusion). \cite{deJ04}
  \item SigH.out.off $\rightarrow$ AbrB.out.off, Spo0E.out.off, SinR.out.off, SinI.out.off\\
All these genes are regulated $\overline{gene.out} =  \overline{\text{SigA.in}} + \text{AbrB.in}$ (+ ...).\\
\end{itemize}
\vspace{-4mm}

In our approach, such dependencies and mutual exclusions can be checked systematically. We searched
the stem base for further interesting and simple implications:
\begin{align}
&< 4500 > \text{Spo0AP.in.on, KinA.out.off} \rightarrow \text{Hpr.out.off}\\
&< 4212 > \text{SigH.in.on. KinA.out.off} \rightarrow \text{Hpr.out.off}\\
&< 3972 > \text{AbrB.in.off, KinA.out.off} \rightarrow \text{Hpr.out.off}
\end{align}
$\overline{\text{Hpr}}$ and $\overline{\text{KinA}}$ are determined by different Boolean functions,
but they are coregulated in all states emerging from any input state characterised by the single
attributes Spo0AP.on, SigH.on or AbrB.on.
\begin{equation}\label{eq:mainRule}
\begin{aligned}
<3904> \text{AbrB.out.on } \rightarrow &\text{ SigA.in.on,
SigA.out.on, SigF.out.off,}\\
&\text{ Spo0A.out.on, Spo0E.out.on, SigH.out.on,}\\
&\text{ Hpr.out.off, SinR.out.off, SinI.out.off}
\end{aligned}
\end{equation}
AbrB is an important \agrave marker" for the regulation of many genes, which is understandable from the
Boolean rules with hindsight. By a PubMed query, a confirmation was found for the downregulation of
SigF together with the upregulation of AbrB
\cite{Tom03}.

Finally, we entered sets of interesting attributes into the Prolog\index{Prolog} knowledge base, such
that a derived implication was computed, and accordingly the closure of the attribute set. In order to avoid infinite regresses, tabled resolution has been applied, which is implemented in the Prolog extension XSB [\texttt{http://xsb.sourceforge.net}]. Complementary to (\ref{imp:SigF.off}), we searched after conditions for an eventual switch towards sporulation (SigF.out.on). If the premise is entered as set of facts \texttt{off(sigF.in), off(Spo0AP.in)} and \texttt{on(sigF.out)}, then with the implications of the stem base the conclusion is returned as a set of derived facts. For the analysed complete simulation, the conditions Signal = SigA = TRUE had been dropped, but they were supposed to be constant (see the first four implications, which in Prolog are read from right to left).
\begin{verbatim}
      :- table off/1.       
      :- table on/1.        

      off(sigF.in).      
      off(Spo0AP.in).
      on(sigF.out).       

      on(signal.in)   :- on(signal.out).
      on(signal.out)  :- on(signal.in).
      on(sigA.in)     :- on(sigA.out).
      on(sigA.out)    :- on(sigA.in).
      off(abrB.out)   :- on(sigF.out).
      off(kinA.out)   :- on(sigF.out).
      off(Spo0AP.out) :- off(spo0A.in), on(sigF.out).
      on(signal)      :- off(sigH.in), on(sigF.out).
      ...
\end{verbatim}

The subsequent implication was found:
\begin{equation}
\begin{aligned}
&\text{SigF.in.off, Spo0AP.in.off, SigF.out.on }\\
\rightarrow\: &\text{Signal.in.on, Signal.out.on, SigA.in.on, SigA.out.on, Spo0AP.out.off,}\\
&\text{Spo0A.out.off, AbrB.out.off, KinA.out.off, Hpr.out.off.}
\end{aligned}
\end{equation}
The latter four attributes follow immediately from the Boolean rules, but Spo0AP.out.off depends in
a more complex manner on the premises. It is also noteworthy that the class of input states
developing to a final state with attribute SigF.out.on is characterised only by the attributes
Signal.in.on and the ubiquitous transcription factor SigA.in.on, i.e. the initial presence or absence of no other gene is necessary for the initiation of sporulation.

\chapter{Gene regulatory networks II: Adapted Boolean network models for extracellular matrix formation}\label{ch:geneRegNets}
\chaptermark{Boolean network models for ECM}
\section{Biomedical and bioinformatics background}\label{sec:ECMBg}
Rheumatoid arthritis (RA)\glossary{name=RA,description={Rheumatoid arthritis}} is a systemic (whole body) autoimmune disease, the causes of which are incompletely elucidated so far. Autoantibodies are directed against citrullinated proteins, where one or several arginine amino acids are post-translationally modified. Susceptibility factors are the genotype of the antigen presenting major histocompatibility complex HLA-DRB1, smoking, and there is a higher risk for women pointing at hormonal factors. The immune disorder might also be triggered by a viral or bacterial infection. \cite{Weg10}

Our focus is a better understanding of the symptom generating processes and the therapy against the cytokine TNF$\alpha$. RA is characterised by chronic inflammation and destruction of multiple joints perpetuated by the synovial membrane (SM)\glossary{name={SM}, description={Synovial membrane}}. A major component of the inflamed SM (also called pannus tissue) are activated, semi-transformed synovial fibroblasts (SFB, or synoviocytes)\glossary{name={SFB}, description={Synovial fibroblast cell}} \cite{Abe06} \cite{Hub06} \cite{Kar06} \cite{Kin95} \cite{Smo03}. In normal joints, SFB show a balanced expression of proteins, regulating the formation and degradation of the extracellular matrix (ECM)\glossary{name={ECM}, description={Extracellular matrix}}, a fibrous structure providing support to the cells (besides other functions). In RA, however, SFB are known for predominant expression and secretion of pro-inflammatory cytokines and tissue-degrading enzymes \cite{Kar06} \cite{Kin95}, thus maintaining joint inflammation and degradation of ECM components of cartilage and bone, which are also invaded by SFB. In addition, enhanced formation of soft ECM components such as collagens in the affected joints (an attempt of wound healing resulting in fibrosis) is also driven by SFB \cite{Pos92}. 

\begin{figure} 
\begin{center} 
\includegraphics[width=11cm]{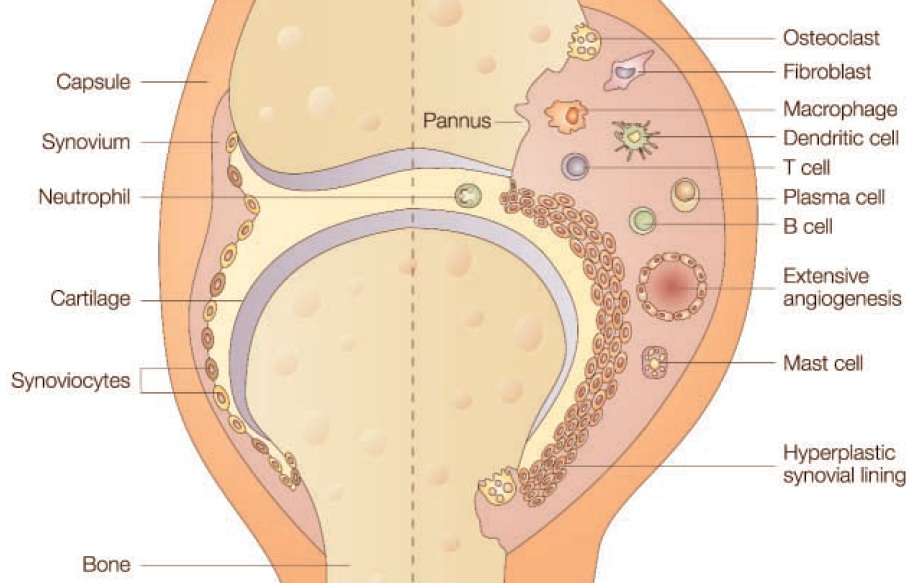} 
\end{center}\caption{The knee joint: normal morphology (left) and schematic representation of rheumatoid arthritis features (right) \cite[Figure 1]{Smo03}.}\label{fig:joint} 
\end{figure} 
Central transcription factors (TFs)\glossary{name={TF}, description={Transcription factor}} involved as key players in RA pathogenesis \cite{Fir99} \cite{Yam01} and in the activation of SFB in RA patients are AP‑1, NF-$\kappa$B, ETS‑1, and SMADs \cite{Ang91} \cite{Asa97} \cite{Fir99} \cite{Han98}. These TFs show binding activity for their cognate recognition sites in the promoters of inflammation-related cytokines (e.g., TNF$\alpha$,\glossary{name={TNF$\alpha$}, description={Tumor necrosis factor alpha}} IL1$\beta$, IL6 \cite{Abe06}) and matrix-degrading target genes \cite{Fir99} \cite{Han01} \cite{New97} \cite{Whi95} \cite{Yam01}, e.g., collagenase (MMP‑1) \cite{Abe06} and stromelysin1 (MMP‑3) \cite{New97}. The latter two show high expression levels in RA \cite{Fir91} \cite{Gra91} \cite{McC91} and contribute to tissue degradation \cite{Wer89} by destruction of ECM components, including aggrecan or collagen types IV, X, and XI \cite{Wel81} \cite{WuJ91}. 

Mathematical and computational models are of particular importance in the context of rheumatic diseases and cartilage/bone metabolism, since the development of new and/or adapted molecular therapies depends on the understanding of superordinate pathway interrelationships in the pro-inflammatory microenvironment of the joint \cite{Kar06}. Therefore, we developed a method for simulating the temporal behaviour of regulatory and signalling networks. It was used to create two gene regulatory networks emulating ECM formation and destruction, based on literature information about SFB on the one hand and on experimental data on the other, which we obtained from TGF$\beta$1 or TNF$\alpha$ stimulated SFB. As motivated in Section \ref{sec:BooleNets}, we applied Boolean network architecture for modelling. Using attribute exploration, the simulation results and the observed time series were further integrated in a fine-tuned and automated manner resulting in sets of rules that determine system dynamics.

For our analysis we used a collection of 18 genes, which can be classified into five functional groups, sufficient to create a self-contained regulatory network of ECM maintenance: (1) structural proteins which are the target molecules (i.e., the collagen-forming subunits COL1A1 and COL1A2); (2) enzymes degrading them (i.e., the matrix metalloproteinases MMP1, -3, ‑9, and ‑13); (3) molecules that inhibit these proteases (tissue inhibitor of metalloproteinases TIMP1); (4) TFs (i.e., ETS1, FOS, JUN, JUNB, JUND, NKFB1) and modulators acting as TFs (i.e., SMAD3, SMAD4, SMAD7) which are regulated by (5) the external signalling molecules TNF$\alpha$ (TNF)\glossary{name={TNF$\alpha$}, description={Tumor necrosis factor alpha}} and TGF$\beta$1 (TGFB1).\glossary{name={TGF$\beta$1}, description={Transforming growth factor beta I}} These genes (Figure \ref{fig:geneList}) are known to be expressed by SFB, except for TNF and MMP9, for which the expression is still under question (compare p. \pageref{noTNF}, \pageref{noMMP9}).
\begin{figure}
\begin{center}
\includegraphics[width=151mm]{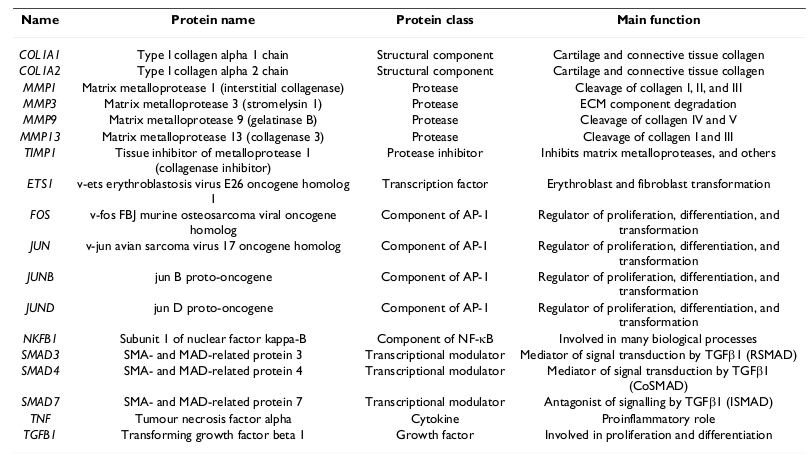}
\end{center}
\vspace{-5mm}
\caption{List of genes used in this analysis.}
\label{fig:geneList}
\end{figure}

\section{Methods}\label{sec:methods}
\subsection{Clinical data} 
\subsubsection{Patients and samples} 
Synovial membrane samples were obtained following tissue excision upon joint replacement/synovectomy from RA and osteoarthritis (OA)\glossary{name=OA,description={Osteoarthritis}} patients ($n$ = 3 each; Figure \ref{fig:clinicalChar}). Informed patient consent was obtained and the study was approved by the ethics committees of the respective university. RA patients were classified according to the American College of Rheumatology (ACR) criteria \cite[p. 234]{Arn88}, OA patients according to the respective criteria for osteoarthritis \cite[p. 237]{Alt86}. 

The preparation of primary SFB from RA and OA patients was performed as previously described \cite[p. 235]{Zim01}. Briefly, the tissue samples were minced and digested with trypsin/collagenase~P. The resulting single cell suspension was cultured for seven days. Non-adherent cells were removed by medium exchange. SFB were then negatively clarified using Dynabeads$^\text{\textregistered}$ M‑450 CD14 and subsequently cultured over 4 passages in DMEM containing 100 $\mu$g/ml gentamycin, 100 $\mu$g/ml penicillin/streptomycin, 20 mM HEPES (all from PAA Laboratories, Coelbe, Germany), and 10 \% FCS (BioWhittaker-Lonza, Basel, Switzerland).
\begin{figure}[t]
 \centering
 \includegraphics[width=151mm]{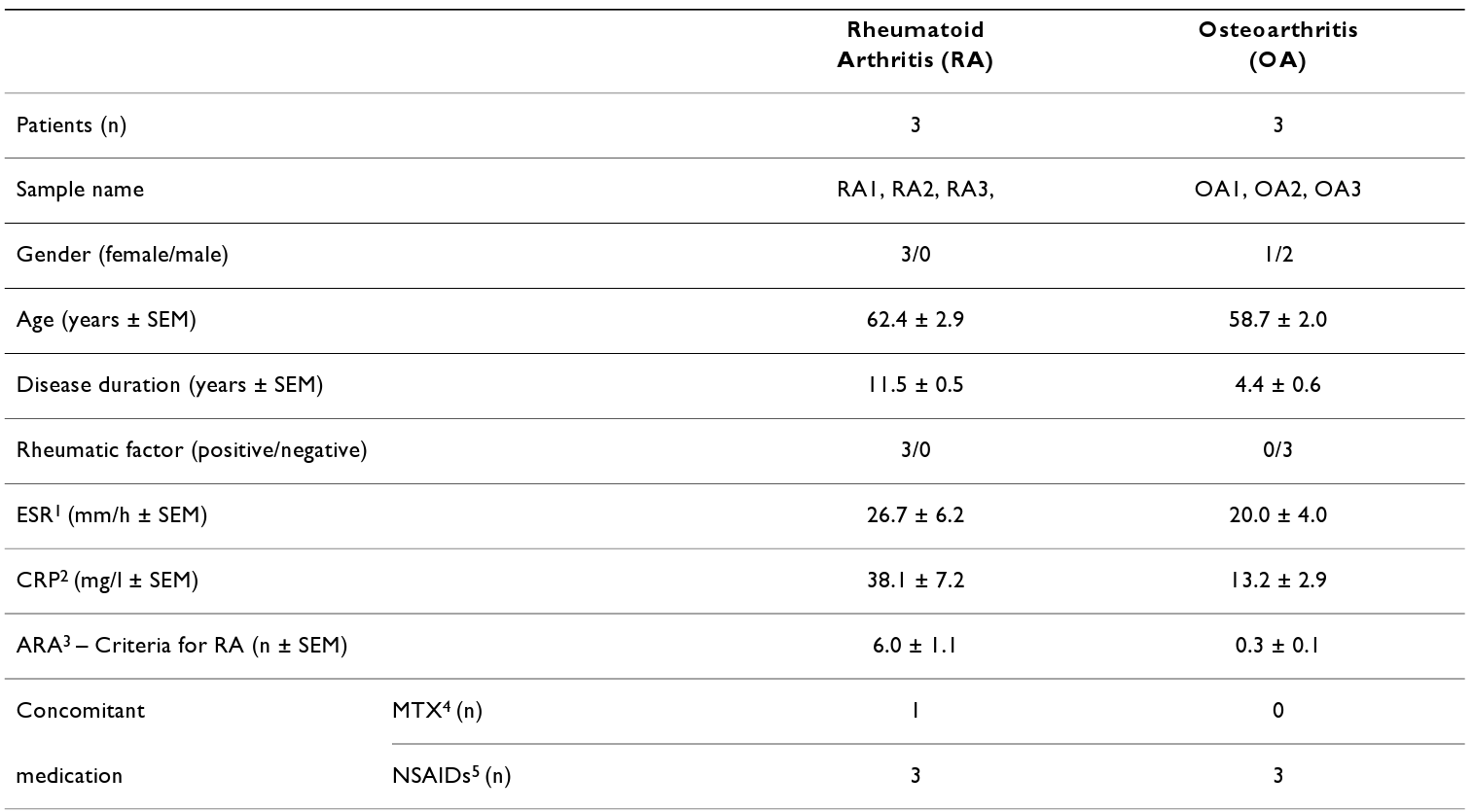}
\caption{Clinical characteristics of the patients at the time of synovectomy/sampling.
$^1$Erythrocyte sedimentation rate, 
$^2$C-reactive protein, normal range: $<$ 5 mg/l, 
$^3$American Rheumatism Association (now: American College of Rheumatology) 
$^4$Methotrexate 
$^5$Non steroidal anti-inflammatory drugs.}
 \label{fig:clinicalChar}
\end{figure}

\subsubsection{Cell stimulation and isolation of total RNA} 
At the end of the fourth passage, SFB were stimulated by 10 ng/ml TGF$\beta$1 or TNF$\alpha$ (PeproTech, Hamburg, Germany) in serum-free DMEM for 0, 1, 2, 4, and 12 h. At each time point, the medium was removed and the cells were harvested after treatment with trypsin (0.25\% in versene; Invitrogen, Karlsruhe, Germany). After washing with PBS, they were lysed with RLT buffer (Qiagen, Hilden, Germany) and frozen at ‑70°C. Total RNA was isolated using the RNeasy Kit (Qiagen) according to the supplier's recommendation. 

\subsubsection{Microarray data analysis} 
Analysis of gene expression was performed using U133 Plus 2.0 RNA microarrays (Affymetrix$^\text{\textregistered}$, Santa Clara, CA, USA). Labelling of RNA probes, hybridisation, and washing were carried out according to the supplier's instructions. Microarrays were analysed by laser scanning (Hewlett-Packard Gene Scanner). Background-corrected signal intensities were determined and normalised using the MAS 5.0 software (Affymetrix$^\text{\textregistered}$). For this purpose, arrays were grouped according to the respective stimulus (TGF$\beta$1 and TNF$\alpha$, $n=6$ each). The arrays in each group were normalised using quantile normalisation \cite[p. 236]{Bol03}. Original data from microarray analysis have been deposited in NCBI Gene Expression Omnibus \cite{GEO} and are accessible through GEO series accession number GSE13837. A list of probe sets and all expression time courses are provided in additional file 4. 

\subsection{Creating network and Boolean functions}
For the selection of genes and proteins involved in ECM maintenance and for network generation,\index{Gene regulatory network} Boolean queries were performed in PubMed \cite{PubMed}. Articles were selected containing information about relevant genes expressed in SFB and involved in ECM maintenance. For information extraction, the abstracts were screened and filtered manually for statements on healthy conditions only. This knowledge-based collection yielded the set of gene candidates analysed in detail. The final gene list is presented in Figure \ref{fig:geneList}.

The genes were also analysed using Bibliosphere \cite{Bib} and literature not extracted from PubMed was added. Subsequently, information concerning regulatory relationships was collected and transformed into short statements serving as input relations (edges) for the network building program Cytoscape, version 2.6.0 \cite{Sha03}. Contradictory literature information was resolved by preferring facts applying to the target cell type (human fibroblasts) and/or by comparison with experimental gene expression results from our and other microarrray data (GSE1742 and GSE2624, see additional file 5). The complete list of used statements and the respective literature basis can be found in additional file 1. In a further step, simulation results were iteratively compared to the experimental data in the present study, resulting in two adapted Boolean networks\index{Boolean network} which represent hypotheses about regulatory processes initiated by TGF$\beta$1 and TNF$\alpha$. 

\subsection{Data discretisation}\label{sec:discretisation}\index{Discretisation}
Since we were interested in regulatory interactions, the fold-change of the expression values was more important than absolute levels. Hence, we discretised individual time series separately. The discretised data served to verify or falsify the temporal dependencies predicted from the extracted literature knowledge. For that reason, we wanted to conserve as many effects on gene expression as possible and set weak criteria for up-regulation: if the highest fold-change (i.e., the difference of $\operatorname{log}_2$ values) between two arbitrary time points was larger than 1, then the time profile was discretised to 0 or 1 by k-means clustering (100 iterations, vote of 25 restarts). We set the constant value 0 if: (i) the highest fold-change between two arbitrary points in a time series was less than 1; (ii) the absolute expression value was below the threshold of 100 for one probe set; or (iii) the Affymetrix detection value $p$ indicating the reliability of the measurement exceeded 0.05. In all other cases, the constant was set to 1. Applying these criteria, also individual values were set to 0 (i.e., off) following clustering. 

\subsection{Principles of simulation}\label{sec:principlesSim}
Using the deterministic Boolean network, simulations were generated using an
asynchronous update scheme based on the subsequent biologically-founded assumptions.
In order to simulate the time courses more realistically, transcription and translation were
separated, i.e., the left side of a Boolean function (output) was considered as mRNA and the right side as TF and/or stimulus (input). Unfortunately, time-resolved data for gene expression events, mRNA, or protein half-life are scarce in the literature. Therefore, time steps were selected based on general expert knowledge and comparison of literature and experimental data, if available. For example, the duration of transcription was generally set to 1 time unit, for NF-$\kappa$B it was set to a doubled time period, reflecting its markedly
prolonged response time before expression compared to the immediate early response transcription factors AP-1 and ETS1 \cite{Bou90}.

In summary, we selected the time steps as follows: transcription 1 (NFKB1: 2),
translation: 1, mRNA lifespan: 1, and protein lifespan: 2. Since TGF$\beta$1 and TNF$\alpha$ have to
be released into the extracellular medium after translation, they were assumed to take
effect three time units after induction. The starting conditions of the simulations were characterised by the initially observed, discretised states, and an initial state was introduced, for which the TFs were set to on.
Supposing a steady state\index{State!steady} situation before starting the stimulation with TGF$\beta$1 and TNF$\alpha$,
the protein levels at step 0 and 1 were defined according to that of the corresponding
mRNA, and, in addition, the respective stimulating protein was set to on. The simulations
were performed over twelve time units, roughly corresponding to the twelve hour duration
of the gene expression experiments.

\subsection{Creation of a temporal rule knowledge base}\label{sec:KBMethods}
The sets of observed and simulated states $S^{obs}$ and $S^{sim}$ were characterised by the expression levels of each gene, i.e., by a subset of attributes $M = E \times F$, with entities or genes $E$, and the corresponding values $F = \{\text{off, on}\}$. Hence, they were assembled into state contexts (Definition \ref{def:stateCxt1}) $\Ks^{obs}$ and $\Ks^{sim}$. A state can be considered as a tuple $(f_1,..., f_n)$ with $f_i \in F, n = |E|$. 

The transitions after one time step define relations $R^{obs} \subseteq S^{obs} \times S^{obs}$ and $R^{sim} \subseteq S^{sim} \times S^{sim}$ on the states. Thus, in general multiple output states $s^{out}$\index{State!output} following an input state $s^{in}$ are possible. However, this case rarely occurred, justifying the use of a deterministic simulation procedure.\index{Process!deterministic}
 
We computed the transitive closure of these relations, since we were interested in all states emerging from a given one, within the observation or simulation time. The data of all time series related to one stimulus was assembled in the transitive contexts (Definition \ref{def:transCxt}) $\Ktt^{obs}$ and $\Ktt^{sim}$. These define relations $I$ between objects (the transitions) and attributes (the discretised gene expression levels in input and output states).

By attribute exploration, we compared the literature-based implications with those merely derived from the data and applied a strong criterion: implications of $\Ktt^{sim}$ had to be valid for all transitions of the observed context $\Ktt^{obs}$. This is equivalent to an exploration of the union of the two contexts, where every proposed implication is accepted. Thus, the resulting stem base was computed automatically with the Java tool \texttt{Concept Explorer}, which supports also expert centered attribute exploration \cite{ConExp}. The other calculations were made with my own \texttt{R} \cite{R11} programs (Section \ref{sec:RgeneNetsII}). Computing the
2713 (8785) TGF$\beta$1 (TNF$\alpha$) rules, Concept Explorer ran 21 (30) minutes on a 2.66 GHz/2 GB computer.

\subsection{Expert analysis of transition rules}
The calculated transition rules were screened manually, focussing on the appearance and the temporal behaviour
of the following features: (i) constitutive vs. induced gene expression; (ii) co-expression vs. divergent expression of mediators, TFs, and target genes; (iii) expression of mediators/transcription factors vs. expression of target genes; (iv) regulation of target gene expression based on the expression of different transcription factors; (v) expression of individual genes vs. expression of their functional groups; and (vi) discrepancies to the literature. Subsequently, the extracted rules were assessed with respect to
biological coherence and relevance.

\subsection{Overview of the implemented R scripts}\label{sec:RgeneNetsII}
The following \texttt{R} scripts have been developed; most of them were also utilised for the \textit{B. Subtilis} network analysis of Chapter \ref{ch:bSubtilis}.
\begin{itemize}
 \item \texttt{discretise4.1.r}.\\
Data discretisation as described above.
 \item \texttt{simulationTnf.r, simulationTgf.r}.\\
  Simulations according to the Boolean networks.\\
  (\texttt{allStates.r}: Simulations for the nondeterministic \textit{B. subtilis} network, starting from all possible initial states.)
 \item \texttt{aposterioriTrans2.0.r}.\\
Generate a transition and a transitive context from observed or simulated data. 
 \item \texttt{selectGenes2.1.r}.\\
Compute support and confidence of a rule in an (observed) transition context, show transitions in the observed and simulated context with the premise attributes.\\
Provides a decision criterion for expert attribute exploration.
 \item \texttt{convert2Prolog.r}.\\
Convert a stem base into Prolog format for queries.
\end{itemize}

\subsection{Additional files}
In the data CD, the \texttt{R} scripts are included as well as the additional files published with \cite{Wol09}, at \texttt{http://www.biomedcentral.com/content/supplementary/}:
\begin{itemize}
\item \texttt{1752-0509-3-77-S1.doc}.\\
 Literature used for the network construction. Each citation corresponds
 to one edge in the regulatory network.
\item \texttt{7521-0509-3-77-S2.zip}.\\
     Cytoscape import file. 
     Import this file into Cytoscape \cite{Sha03}
     to analyse the gene regulatory network in more detail.
     It also includes external links for the genes and references
     cited to GenBank, Uniprot, and PubMed.
\item \texttt{7521-0509-3-77-S3.zip}.\\
     Cytoscape import file. 
     Open this file after importing the CYS file (provided by Additional file 2)
     into Cytoscape \cite{Sha03}
     if the layout of the CYS file cannot be displayed correctly
     with your Cytoscape version.
\item \texttt{7521-0509-3-77-S4.xls}.\\
     List of probe sets used, processed microarray data
     and visualisation of expression time courses for the genes analysed.
     Raw data are deposited under accession number GSE13837
     at GEO \cite{GEO}.
\item \texttt{7521-0509-3-77-S5.xls}.\\
     Processed and visualisation of GEO Data.
     Data were extracted from GSE1742 (TGF$\beta$1) and GSE2624 (TNF$\alpha$)
     at GEO \cite{GEO}.
\item \texttt{7521-0509-3-77-S6.xls}.\\
     Discretised gene expression time courses.
\item \texttt{7521-0509-3-77-S7.xls}.\\
     Histograms of gene expression simulation.
     The simulations for TGF$\beta$1 (blue) and TNF$\alpha$ (red) were run
     for 12 time steps (x-axis)
     and for each initial state derived from the patients' data separately.
     A simulated expression of 100\% (y-axis) means that
     in all six cases the gene was on.
\item \texttt{7521-0509-3-77-S8.xls}.\\
     List of the top 500 knowledge base rules
     valid for the simulations as well as for the data
     from stimulations with TGF$\beta$1 and TNF$\alpha$.
\item \texttt{ECMData.xls}\\
Measured values for probesets, after normalization of mean and variance for all TGFB1 and TNF chips, respectively.
Selection of probesets, logarithmic values of the geometric mean for one gene.
\end{itemize}
Moreover, the attached CD contains the discretised measured and simulated data as formal contexts (\texttt{*.txt} files readable also with \texttt{Concept Explorer} \cite{ConExp}), as well as the complete stem bases in \texttt{*.txt} and \texttt{PROLOG}\index{Prolog} format (see contents at \texttt{readme.txt}). 

\section{Results and discussion}
\subsection[Creating a regulatory network from literature information]{Creating a regulatory network by information extraction from literature}
The available literature was screened for genes and proteins involved in ECM maintenance and expressed in the lining layer SFB of the SM. In order to derive a regulatory network, we comprehensively collected literature knowledge related to the formation and degradation of ECM in human fibroblasts and analysed it manually. We chose collagen type I, which is formed by the COL1A1 and COL1A2 gene products, as a connective tissue representative, several MMPs as ECM-degrading enzymes, their inhibitors, and TFs regulating them. Finally, we selected 18 genes (Figure \ref{fig:geneList}) and the literature was screened again for gene regulatory relations and interactions connecting them (see additional file~1 for a complete list). Some contradictory literature findings were resolved manually (see Section \ref{sec:BFAdapted}).
 
The resulting regulatory network\index{Gene regulatory network} is almost closed and represents the most important ECM network functions. Here, we imply that the receptors for the external signalling molecules are always available and functional in SFB. Note, that TGF$\beta$1 (TGFB1) and TNF$\alpha$ (TNF) are the only entities playing a dual role as both external signal molecules and target genes because of their introduction into the simulation as starting effectors. 

It turned out that the knowledge about gene regulatory events is limited and that, to the best of our knowledge, the regulation of SMAD and SMAD expression has not been fully characterised so far. The SMAD gene products seem to be available in sufficient amounts and we were unable to find reports about their regulated expression. In addition, not all influences of TGF$\beta$1 and TNF$\alpha$ on gene expression could be described as direct effects of transcription factors at the mRNA level because the important SMAD family members act as regulators on the protein-protein interaction level. All influences were included in the network at this point to avoid premature loss of information. 

\begin{sidewaysfigure}
 \centering
 \includegraphics[width=225mm]{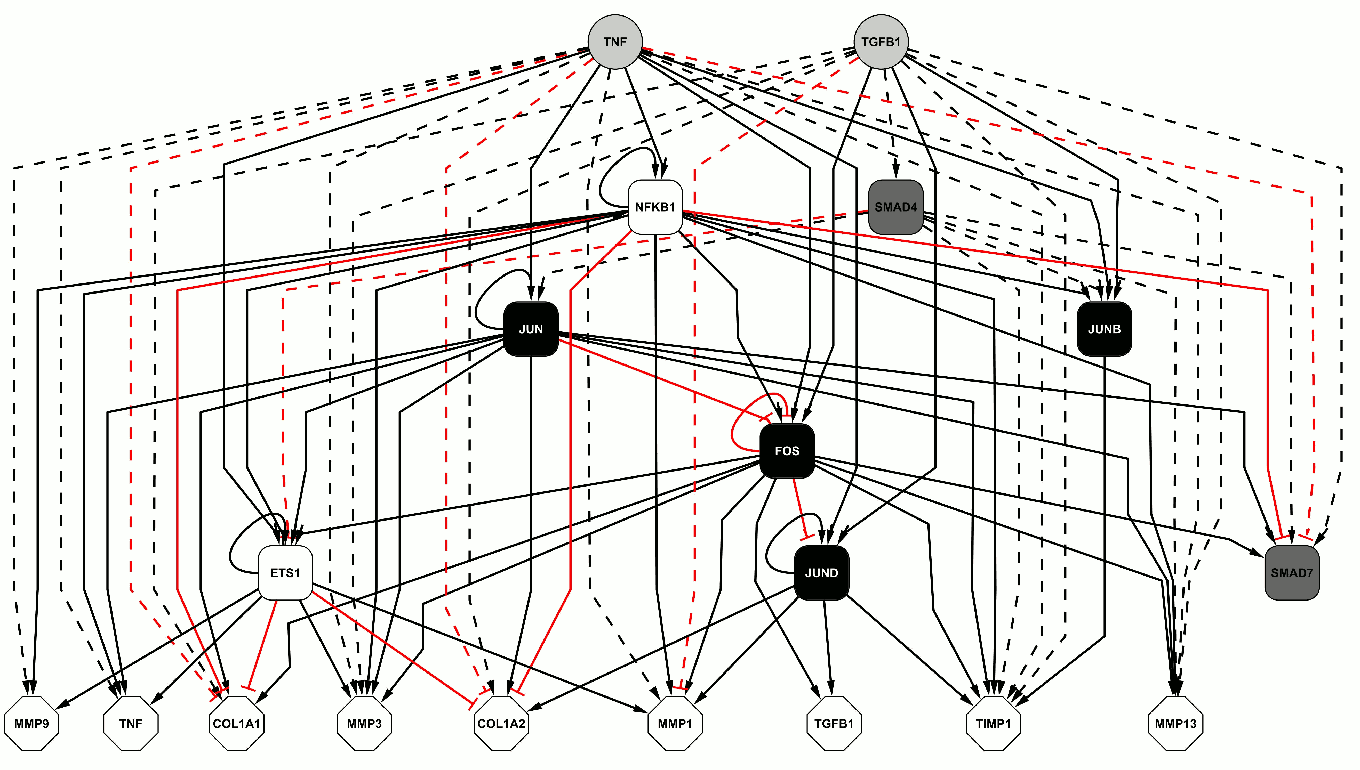}	
 \caption{Overview of the ECM network in hierarchical order.\index{Gene regulatory network}
\label{fig:ECMNetwork}
Regulatory effects via TFs are shown as continuous lines, others as indirect effects as dashed lines. Inhibition is marked by a red T-arrow, induction is illustrated by black arrows. The external signals TGF$\beta$1 and TNF$\alpha$ are shown as light grey circles, the internal SMAD signalling molecules as dark grey squares, TFs are depicted as black (AP‑1 components) or white squares, and the target genes are shown as white octagons. This picture was generated using Cytoscape 2.6.0.}
\end{sidewaysfigure}
Although many TFs such as AP‑1 are also regulated at the protein level (e.g., by phosphorylation), those effects can be reflected simplistically by regulatory processes at the transcriptional level. However, activating SMADs as SMAD3 and SMAD4 are also regulated by inhibitory members of the SMAD family (SMAD6 and SMAD7), which may counteract transcriptional activation and add an extra level of complexity \cite{Ros08}. Therefore, SMAD7 was introduced into the network as a TGF$\beta$1-dependent repressor of SMAD-dependent transcription.

In the case of SMAD3, we decided to subsume its influence under the SMAD4 effects because both are described to have nearly identical effects and act in concert. Moreover, we could not find well-defined information about SMAD3 regulation. Hence, we added an inducing influence of SMAD4 on MMP13 (at present only known for SMAD3) for keeping all the SMAD effects in the network. 

The subunits of the homo- or heterodimer TF AP‑1, i.e., Jun, JunB, JunD, and Fos (JUN, JUNB, JUND, FOS), determine its different regulatory activities (for AP‑1 components see \cite{Her08} and references therein). Therefore, we decided to disassemble the transcriptional active entity AP‑1 into its subunits. In contrast, for the dimeric TF NF‑$\kappa$B, which can be composed of the gene products of NFKB1, NFKB2, RELA, RELB, and/or REL \cite{Per08}, we selected NFKB1 as the representative gene with respect to our signalling framework. All the genes and their interrelations were transferred into the program Cytoscape \cite{Sha03} to visualise our network containing 19 nodes and 79 edges, respectively, as shown in Figure \ref{fig:ECMNetwork}. Detailed network examination is available through the network description files (additional files 2 and 3), also providing external links to GenBank, Uniprot, and PubMed for all edges and nodes. 

Available tools for automatic text mining decide schematically, e.g., by pre-built rules like co-occurrence of gene names and interaction verbs or pattern matching, whereas a human expert is able to integrate unanticipated types of information and to decide whether the paper confirms the investigated situation. However, we used the tools Bibliosphere \cite{Bib} and Pathway Studio \cite{Pat} in order to verify completeness and consistency of the assembled network (data not shown).

\subsection{Boolean functions}
Due to its capability for displaying dynamic dependencies between individual parameters, a Boolean network is more specific than the graphical network in Figure \ref{fig:ECMNetwork}, which summarises isolated literature facts. In order to decide about the connectives OR/AND, which represent causally determined relations between different genes, cellular signalling processes were also considered. 

In the case of a known transcriptional activation of any gene by the stimuli TNF$\alpha$ or TGF$\beta$1 via a specific TF, this activation was represented in the network using the term GENE.out = STIMULUS AND TF. Without such evidence, these influences were connected by GENE.out = STIMULUS OR TF. Since it is well known that the so-called SMAD pathway is activated by TGF$\beta$1 but not influenced by TNF$\alpha$ \cite{Ver02}, we used the AND connection for SMAD3/4 and TGF$\beta$1, even if there was no explicit literature evidence for an impact of TGF$\beta$1 onto the respective gene. 

Another example for setting up the functions is the integration of: (i) the known auto-regulatory transcriptional activation of JUN by TNF$\alpha$ via JUN, and (ii) the activation of JUN via SMAD4 (TNF$\alpha$-independent) into the single Boolean function 5 (compare Table \ref{fig:BFLit} with Tables \ref{fig:BFTgf} and \ref{fig:BFTnf}): JUN.out = (TNF AND JUN) OR SMAD4. Based on the illustrated principles, the Boolean functions characterising formation and remodelling of the ECM were generated (Table \ref{fig:BFLit}).
\begin{figure}[t]
 \centering
 \includegraphics[width=151mm]{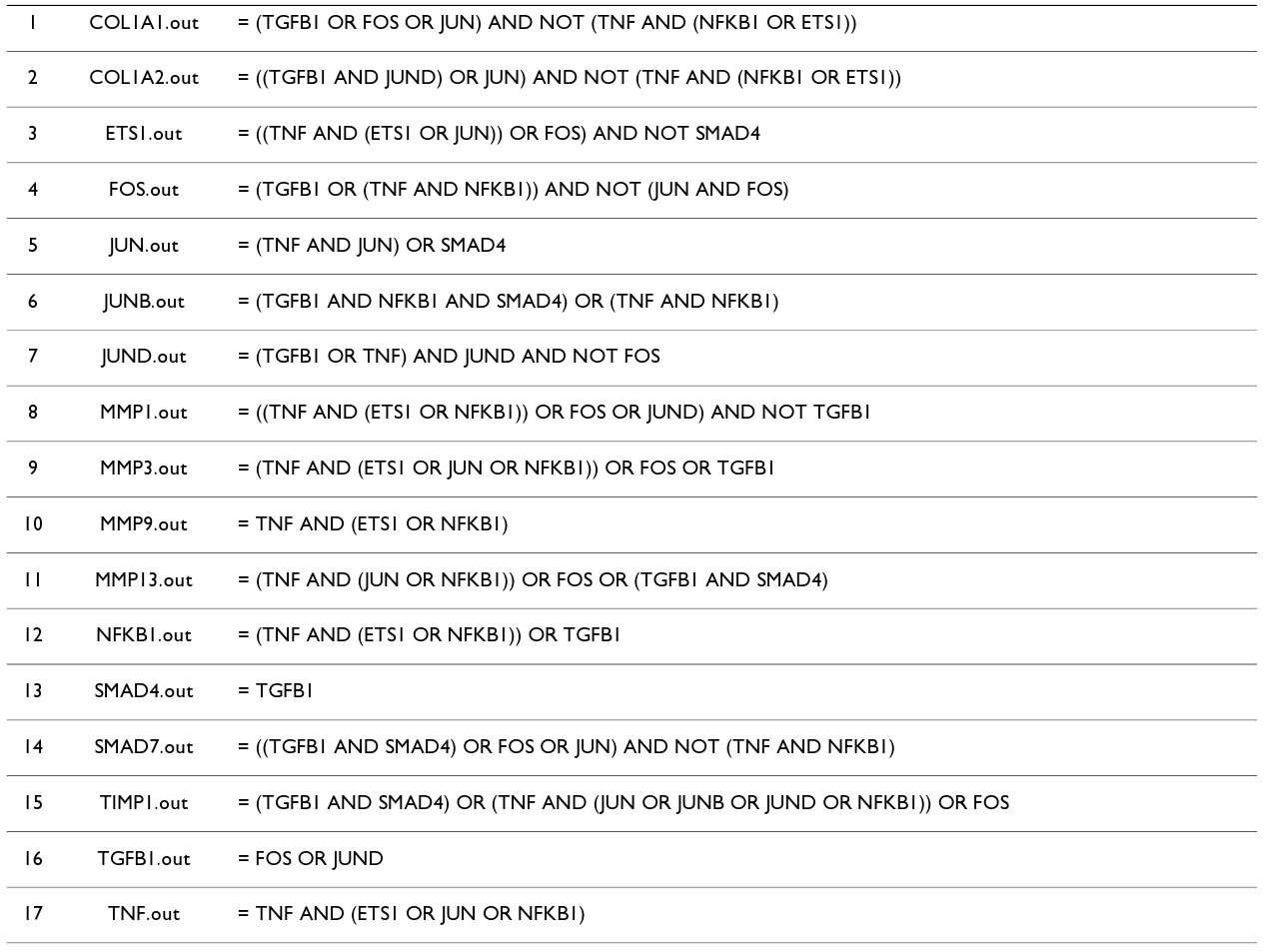}
  \caption{Boolean rules based on literature information.}
\label{fig:BFLit}
\end{figure}

\subsection{Gene expression time courses following TGF$\beta$1 and TNF$\alpha$ stimulation}
\begin{sidewaysfigure}
 \centering
 \includegraphics[width=238mm]{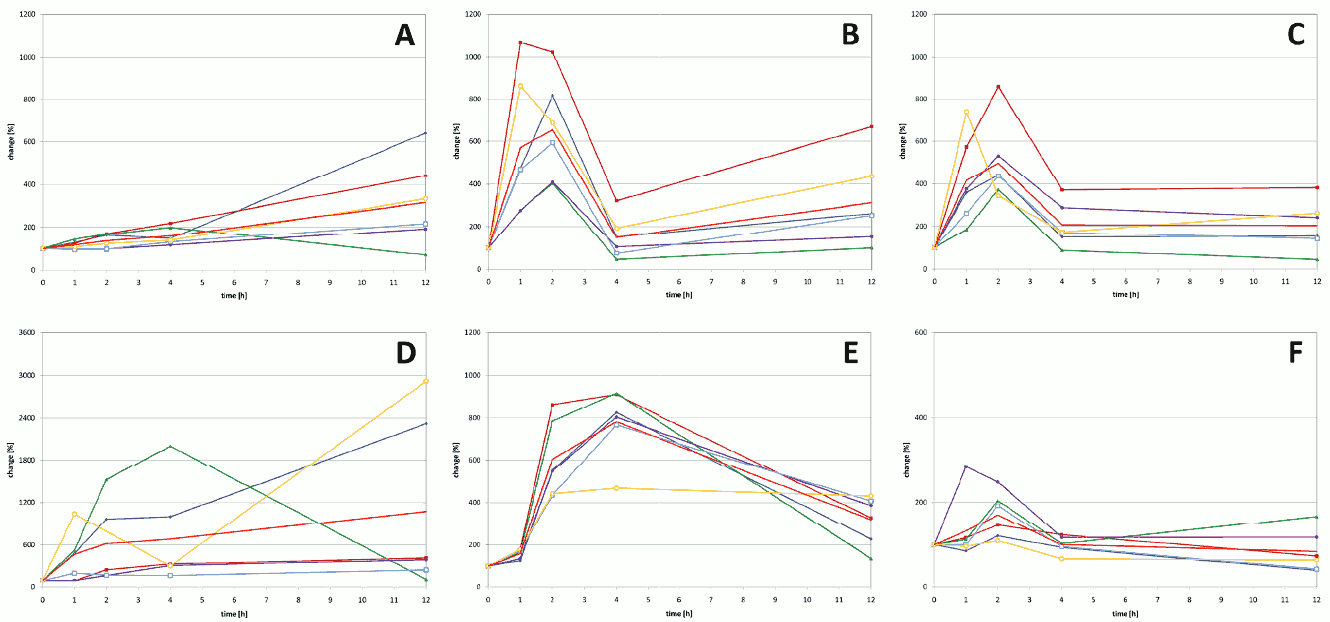}
\caption{Gene expression time courses following TGF$\beta$1 or TNF$\alpha$ treatment.
COL1A1 (A), JUNB (B), and SMAD7 (C) gene expression in response to TGF$\beta$1 treatment (upper row); TNF$\alpha$ response (lower row) of NFKB1 (D), MMP1 (E), and SMAD7 (F). The average time course is shown as light red curve without symbols, the data for individual samples are depicted in other colours (OA1: blue, filled symbol; OA2: red, filled symbol; OA3: green, filled symbol; RA1: purple, filled symbol; RA2: blue, open symbol; RA3: yellow, open symbol). The time courses and the values calculated from the microarray experiments for all analysed genes are included in additional file 4.}
\label{fig:timeCourses}
\end{sidewaysfigure} 
We analysed gene expression changes of SFB from patients with RA (3 patients) or OA (3 patients) following TGF$\beta$1 and TNF$\alpha$ stimulation (Figure \ref{fig:clinicalChar}). Due to the strong stimuli, both groups of cells reacted in an almost identical way, and we did not differentiate among them. In another study, for example, OA cells were considered to be a disease control group \cite{Wol11}.
 
Following pre-processing of the microarray data gained from U133 Plus 2.0 arrays, we extracted the data for probe sets related to our genes of interest (see Methods). The data are available in the GEO database (GSE13837 at \cite{GEO}). For the following analyses we excluded values which exceeded the reliability threshold of $p \leq  0.05$ for any patient at any time point (0, 1, 2, 4, 6, and 12 hours). In Figure \ref{fig:timeCourses}, some selected examples for the influence of TGF$\beta$1 and TNF$\alpha$ on gene expression are presented. The time courses of COL1A1 and JUNB expression are shown to illustrate the TGF$\beta$1 response in SFB, and the TNF$\alpha$ response is illustrated by NFKB1 and MMP1 expression. SMAD7 expression data are also included for both treatments. The data and the respective diagrams for all genes and both treatments can be found in additional file 4.
 
For comparative purposes, we also analysed public data from the GEO database, first, TGF$\beta$1 treated murine embryonic fibroblasts (GSE1742) and second, TNF$\alpha$ stimulation of endothelial cells (HeLa, GSE2624). Following prolonged TGF$\beta$1 treatment in murine cells, COL1A2, JUN and TIMP1 gene expression increased, whereas FOS decreased. In contrast, FOS, JUN and JUNB expression in HeLa cells rapidly increased following TNF$\alpha$ stimulation. Unfortunately, no data about the protease genes were available in this dataset (additional file 5). Even though cell type, experimental design and duration of treatment differ from our experiments, they reflect the two general trends: a positive effect on ECM formation by TGF$\beta$1 and a degradative influence on ECM by TNF$\alpha$ (mediated at least in part by FOS and JUN), which is consistent with our data. However, the evaluation of the complete data sets revealed discrepancies between the expected expression profile of individual genes and their time courses following stimulation in the experiment. 

\subsection{Data discretisation}\index{Discretisation}
We developed a data discretisation method which appropriately captures biologically relevant differences in gene expression levels. The individual time profiles for each gene were separately discretised to the values 0 or 1 by $k$-means clustering \cite{Har79}, a method which is often applied for gene expression time series. No improvements were observed when applying Ward’s hierarchical clustering \cite{War63} or single linkage clustering as proposed in \cite{Dim05} (data not shown). We introduced several supplementary criteria (see Methods), e.g., the values of a time series were all discretised to the constant value 0 or 1, if the differences of all $\operatorname{log}_2$ values (fold-changes) were less than 1 \cite{Mar07}. For the discretised data see additional file 6. 

\subsection{Boolean functions adapted to the data}\label{sec:BFAdapted}
\begin{figure}[t]
 \centering
 \includegraphics[width=151mm]{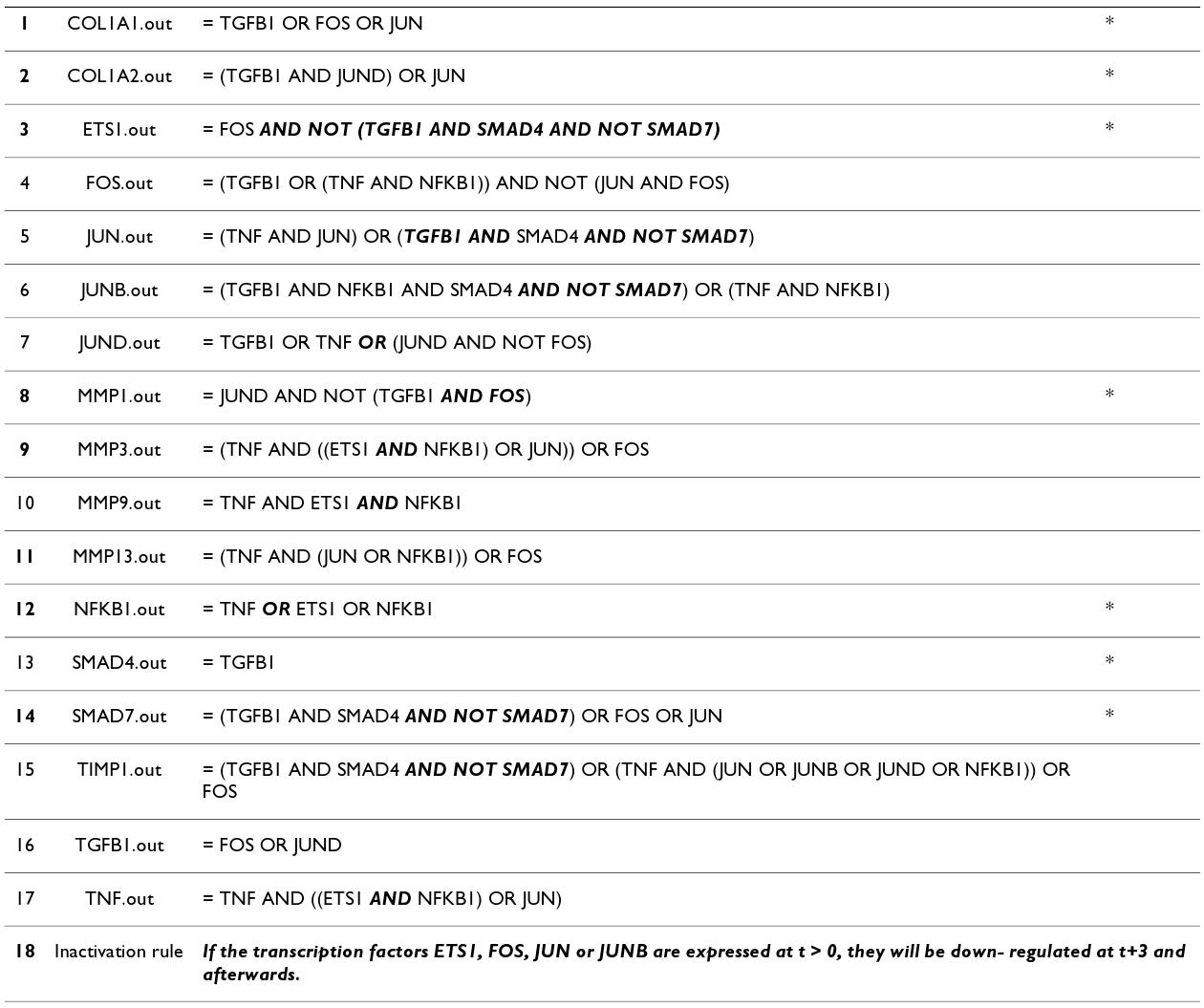}
  \caption{Revised Boolean functions for the simulation of TGF$\beta$1 stimulation. 
Rows marked by an asterisk indicate differences of the functions for TGF$\beta$1 and TNF$\alpha$ stimulation, changes to Table \ref{fig:BFLit} are italicised and bold, and function numbers in bold indicate omitted (1, 2, 3, 8, 9, 11, 12, 14) or inserted (18) terms.}
\label{fig:BFTgf}
\end{figure}
Simulations were generated using an asynchronous update scheme, assuming time intervals -- approximately equal to 1 h time steps -- as follows: transcription 1 (NFKB1: 2), translation: 1, RNA lifespan: 1, and protein lifespan: 2. The Boolean functions generated the transcriptional states according to the functional influence of proteins (stimuli or TF); translation and mRNA/protein degradation were computed from this output state according to the predefined intervals (see Section \ref{sec:principlesSim}).
 
As starting conditions of the simulations we chose the discretised initial states derived from our experimental data. An additional initial state was introduced in which solely the transcription factors were set to on, which enables the model system to respond to the external stimulators TNF$\alpha$ or TGF$\beta$1 immediately. The simulations were performed over twelve time steps; however, we did not aim at an exact correspondence to the experimental observation time of twelve hours, but tried to adjust the simulated time courses to qualitative features such as early, intermediate or late up-regulation. Improving the Boolean functions accordingly, the initially applied literature-based information was completed by: (i) valid and specific experimental information; (ii) knowledge and experience of biological experts; and (iii) in some cases, a more focused and precise literature query. For a comparison of the discretised observed time series and the final simulations see the additional files 6 and 7. We developed several biologically interesting and plausible data-independent hypotheses, for example, we modelled the regulation of SMAD3/SMAD4 effects by a protein-protein interaction with SMAD7.

\begin{figure}[t]
 \centering
 \includegraphics[width=151mm]{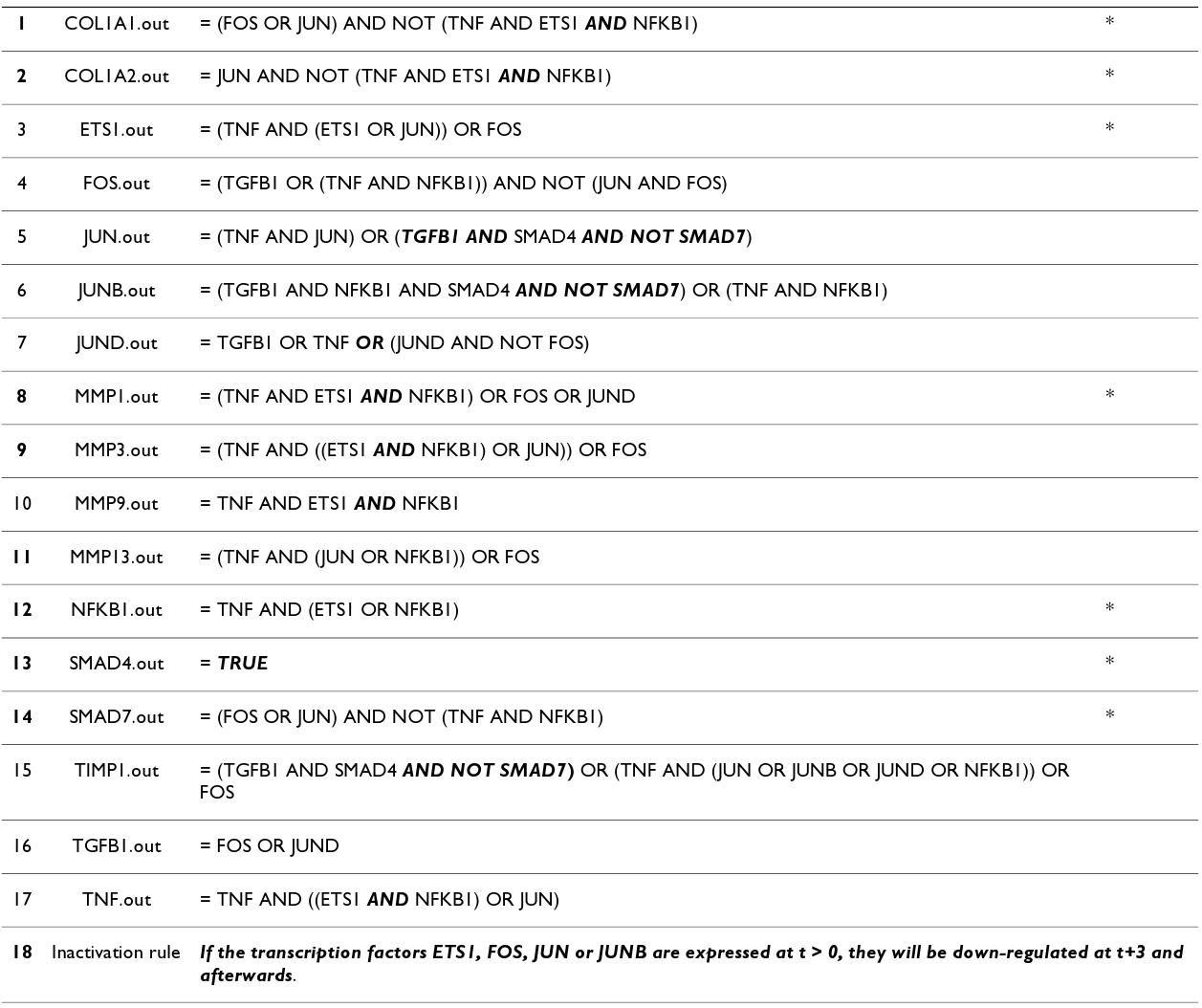}
  \caption{Revised Boolean functions for the simulation of TNF$\alpha$ stimulation.
Rows marked by an asterisk indicate differences of the functions for TGF$\beta$1 and TNF$\alpha$ stimulation, changes to Table \ref{fig:BFLit} are italicised and bold, and function numbers in bold indicate omitted (1, 2, 8, 9, 11, 12, 13, 14) or inserted (18) terms.}
\label{fig:BFTnf}
\end{figure} 
The resulting optimised Boolean network with the revised Boolean functions (Tables \ref{fig:BFTgf} and \ref{fig:BFTnf}) represents an enhanced ECM model, roughly matching the given biological conditions and extensively exceeding the present possibilities of automatic methods such as text mining, symbolic computation or machine learning. Considering the additional information available, we accepted these biologically reasonable changes:
\begin{enumerate} 
\item In the case of TNF$\alpha$ (or TGF$\beta$1) stimulation, the production and secretion of TGF$\beta$1 (or TNF$\alpha$) by SFB should not contradict the influence of the abundant stimulating protein TNF$\alpha$ (or TGF$\beta$1). In these cases (e.g., for the COL1A1.out function in Table \ref{fig:BFTnf}, and for the MMP1.out function in Table \ref{fig:BFTgf}) we removed TGFB1 (TNF) AND (...) from the Boolean function term. This adjustment did not always change the simulation result, since, for example, TNF was always off following TGF$\beta$1 stimulation (numbers of the Boolean functions (BF)\glossary{name={BF}, description={Boolean function}} affected: 1, 2, 3, 8, 9, 11, 12 and 14). 

\item	Down-regulation of gene expression is an essential biological principle. For that reason we had to introduce a time-limited inactivation mechanism which could not be derived from the literature because information regarding down-regulatory mechanisms is very restricted. Moreover, complex and variable mechanisms were hard to model, e.g., JUN down-regulation which is driven by: (i) inactivation of the TF protein itself; (ii) a general shift in the composition of the TF AP‑1, resulting in a reduced amount of TF enhancing JUN transcription; and (iii) binding/inactivation of JUN by other proteins. Therefore, a time-limited mRNA inactivation was introduced for JUN, JUNB, FOS and ETS1. Accordingly, an inactivating rule was created: if these TFs are expressed at $t > 0$, they will be set to off at $t+3$ and afterwards (no. of BF affected: 18). In addition, at that step we included an inhibition of TGFB1/SMAD4 signalling-based target gene expression by integrating a SMAD4-inhibiting signal (i.e., SMAD7, included as AND NOT SMAD7) guaranteeing the subsequent inactivation of TGF$\beta$1-related gene expression (BF affected: 3, 5, 6, 14 and 15). JUND is constitutively expressed at an intermediate level, which is consistent with GEO (GSE1742 and GSE2624) and our own data, as well as with the literature \cite{Hir89}. For NFKB1 transcription, an inhibitory effect was not implemented, since the activity of NF‑$\kappa$B at the protein level is controlled by interaction with several IKB proteins \cite{Hac06} which were not included in our ECM network model. 

\item	SMAD4 induction is not dependent on TGF$\beta$1 stimulation, because it is constitutively expressed (i.e., always TRUE, BF affected: 13). However, without TGF$\beta$1-mediated phosphorylation, SMAD4 is not activated at the protein level and shows no transcriptional activity, even though constitutively expressed. Therefore, we amended the term SMAD4 to TGF AND SMAD4 in order to represent the necessity of TGF$\beta$1 for SMAD4 activation (BF affected: 3, 5). 

\item We considered the relation ETS1 AND NFKB1 for the target genes \cite{Lam97} instead of assuming alternative pathways by ETS1 OR NFKB1 because regulation by NF‑$\kappa$B seems to be dependent on ETS1, and the MMPs, for example, require both factors \cite{Miz06} (BF for TGF stimulation affected: 9, 10 and 17, Table \ref{fig:BFTgf}; BF for TNF stimulation affected: 2, 8, 9, 10 and 17, Table \ref{fig:BFTnf}).
 
\item Since the inhibition of JUND expression by FOS is only observed in the case of a concomitant JUND-based positive feedback, the inhibitory effect of FOS has been restricted to this case \cite{Ber98} (BF affected: 7). 

\item Since a TF should not necessarily be required for its own expression (positive feedback), in the case of JUND (and also NFKB1) the AND connective was changed to OR. The revision of this function prevents the absence of JUND expression following TGF$\beta$1 stimulation (BF affected: 7, 12).

\item Concerning the regulation of MMP1 expression by FOS, there were contradictory findings in the literature \cite{Bar00}, \cite{Whi00}. We decided for an inhibitory influence of FOS following TGF$\beta$1 stimulation, because otherwise MMP1 would have been permanently down-regulated by TGF$\beta$1 during the simulation (BF affected: 8, Table \ref{fig:BFTgf}).
 
\item The Boolean function MMP3.out = (...) OR TGFB1 was in obvious contradiction to the data of the present study, thus, the term OR TGFB1 was deleted. The same was done for the MMP13 function (BF affected: 9, 11).
 
\item In the case of NFKB1, the absence of TNF$\alpha$ following TGF$\beta$1 stimulation had no decisive influence (NFKB1 was not always off). For that reason, we changed NFKB1.out = TNF AND (ETS1 OR NFKB1) to NFKB1.out = TNF OR ETS1 OR NFKB1 (BF affected: 12, Table \ref{fig:BFTgf}).
 
\item \label{noTNF}However, concerning the expression of TNF itself, the necessity of a positive feedback could explain its complete absence following TGF$\beta$1 stimulation. On the other hand, TNF was expressed at some time points following TNF$\alpha$ stimulation, whereas it is commonly assumed that fibroblasts do not express TNF (BF not changed: 17). 
\end{enumerate}

In summary, we adjusted the set of BF obtained by adaptation to the gene expression data measured under two experimental conditions (TNF$\alpha$ and TGF$\beta$1 stimulation), in order to create an appropriate set of BF representing the existing knowledge about naturally occurring interrelationships as accurately as possible.

\subsection{Computing temporal rules by attribute exploration}

For each stimulus, the observed and the final simulated time series were translated and merged into a single transitive context $\Ktt$ (see Definition \ref{def:transCxt} and the Methods Section \ref{sec:KBMethods}). States are defined by the value \textit{on} or \textit{off} for each gene, and transitions were computed by linking an occurring input state to an arbitrary future (output) state of the simulation or observation. The set of all these transitions, represented by $\Ktt$, was analysed by the automatic, non-interactive version of the attribute exploration algorithm. 

The implications of the resulting stem base are temporal rules expressing hypotheses about attributes of states (e.g., co-regulation or mutual exclusion of gene expression) or system dynamics, which are supported by pre-existing knowledge and by the analysed data. Since an implication holds for the transitions between all temporally related states, a rule such as GENE1.in.out $\rightarrow$ GENE2.out.on means: if gene1 is expressed, gene2 always will be up-regulated in the future, at all subsequent observation time points and simulation steps. Due to this semantics, the implications neither depend on the correspondence of a simulation time step to a specific observation interval, nor on prior knowledge about time periods of direct or indirect transcriptional effects. Within the large knowledge bases for TGF$\beta$1 (2713 rules) and TNF$\alpha$ (8785 rules) stimulation, the most frequent and simple temporal rules were considered and analysed for dependencies between stimuli, induced TFs and their target genes. 

\subsection{Results of the attribute exploration}
\subsubsection{Stimulation with TNF$\alpha$}

Regarding stimulation with TNF$\alpha$, a coordinated down-regulation of the two TF SMAD7 (inhibitor of TGF$\beta$1/SMAD4 signalling) and ETS1 emerges, as indicated by the rules 33, 114, 135, 144, 157, and 186 (see \texttt{combTransTnf.txt} or additional file 8). For example, rule 186: 
\[< 90 > \text{COL1A1.out.off, ETS1.out.off} \rightarrow \text{SMAD7.out.off}\]
has the meaning: in all simulated and observed states characterised by the absence of COL1A1 and ETS1, SMAD7 is also off. $<90>$ stands for the support of the rule, i.e., the number of transitions (90 out of 294) that actually have the attributes of the premise. Rules 114, 135, 144, 157 and 186 indicate: if the TNF$\alpha$-dependent genes are not induced (ETS1 as mediator), then simultaneously the expression of TGF$\beta$1-dependent genes is enabled (SMAD7 is off). This suggests that TNF$\alpha$ and TGF$\beta$1 may act as antagonists in SFB, as described in \cite{Rog98}, \cite{Ver00}. 

The expression of NFKB1, which is also induced by TNF$\alpha$, proceeds conversely to that of ETS1 and SMAD7 (rules 34, 45, 70, 71, 134, 144, 154, 157 and 173) reflecting the different targets of NF-$\kappa$B and SMAD7. The antagonistic expression pattern of NFKB1 and SMAD7 appears indirectly in rule 33, where the two genes show up in the premise of a rule with high support: 
\[<150> (...)\: \text{NFKB1.out.on, SMAD7.out.off} \rightarrow \text{ETS1.out.off}\]
Regarding this rule, it is interesting that ETS1 always acts in the same direction as NF-$\kappa$B, according to the network derived from the literature (Figure \ref{fig:ECMNetwork}). In the adapted network (Table \ref{fig:BFTnf}), we assumed a necessary cooperation (i.e., an AND connective) for the positive regulation of ETS1, MMP1, MMP3, MMP9, and TNF, as well as for the inhibition of COL1A1 and COL1A2. Thus, rule 33 further suggests that the coordinated action of NF-$\kappa$B and ETS1 is turned off in states which are characterised by supplementary conditions as SMAD7.out.off. 

The generated rules adequately reflect the major influence of the TF AP‑1 in the TNF$\alpha$ system: the expression of prominent targets, such as COL1A1, MMP1, and MMP3, depends on JUN (rules 211 and 258) and/or FOS (rule 204), with JUN as the key player. These rules connect input and output states and thus their semantics is directly related to dynamics, as seen in rule 211: 
\[<87> \text{TGFB1.in.on, TIMP1.in.on, ETS1.in.on, JUN.in.on} \rightarrow \text{MMP1.out.on}\] 
making this strong statement: if ETS1 and JUN are on, MMP1 will always be up-regulated in the future (at least within the time frame of 12 hours).

Sometimes, MMP1 is expressed simultaneously or before ETS1 and JUN. In the simulation, MMP1 was always on in the output state and from time point 2 h in the data. An exception can be found for the experimental results from OA sample OA3 (Figure \ref{fig:clinicalChar}), where MMP1 is off after 12 h. This is the reason for the computation of the auxiliary conditions TGFB1.in.on and TIMP1.in.on in rule 211.
 
Concerning the regulation of target genes, the expression of MMP1, MMP3, and MMP13 is co-regulated (rules 35, 63, 82, 86 and 176), while MMP9 is expressed independently (rules 24 and 35). There is a contradiction between the simulation and the data: in the observed experimental time series, MMP13 is always off, whereas the Boolean network predicts an up-regulation similar to MMP1 and MMP3. This unexpected absence of predicted MMP13 expression may be an indication for a more complex regulation of MMP13 transcription, exceeding the already known regulatory interrelations. Therefore, the MMP13 promoter and further enhancer/repressor sequences should be targeted for a more pronounced structural and functional analysis. For MMP9, the simulation and the experimental data are in good agreement: the gene is off in most, but not all states.\label{noMMP9} However, since the expression of MMP9 by (S)FB is discussed controversially in the literature (see \cite{Stu03} and \cite{Xue07} vs. \cite{Gau03}), the calculated expression of MMP9 by fibroblasts -- at least at a limited number of time points -- supports the majority of studies, reporting detectable MMP9 mRNA amounts in (S)FB. 

Several rules unanimously indicate the co-expression of the ECM-forming genes COL1A1 and COL1A2 (rules 87, 88 and 95), but contradictory rules occur concerning their expression profile in comparison to the MMPs. COL1A1 and COL1A2 seem to be co-expressed with MMP1 (rules 90 and 176), for COL1A2, however, a certain co-expression with MMP9 is calculated as well (rules 76 and 77), which conflicts with the opposing expression of MMP1 and MMP9 (see above). Therefore, the expression of collagens does often, but not necessarily always correlate with the expression of MMPs. This reflects the imbalance between MMP-dependent destruction and collagen-driven regeneration/fibrosis of ECM in the joints in inflammatory RA. 

The calculated knowledge base also contains a further unexpected correlation. According to rule 166:
\begin{align*}
<94>\quad\: &\text{FOS.in.off, TIMP1.in.on, SMAD7.out.off}\\
\quad \rightarrow\: &\text{TGFB1.in.on, MMP1.out.on, TGFB1.out.on}
\end{align*}
and rule 188, the expression of MMP1 may also be induced in the absence of FOS (e.g., by JUN-containing AP-1 complexes), indicating that the regulation of MMP1 does not predominantly depend on FOS as proposed in the literature \cite{Whi95}, \cite{Sun02}. This result may point to the influence of other TFs, e.g., NF-$\kappa$B, ETS1, or AP-1 complexes containing JUN, which may indeed be able to induce target gene expression in the absence of FOS. 

\subsubsection{Stimulation with TGF$\beta$1}
For the stimulation with TGF$\beta$1, we had a total number of 341 transitions. The SMADs play a major role for the expression of TGF$\beta$1-dependent target genes, as reflected by various classes of rules containing SMAD4 and/or SMAD7 (see \texttt{combTransTgf.txt} or additional file 8). For example, SMAD4 can be involved in the expression of COL1A1, see rule 15 (and also rules 21, 26 and 30):
\[<239> \text{ETS1.out.off} \rightarrow \text{SMAD4.in.on, COL1A1.out.on, SMAD4.out.on}\]
This also suggests an antagonistic behaviour of ETS1 and SMAD4: if ETS1 was off, then SMAD4 was on, as well as in all previous states. Rules 52 and 57 suggest a dependency of MMP1 on SMAD4. However, this seems to be one amongst many other influences (or could be a non-influencing coincidence), since SMAD4 was permanently on during simulation and experimental stimulation with TGF$\beta$1 (exception: sample RA3 at time point 2h). 

The expression of MMP9 is neither induced by SMAD4 (rules 7, 24 and 41) nor by any other TF, indicating that MMP9 is not influenced by TGF$\beta$1. The fact that TGF$\beta$1 obviously does not induce MMP9 (but other MMPs) agrees with findings reported previously \cite{Gau03} and represents a clear contrast to the MMP expression profiles following TNF$\alpha$ stimulation.

A further case of an antagonistic expression pattern was calculated for MMPs and COL1A1 (rules 21, 30, 36, 41, 54, and 60), for example, in rule 54: 
\begin{align*}
<170>\quad\: &\text{SMAD4.in.on, MMP3.out.off, MMP9.out.off, MMP13.out.off, (...)}\\
\quad \rightarrow\: & \text{COL1A1.out.on}
\end{align*}
Antagonistic expression profiles also can be observed for SMAD4 and other TFs, e.g., JUN and JUNB (rules 12, 39) or ETS1 (rule 15, see above). The variety of TF combinations found, even following the same stimulus, exceeds the possibilities of conventional TF studies because stimulation experiments are generally restricted to a selected set of read-out parameters (e.g., the expression of single TFs or target genes) which are not able to reflect the multiplicity of different effects in the cell.

Following stimulation with TGF$\beta$1, interestingly COL1A2 appears to be constitutively expressed since its status is always calculated as on (rule 1). Therefore, for the formation of collagen I, which contains COL1A1 and COL1A2 chains, COL1A1 expression seems to be the critical switch. 

\subsubsection{TGF$\beta$1 versus TNF$\alpha$ effects}
The calculated results impressively illustrate that TGF$\beta$1 and TNF$\alpha$ stimulation are mediated via separate signal transduction pathways, leading to the expression and activation of different TFs. In general, ETS1 and NFKB1 are induced predominantly by TNF$\alpha$, whereas SMAD expression depends on TGF$\beta$1 (represented by differential expression profiles of ETS1 and SMAD4). JUN and FOS, however, strikingly respond to both stimuli. This defined pattern results in the expression of target genes with opposing roles. TGF$\beta$1 positively regulates the enhanced formation of ECM components, whereas TNF$\alpha$ is strongly involved in the expression of ECM-degrading enzymes. This was the main reason for a discriminative revision of the BF for TNF$\alpha$ and TGF$\beta$1 (Tables \ref{fig:BFTgf} and \ref{fig:BFTnf}). Six BF were found to be differently adjusted (BF 1, 2, 3, 8, 12 and 14) which concern either the key players for ECM destruction (MMP1; BF 8), ECM formation (COL1A1 and COL1A2; BF 1 and 2) or important regulatory genes (ETS1, NFKB1, SMAD7). This may indicate that the differential effects on ECM induced by TNF$\alpha$ or TGF$\beta$1 are mainly mediated via ETS1 (BF 3), NFKB1 (BF 12, especially in the TNF$\alpha$ pathway), or SMAD7 (BF 14, especially in the TGF$\beta$1 pathway) identifying ETS1- and NFKB1-associated pathways as the major TNF$\alpha$-induced pro-inflammatory/pro-destructive signalling modules in rheumatic diseases, whereas TGF$\beta$1-driven and SMAD7-related signalling appears prominently involved in fibrosis.

\subsection{Querying the knowledge base}
The minimal rule set gave many new insights, and further queries can be addressed by accessing the TNF$\alpha$ and TGF$\beta$1 knowledge bases in one of two ways: (i) the Excel file containing the transition rules for structured searches within the rule sets (see additional file 8 containing the top 500 transition rules, \texttt{impCombTransTgf.txt} and \texttt{impCombTransTnf.txt} for complete lists); and (ii) the stem base in PROLOG format for queries concerning logically implied rules as in Section \ref{subseq:allTrans} for the \textit{B. subtilis} simulations (\texttt{impCombTransTgf.P} and \texttt{impCombTransTnf.P}).

\subsection{Expert centered attribute exploration}\label{sec:expertExploration}

Resuming the presented study published in \cite{Wol09}, I investigated in detail four genes showing strong changes over the time steps of the simulation modelling TGF$\beta$1 stimulation, by interactive attribute exploration: 
\begin{enumerate}
\item MMP13 is up- and downregulated at one or two time points during the simulation, but it was never expressed significantly in the microarry experiments. A human expert can partly solve this contradiction and decide for each relative implication to which of the two knowledge sources more plausibility is attributed.
\item I was also interested in the interplay of three TF belonging to different pathways, but connected in the Boolean formula for JUNB (BF 6 of Table \ref{fig:BFTgf}):
\[\text{JUNB.out = (TGFB1 $\wedge$ NFKB1 $\wedge$ SMAD4 $\wedge\: \neg$SMAD7) $\vee$ (TNF $\wedge$ NFKB1)}\]
JUNB represents the AP1 complex playing an important role in the investigated biological context.
\item NFKB1 is known to activate MMP13.
\item Whereas JUNB and NFKB1 generally belong to TNF$\alpha$ pathways, SMAD7 inhibits TGF$\beta$1 signaling via SMAD4. The inhibition of SMAD7 expression by TNF$\alpha$ connects both pathways. 
\end{enumerate}

During attribute exploration of the simulated transitive context, the following implications were accepted, or an observed transition was introduced as counterexample, which differed little from a simulated transition. Counterexamples are indicated as binary numbers designating the values of the attributes JUNB.in, MMP13.in, NFKB1.in, SMAD7.in, JUNB.out, MMP13.out, NFKB1.out and SMAD7.out. Important decision criteria were support (number of transitions with the premise attributes / 72 overall transitions) and confidence (transitions with conclusion attributes / premise transitions) of the implication in the observed transitive context. For this purpose, the R script \texttt{selectGenes2.1.r} was used. In the protocol of the exploration, I omit background implications expressing dichotomic scaling\index{Scaling!dichotomic} like \agrave gene.in.on, gene.in.out $\rightarrow \bot$". The resulting context is \texttt{exploredTransTgf\_regulationMMP13.txt}, the complete stem base is documented at \texttt{exploredTransTgf\_regulationMMP13\_inf.txt}.

\begin{enumerate}
\item SMAD7.out.off $\rightarrow$ JUNB.out.off, MMP13.out.off, NFKB1.out.on.\\
The implication was rejected by reason of a week confidence 13/30. The counterexample 1001 0000 was introduced, i.e. a transition where JUNB.in and SMAD7.in are on and the other attributes off.

\item SMAD7.out.off $\rightarrow$ JUNB.out.off, MMP13.out.off.\\
Accepted with support 30/72 and confidence 23/30.

\item MMP13.out.on $\rightarrow$ SMAD7.out.on.\\
The implication was accepted, since there is no information in the observed context: MMP13 is always off, hence the support is 0.

\item JUNB.out.on $\rightarrow$ SMAD7.out.on.\\
Accepted with confidence 29/36. The implication is biologically interesting: If TNF$\alpha$ signaling via JUNB is enabled, TGF$\beta$1 signaling is inhibited via SMAD7 -- in spite of its reported inhibition, in turn, by TNF$\alpha$ (BF 14 of Table \ref{fig:BFLit}). However, TNF$\alpha$ was supposed to have a minor influence compared to the stimulus TGF$\beta$1 (BF 14 of Table \ref{fig:BFTgf}), which points at more complex regulations.

\item JUNB.out.on, SMAD7.out.on, NFKB1.out.off $\rightarrow$ NFKB1.in.on, JUNB.in.off, MMP13.in.off, SMAD7.in.off.\\
Despite a confidence 0/4, the complicated, little expressive implication has been accepted. The latter fact is reflected in the small support of 4/72.

\item JUNB.out.off, MMP13.out.off, NFKB1.out.off, SMAD7.out.off $\rightarrow$ JUNB.in.on, MMP13.in.off, NFKB1.in.off, SMAD7.in.on.\\
Accepted as before.

\item SMAD7.in.on, NFKB1.out.off $\rightarrow$ JUNB.out.off.\\
Accepted with confidence 6/7.

\item SMAD7.in.on, JUNB.out.on, SMAD7.out.on $\rightarrow$ NFKB1.out.on.\\
Accepted with confidence 12/13.

\item SMAD7.in.off $\rightarrow$ JUNB.in.off, MMP13.in.off.\\
The rule was accepted with confidence 32/32, i.e.~it holds as a strict implication also in the data. Like temporal rule 4, it underlines the interdependency of TNF$\alpha$ and TGF$\beta$1 signaling. 

\item NFKB1.in.on, NFKB1.out.off $\rightarrow$ SMAD7.out.on.\\
The rule should be rejected by reason of a small support (5/72) and confidence (0/5).
As counterexample 0010 1000 was chosen. However, as several other counterexamples, it violates previously accepted implications, thus the rule has been accepted in this first run of exploration.

\item NFKB1.in.on, JUNB.out.off, MMP13.out.off, SMAD7.out.off $\rightarrow$ NFKB1.out.on.\\
Accepted with 7/12.

\item NFKB1.in.off, NFKB1.out.off $\rightarrow$ JUNB.out.off.\\
The confidence is 5/9, but the implication is accepted as plausible: JUNB depends on NFKB1 (BF 6); if it remains off, then also JUNB.

\item NFKB1.in.off, JUNB.out.on, SMAD7.out.on $\rightarrow$ NFKB1.out.on.\\
The implication is rejected with confidence 2/6 and due to an almost inverted premise compared to 11, but the same conclusion. Again, counterexample 3 (1001 1001) and others violate implication 12.
\end{enumerate}

At this point, I restarted attribute exploration again with counterexamples 1-3, since the last two examples were judged to be important and should not be omitted. They only had become obvious at this later point of exploration, but the result should not depend on the order of the attributes. 

\begin{enumerate}
\item SMAD7.out.off $\rightarrow$ MMP13.out.off.\\
Subsumed by implication 2 of the first run and accepted with confidence 30/30 (remember that MMP13 was always off in the observations).

\item MMP13.out.on $\rightarrow$ SMAD7.out.on.\\
Accepted with support 0 (identical to implication 3 of the first exploration run).

\item NFKB1.out.on, MMP13.out.off, SMAD7.out.off $\rightarrow$ JUNB.out.off.\\
The confidence is not high (13/20) and the implication contradicts BF 6, if a stable expression of NFKB1 and SMAD7 is given over several time steps. Hence counterexample 4 (0000 1010) was introduced.

\item MMP13.out.off, NFKB1.out.off, SMAD7.out.off $\rightarrow$ MMP13.in.off.\\
Accepted by reason of the MMP13 measurements.

\item JUNB.out.on, NFKB1.out.off $\rightarrow$ MMP13.in.off.\\
Ditto.

\item JUNB.out.on, MMP13.out.off, SMAD7.out.off $\rightarrow$ JUNB.in.off, MMP13.in.off, SMAD7.in.off.\\
Accepted with confidence 3/7, since I did not want to contradict such a complicated, possibly artificial implication. The complexity should not be augmented by an unreliable counterexample.

\item SMAD7.in.on, MMP13.out.on, NFKB1.out.off, SMAD7.out.on $\rightarrow$ JUNB.out.off.\\
Complicated and accepted.

\item SMAD7.in.on, MMP13.out.off, SMAD7.out.off $\rightarrow$ JUNB.out.off.\\
Accepted with confidence 14/18.

\item SMAD7.in.on, JUNB.out.on $\rightarrow$ SMAD7.out.on.\\
The implication is a restriction of rule 4 in the first run and was accepted with confidence 13/17. It is unanticipated by the additional reason that JUNB and SMAD7 are regulated differently. In the adapted network (Table \ref{fig:BFTgf}), the only common influencing genes are SMAD4 and SMAD7. Possibly, the inhibition of both by SMAD7, suggested by SMAD7.in.on, could be the relevant influence on the transitions supporting the implication, since SMAD4 is always on in the simulations and the data, except for a single measurement.\label{impJUNB.on}

\item SMAD7.in.on, JUNB.out.on, MMP13.out.on, SMAD7.out.on $\rightarrow$ NFKB1.out.on.\\
Complicated and accepted.

\item SMAD7.in.off $\rightarrow$ JUNB.in.off, MMP13.in.off.\\
The implication is identical to rule 9 of the first run and relates to the inhibition of TGF$\beta$1 signaling via SMAD7 and the activation of TNF$\alpha$ signaling via JUNB.\label{impJUNB.off}

\item NFKB1.in.on, MMP13.out.off, NFKB1.out.on, SMAD7.out.off $\rightarrow$ JUNB.out.off.\\
Confidence 7/12, accepted as complicated.

\item NFKB1.in.on, JUNB.out.on, NFKB1.out.on $\rightarrow$ SMAD7.out.on.\\
Accepted with confidence 23/28, similar to implication \ref{impJUNB.on}.

\item NFKB1.in.on, JUNB.out.off, NFKB1.out.off $\rightarrow$ SMAD7.out.on.\\
The restriction of implication 10 in the first run was accepted despite confidence 0/5. There are few contradictory observations, and an antagonistic expression is plausible by the BF 6, 12 and 14: JUNB and NFKB1 are activated by NFKB1, SMAD7 is not. The implication constrains the coregulation of JUNB and SMAD7 found in the previous and other implications. However, it remains insecure and should be investigated by further experiments.\label{impNoCoReg}

\item 3 implications related to NFKB1 are accepted as complicated. 
\addtocounter{enumi}{2}

\item MMP13.in.on $\rightarrow$ SMAD7.in.on.\\
No support, accepted, similarly the next implication with MMP13.
\addtocounter{enumi}{1}

\item 7 further implications related to MMP13 are accepted, since they are less expressive and a data control is not possible.
\addtocounter{enumi}{6}

\item JUNB.in.on $\rightarrow$ SMAD7.in.on.\\
This again is a noteworthy, clear implication with highest confidence 38/38. It underlines the discovered coregulation of the two genes (compare \ref{impJUNB.on}, \ref{impJUNB.off} and 13, but also \ref{impNoCoReg}).

\item JUNB.in.off, JUNB.out.off, NFKB1.out.off $\rightarrow$ SMAD7.out.on.\\
Confidence 0/4, but the number of counterexamples in the data is not sufficient. The biological meaning is similar to implication \ref{impNoCoReg}, but this time it might be considered as a hint on a still unknown influence of JUNB on its own expression and that of NFKB1 and SMAD7 (where BF 14, however, reflects inhibition).

\item SMAD7.in.on, JUNB.in.off, NFKB1.out.off $\rightarrow$ SMAD7.out.on, JUNB.out.off.\\
No support in the data, accepted.

\item The 7 last implications concerning JUNB, together with other genes, are accepted as too complicated or because of a small support in the set of observated transitions.
\end{enumerate}

The two rules with highest support in the resulting stem base relate MMP13 expression to SMAD7 upregulation, in the input and the output state:
\begin{align*}
<104>\: &\:\text{MMP13.out.on} \rightarrow\: \text{SMAD7.out.on}\\
<89>\: &\:\text{MMP13.in.on} \rightarrow\: \text{SMAD7.in.on}
\end{align*}
In contrast, there is no implication with MMP13.on in the conclusion (already for the simulated transitions). This reflects the weak empirical evidence for MMP13 expression after Tgf$\beta$1 stimulation, which in turn is the reason why there were no counterexamples to the two rules in the observed data. The high support originates from the literature based prediction of expression after Tgf$\beta$1 stimulation. It was no decision criterion during the exploration aiming at qualitative, not quantitative relations. The rules can be interpreted as follows: In the rare cases of MMP13 upregulation, it is expected to be coregulated with SMAD7, but no positive conditions for its expression are found.

\chapter{Conclusion and outlook}\label{ch:outlook}
In this thesis, an FCA framework for the description and analysis of discrete processes was developed and investigated. The knowlege bases generated by attribute exploration suppport automatic reasoning. So it is worthwhile to sketch connections to DL, where fast reasoning software exists. The second section of the present chapter outlines possibilities of solving open mathematical and logical questions. Finally, the biological applications will be discussed, in particular the study of the ECM network related to rheumatic diseases.

\section{Transfer to Description Logics}\label{sec:compDL}\index{Description logics (DL)}
First I will discuss how our approach may be expressed in a standard DL and in $\mathcal{TDL}-Lite_{Bool}$, a temporal extension by \cite{Art07}. I will point at advantages of remaining within the general framework of FCA, instead of using the temporal expressivity of DL and attribute exploration adapted to the construction of DL knowledge bases. Furthermore, reasons will be given to concentrate on specific, practically relevant parts of temporal logic. I will give hints to the embedding of the syntax and semantics given by the defined formal contexts and their stem bases into $\mathcal{EL}$\glossary{name={$\mathcal{EL}$}, description={Weak DL with tractable subsumption algorithms},sort=$EL$} (a DL designed for reasoning about ontologies) and $\mathcal{TDL}-Lite_{Bool}$. It is out of the coverage of this study to search for a DL corresponding to the expressivity of the proposed formal contexts, not to mention its definition. However, we see a potential of the integration of our ideas into a DL and of adapting attribute exploration accordingly. 

A transition or transitive context (analogously and easier a state context) may be translated into a DL like $\mathcal{EL}$ as follows:
\begin{itemize}
\item The objects are elements of the model domain $\Delta:=S$ (states).\glossary{name={$\Delta$}, description={Domain (DL)}, sort=$Delta$}
\item Attributes $m \in M$ of the input states\index{State!input} are considered as concept names $C_m \in N_C$.
\item Attributes of the output states\index{State!output} are given by concept descriptions with the role\index{Role (DL)} $n \in \mathcal{N}_r$ (transition to the next state of a path) or $t \in \mathcal{N}_r$ (transition to an arbitrary subsequent state).
\end{itemize} 
Then a formal context similar to \cite[p.~159]{Baa09b} will be defined. It corresponds to a deterministic transition context $\Kt$ (compare Table \ref{tab:DLKt} with Table \ref{tab:wlanTrCxt}). Nondeterminism cannot be distinguished syntactically from a deterministic transition to a state with all alternative attributes (the first 3 lines in Table \ref{tab:DLKt} collapse), only semantically by the formal context. As in Table \ref{tab:wlanTrCxt}, the states are labelled by the path they belong to, e.g. $s_{00}, s_{01}, s_{02}$ are the initial states of the Linux driver installation and the Windows process with \texttt{ndiswrapper} or \texttt{driver.windows} as next steps, respectively. (At $s_{21}$ the two Windows paths coincide.)
\begin{table}\index{Transition context!DL}
\centering
\begin{tabular}{|l||c|c|c|c||c|c|c|c|}
\hline
&\begin{sideways}\texttt{$driverLinux$}\end{sideways}
&\begin{sideways}\texttt{$ndiswrapper$}\end{sideways}
&\begin{sideways}\texttt{$driverWindows$}\end{sideways}
&\begin{sideways}\texttt{$connection$}\end{sideways}
&\begin{sideways}\texttt{$\exists\, n.driverLinux$}\end{sideways}
&\begin{sideways}\texttt{$\exists\, n.ndiswrapper$}\end{sideways}
&\begin{sideways}\texttt{$\exists\, n.driverWindows$}\end{sideways}
&\begin{sideways}\texttt{$\exists\, n.connection$}\end{sideways}\\
\hline \hline
$s_{00}$ &&&&&$\times$&&&\\ \hline 
$s_{01}$ &&&&&&$\times$&&\\ \hline
$s_{02}$ &&&&&&&$\times$&\\ \hline
$s_{10}$ &$\times$&&&&$\times$&&&$\times$\\ \hline
$s_{11}$ &&$\times$&&&&$\times$&$\times$&\\ \hline
$s_{12}$ &&&$\times$&&&$\times$&$\times$&\\ \hline
$s_{20}$ &$\times$&&&$\times$&$\times$&&&$\times$\\ \hline
$s_{21}$ &&$\times$&$\times$&&&$\times$&$\times$&$\times$\\ \hline
$s_{31}$ &&$\times$&$\times$&$\times$&&$\times$&$\times$&$\times$\\ \hline
\end{tabular}
\caption{A transition context $\Kt$ in the language of description logics.}\label{tab:DLKt}
\end{table}

A transitive context $\Ktt$ might be translated by means of the role $t$. Then the information -- expressed by $n \in \mathcal{N}_r$ -- related to individual subsequent states is lost. Instead, $\exists t.C_m,\, m \in M$ indicates if a reachable state is in the extent of $\{m\}$ in $\Ks$ or, in the language of DL, if $C_m$ is interpreted by the respective reachable state. Thus, it makes sense to combine a DL-$\Kt$ and a DL-$\Ktt$ (Table \ref{tab:DLKtKtt}).
\begin{table}\index{Transitive context!DL}
\centering
\begin{tabular}{|l||c|c|c|c||c|c|c|c||c|c|c|c|}
\hline
&\begin{sideways}\texttt{$driverLinux$}\end{sideways}
&\begin{sideways}\texttt{$ndiswrapper$}\end{sideways}
&\begin{sideways}\texttt{$driverWindows$}\end{sideways}
&\begin{sideways}\texttt{$connection$}\end{sideways}
&\begin{sideways}\texttt{$\exists\, n.driverLinux$}\end{sideways}
&\begin{sideways}\texttt{$\exists\, n.ndiswrapper$}\end{sideways}
&\begin{sideways}\texttt{$\exists\, n.driverWindows$}\end{sideways}
&\begin{sideways}\texttt{$\exists\, n.connection$}\end{sideways}
&\begin{sideways}\texttt{$\exists\, t.driverLinux$}\end{sideways}
&\begin{sideways}\texttt{$\exists\, t.ndiswrapper$}\end{sideways}
&\begin{sideways}\texttt{$\exists\, t.driverWindows$}\end{sideways}
&\begin{sideways}\texttt{$\exists\, t.connection$}\end{sideways}\\
\hline \hline
$s_{00}$ &&&&&x&&&&x&&&x\\ \hline 
$s_{01}$ &&&&&&x&&&&x&x&x\\ \hline
$s_{02}$ &&&&&&&x&&&x&x&x\\ \hline
$s_{10}$ &x&&&&x&&&x&x&&&x\\ \hline
$s_{11}$ &&x&&&&x&x&&&x&x&x\\ \hline
$s_{12}$ &&&x&&&x&x&&&x&x&x\\ \hline
$s_{20}$ &x&&&x&x&&&x&x&&&x\\ \hline
$s_{21}$ &&x&x&&&x&x&x&&x&x&x\\ \hline
$s_{31}$ &&x&x&x&&x&x&x&&x&x&x\\ \hline
\end{tabular}
\caption{The apposition\index{Formal context!apposition} $\Kt \mid \Ktt^{out}$ of a transition context and the output part of a transitive context in the language of DL.}\label{tab:DLKtKtt}
\end{table}

A transitive context is more expressive than its DL equivalent. Because the objects of the context \ref{tab:DLKtKtt} are states, not transitions, only the information regarding the next and any reachable state is kept. 
The following implication of $\Ktt$ (p. \pageref{baseTransCxt}) expresses: If the wrapper module and the driver are and remain installed (at a second time point), then the connection data has to be (is) entered: 
\begin{center}
\texttt{ndiswrapper.in, driver.windows.in, ndiswrapper.out,}\\
\texttt{driver.windows.out $\rightarrow$ connection.out}
\end{center}
Yet in the DL context of Table \ref{tab:DLKtKtt}, already the less meaningful implication holds $\top \rightarrow \exists t.connection$. It turns out that the role $t$ corresponds better to the operator $\ev$ of a temporal context.

The semantics for the $\Ktmp$ attributes from CTL is given by $\Kt$, for $\Diamond Fm$, $\Box Gm$ and $\Box \neg F m$ also by $\Ktt$ (Remark \ref{rem:semanticsKtt}). Two of the $\Ktmp$ attributes are equivalent to the following concept descriptions:
\begin{align*}
\Diamond Fm &\cong \exists\,t.C_m\\\glossary{name={$\Diamond$}, description={Possibility operator},sort=$.Pos$}
\Box \neg F m &\cong \neg \exists\,t.C_m\glossary{name={$\Box$}, description={Necessity operator}, sort=$.Nec$}
\end{align*}
If concept descriptions relate to states indexed by the path they belong to, as in Table \ref{tab:DLKtKtt}, the following temporal operator may also be expressed:\glossary{name={$\neg F$}, description={Never}, sort=$.Never$}
\[\Diamond \neg F m \cong \neg \exists\,t.C_m.\]


S. Rudolph \cite{Rud06} and F. Baader / F. Distel \cite{Baa09b} deal with a potentially infinite attribute set,\label{infAttr} caused by nesting of roles\index{Role (DL)} like $\exists r.\exists r.\cdots$ . In the case of cyclic concept descriptions, e.g., the potential \textit{role depth}\index{Role (DL)!depth} (number of nested roles) is infinite. If such DL concepts are interpreted temporally, they express properties of cyclic time developments. S. Rudolph investigated attribute exploration of formal concepts with increasing role depth (in the more expressive DL $\mathcal{FLE}$) and showed the existence of a termination condition for finite models. However, there are computational problems, since the number of attributes grows exponentially with the depth. F. Distel \cite{Baa08} proved that a finite basis exists for the set of all $\mathcal{EL}$-implications (or GCIs) holding in a finite model. In \cite{Baa09b}, a single formal context with relational attributes has been defined; its attribute exploration provides an algorithm to explicitly compute the implicational base.

As mentioned before, the DL transition context of Table \ref{tab:DLKt} is derived from such a context. However, we restricted our framework to simple relational attributes (role depth one). Since temporal attributes can also refer to infinite paths, we do not need nested roles in this case. Roles of depth two or three may still be meaningful to an expert, but the purpose of the present work was not to develop a general framework for the conceptual exploration of relations. If the number of relational attributes remains restricted and almost all attributes occur in at least one of the output states, the set of attributes $\exists r.C_m, C_m \in \mathcal{N}_C$ may be fixed in advance. Then general attribute exploration can be applied, and no algorithm handling a growing set of attributes is needed as designed in \cite{Baa09b}.

Thus, standard FCA algorithms could be used. Nevertheless, our approach is generalisable to more complicated relational attributes; also \textit{Relational Concept Analysis (RCA)}\index{Relational Concept Analysis (RCA)} \cite{Hac07} may be applied. Inversely, attribute exploration software using fast DL reasoners could be used for our special formal contexts.

In the approach of \cite{Art07} (see Section \ref{sec:tempDL}), a flow of time is represented by $\N_0$, and the interpretation consists in a family of succeeding situations $(\mathcal{I}_n)_{n \in \N_0}$ defined by the interpretation function $\mathcal{I}$. According to \cite[p.~48 f.]{Baa08},
the family can be seen as a set of temporally changing models with a constant domain $\Delta$ for the concept description language $(\mathcal{L}, \mathcal{I})$ given by $\mathcal{TDL}-Lite_{Bool}$, and each model\index{Model} correponds to a formal context. This is a more expressive framework
than ours and may not  easily be translated into the defined formal contexts. On the other hand, $\mathcal{TDL}-Lite_{Bool}$ is based on the linear time logic LTL, whereas attributes of the temporal context are definable by the nondeterministic,\index{Process!nondeterministic} branching time logic CTL.


A state context $\Ks$ can be translated into $\mathcal{TDL}-Lite_{Bool}$ as follows:
\begin{itemize}\index{State context!$\mathcal{TDL}-Lite_{Bool}$}
\item The domain $\Delta$ does not correspond to the object set, but to the attribute set $M$ of $\Ks$.
\item The object names $a_i$ are in one-to-one correspondence to the attributes by $a_i^{\mathcal{I}(n)},\: i \in \{0,...,|M|-1\}, n\in \N_0$.
\item The object intent of a state $s \in S$ could be identified with a single concept name $A_0$. Then $A_0^{\mathcal{I}(n)} := s'$, but the whole syntactic information concerning attributes holding in a state would be lost. A better alternative is to define $|M|$ DL concepts $A_i$ by $A_i^{\mathcal{I}(n)} := a_i^{\mathcal{I}(n)}$, if $a_i^{\mathcal{I}(n)} \in s'$, else $A_i^{\mathcal{I}(n)} := \emptyset$. Then the assertion $A_i(a_j)$ holds, if $i = j$ and $a_j^{\mathcal{I}(n)}$ is contained in the intent of $s$.
\item Then every line of a state context corresponds to an interpretation $\mathcal{I}(n)_C$ without role names:
\[\mathcal{I}(n)_C := (\Delta, a_0^{\mathcal{I}(n)},..., A_0^{\mathcal{I}(n)},...),\]
where the object intent is given by the current interpretation of the DL concepts $A_i$. This is a bijection only if the state context is not clarified.
\end{itemize}

This translation defines a concept description language $(\mathcal{L}, \mathcal{I})$ of which $\Ks$ is a model. GCIs $C_1 \sqsubseteq C_2$ hold iff $C_1^{\mathcal{I}(n)} \rightarrow C_2^{\mathcal{I}(n)}$ holds in $\Ks$ for all $n \in \N_0$,
with $C_1:= A_i \sqcap ... \sqcap A_j,\: C_2:= A_i \sqcap ... \sqcap A_j \sqcap ... \sqcap A_k$, for $i<j<k \leq |M-1|$ (compare \cite[p.~48-50]{Baa08}).

It is only mentioned here that the interpretation $P_i^{\mathcal{I}(n)} \in \Delta \times \Delta$\index{Role (DL)!local} of local roles $P_i$ may be given by a transition context.


\cite{Art07} investigate also a weak temporalisation of $\mathcal{EL}$. In $\mathcal{TL}_\mathcal{EL}$\glossary{name={$\mathcal{TL}_\mathcal{EL}$}, description={Temporal extension of the DL $\mathcal{EL}$}, sort=$TLEL$} existential role restrictions are admitted, but no local roles\index{Role (DL)!local} nor negation, and temporal operators are restricted to $F$ (eventually) for concept constructors and $G$ (always) for formula.
\[C := \top \mid A_i \mid C_1 \sqcap C_2 \mid \exists T_i.C \mid F\,C\]
In this logic GCIs are satisfiable in the most general model (where all concepts and roles are interpreted by the whole domain at every time point), but it is undecidable whether a GCI is a consequence 
of a finite set of GCIs \cite[Theorem 11]{Art07}. In our \agrave pure" FCA framework however, we are safe that such problems do not occur: For each formal context the specific implications are decidable, even in linear time relative to the size of the stem base.

A detailed investigation of decidability and complexity issues of the presented approach in comparison to different DLs is out of reach of this thesis.


\section{Open mathematical and logical questions}
We have established the foundation to exploit manifold mathematical results of
FCA for the analysis of gene expression dynamics and of discrete temporal transitions in general. In Chapter \ref{ch:bg} mathematical questions related to background knowledge and to the mutual dependency of the defined formal contexts and their stem bases were solved. While it was not feasible to develop a complete rule calculus for a large class of temporal implications, the method was applied to special cases and possibilities of tackling the general problem were discussed. 

A related question should be further analysed: How can attribute exploration be split into partial problems for these special contexts? For instance, one could focus on a specific set of genes (Remark \ref{rem:subsets}) or temporal attributes first, which is understandable as a scaling (p. \pageref{subCxtScaling}).\index{Scaling} Then, the decomposition theory of concept lattices will be useful, which permits an elegant description by means of the corresponding formal contexts \cite[Chapter 4]{GW99}.

The price of the logical completeness of attribute exploration is its computational complexity.\index{Attribute exploration!complexity} Computation time strongly depends on the
logical structure of the context, and there exist cases where the size of the stem base is
exponential in the size of the input \cite{Kuz06}. It was proved recently that the basic step -- recognising whether a subset of attributes is a pseudo-intent -- is coNP-complete \cite{Bab10}. However, deriving an implication from the stem
base is possible in linear time, related to the size of the base, and the Prolog queries in Section
\ref{subseq:allTrans} were very fast. 

In addition to the integration of background knowledge, attribute
exploration can be simplified and shortened, if implications are decided without the necessity to generate all possible transitions. For that purpose, model
checking\index{Model!checking} could be a promising approach \cite{Esp94}, \cite{Cha04}, or the structural and functional analysis of
Boolean networks\index{Boolean network} by an adaptation of metabolic network methods in \cite{Kla06}. In this study, determining
activators or inhibitors corresponds to the kind of rules found by our method, and logical steady
state analysis indicates which species can be produced from the input set and which not. 
Another direction of research would be to conclude dynamical properties of Boolean networks
from their structure and the transition functions, for instance by regarding them as
polynomial dynamical systems over finite fields and by exploiting
theoretical work in the context of cellular automata \cite[Section 4 and 6]{Lau05}.

\section{Assessment of the biological applications}\label{sec:conclusionECM}
The analyses in the ECM study (Chapter \ref{ch:geneRegNets}) were based on literature data refering to healthy human SFB. These findings were fine-tuned and adapted to gene expression time course data triggered by TGF$\beta$1 and TNF$\alpha$ in SFB from RA and OA patients. Both the assembly of previous knowledge and the adaptation of the Boolean functions gave detailed insights into disease-related regulatory processes. To the best of our knowledge, this is the first dynamical model of ECM formation and degradation by human SFB. 

Manual adaptation of the network may be superseded by algorithms of network inference \cite{Hec09} in order to achieve a scalability to larger networks.\index{Gene regulatory network} Then, however, knowledge and data are integrated by criteria fixed in advance and hardly controllable in their effects, not by flexible and open expert decisions. Moreover, the relation of the validation by network inference and by computing the stem base should be clarified, or an expert centered attribute exploration should be applied.

The fidelity of the obtained temporal rules was reinforced by the comparison of simulated and observed time series data, first manually, then automatically by computing the stem base of the combined transitive context. One of the strengths of the FCA method applied here is its ability to give a complete, but minimal representation of a data set. This complete overview of temporal rules enabled us to find new relationships. The most unexpected result is the expression of TNF at some time points following TNF$\alpha$ stimulation, whereas it is commonly assumed that SFB do not express TNF \cite{Bur97}, \cite{Rit00}. Similarly, our experimental data as well as our simulation results support MMP9 expression in SFB thus corroborating the majority of the literature regarding expression of this protease \cite{Stu03}, \cite{Xue07}. Here, it is important to note that a contamination of the SFB population with macrophages (potentially contributing to MMP9 production) can be excluded due to the SFB isolation protocol, resulting in a pure SFB population \cite{Zim01}. We also found that MMP1 was induced in the absence of FOS after TNF$\alpha$ stimulation, whereas MMP13 was not expressed despite reports about its induction by NF-$\kappa$B, JUN or FOS. These facts indicate that the regulation of MMP expression may be more diverse than presently known and that it still represents a relevant research target to elucidate the role of SFB in the pathophysiology of rheumatic diseases.

Concerning the formation of collagen type I fibres by COL1A1 and COL1A2 proteins following the stimulation with TGF$\beta$1, a constitutive expression of COL1A2 was calculated. Based on these data, COL1A1 has to be regarded as the critical switch for the formation of collagen I. In contrast, the corresponding literature generally postulates a co-regulation of both genes, due to similarities in their promoters \cite{Kar95}, \cite{But04}. This difference suggests that the regulation of COL1A1 and COL1A2 may not have been fully elucidated so far possibly pointing at COL1A1 as a more promising target for the exploration of fibrosis.

Our analyses also show that TNF$\alpha$-induced signalling predominantly results in the activation of ETS1 and NFKB1, whereas TGF$\beta$1-related signal transduction is ultimately regulated via proteins of the SMAD family. Defined intervention addressing these signalling modules, alone or in combination with established therapies targeting TNF$\alpha$ (e.g. etanercept), may therefore improve the efficiency and outcome of current anti-rheumatic treatment \cite{And02}. Alternatively, the present results may be employed to define subpopulations of RA patients in characteristic phases of RA (active inflammatory early versus burnt-out/fibrotic late) and tailor anti-rheumatic treatment to the particular needs of the respective phase \cite{Con99}.

The semantics of a transitive context implies rules to be valid for the complete simulated and observed time interval. Thus, besides on co-regulation and contrary regulation, the focus is on (positive or negative) attributes of gene expression processes which will always hold after a given class of states. Since in the ECM study \texttt{always} was restricted to the observation time of 12 h, meaningful results were obtaned. This type of rules is even more appropriate in cases of a dramatic, permanent switch of the cellular behaviour, like for sporulation of \textit{B. subtilis} (Chapter \ref{ch:bSubtilis}). An exploration of a temporal context was sketched in Section \ref{sec:exEx} for the computer example. Also in biological applications it will give supplementary insight, in particular by the \texttt{eventually} attributes.

Exact logical rules (implications) are a precondition for the uniqueness and minimality of the stem base, in contrast to association rule mining. Data discretisation is one possibility of handling biological imprecion and flexibility and to filter out noise resulting from measurements \cite{Lau04}. The presented method could be developed further by integration of \agrave blurred" methods like fuzzy FCA \cite{Alc10} or rough sets \cite{GMe09}. Clustering methods could be applied for data preprocessing.

Data discretisation\index{Discretisation} inevitably causes loss of information. Carefully evaluating the method, we tried to keep as much important information as possible and set low thresholds for differential gene expression. A recently developed FCA-based method avoids predefined discretisation but computes an ordered set of \textit{interval pattern structures} depending on the observed values \cite{Kay11}. Thus, a data set may be described without loss of information or by means of any desired granularity. 

Because the mathematical framework of FCA is very general and open, many adaptations of the presented methods are possible. According to the mainstream of applied FCA research, they should aim at the reduction of complexity as indicated. Then FCA methods can demonstrate their strength best: a clear account and visualisation of a conceptual logic. Beside this, multifarious refinements are possible, according to current approaches of modelling dynamics within systems
biology. For instance, the introduction of more fine-grained expression levels is possible, e.g. in the sense of \textit{qualitative reasoning} \cite{KGC05}. 

However, despite the rough discretisation to two levels off and on, the complexity of even relatively small networks such as our ECM network and the completeness of the attribute exploration algorithm led to a large number of temporal rules. High support\index{Support} of a rule (often correlated to its simplicity) was used as an indicator for the most meaningful hypotheses about co-regulation, mutual exclusion, and/or temporal dependencies not only between single genes, but also between small sets of (functionally related) genes. The inspection of the rule sets by experts should be aided by automatic preselection. \textit{Iceberg concept lattices} lead to a set of \agrave important" rules with high support, by taking advantage of the duality between the stem base and the concept lattice of a formal contexts \cite{Stu02}. 

We applied also complete and in-depth investigation of a small set of interesting genes by interactive attribute exploration. Using this procedure for the knowledge base construction, single rules can be validated manually or by a supporting computer program, or even new experiments can be suggested. While for the automatic stem base computation we applied a strong validation criterion requiring rules to hold for all simulated and observed transitions, the expert can evaluate thresholds of support and confidence. This may reduce noise or eliminate measurement errors. Importantly, these decision criteria can be freely combined with relevant knowledge, and the method does not depend on mathematical problems of association rule mining, since the expert decides about strict implications. As a main result, several temporal rules were found and confirmed, which express a coregulation of JUNB and SMAD7, i.e. an activation of TNF$\alpha$ signalling together with an inhibition of TGF$\beta$1 signalling (partly also for the inverse case). This is biologically plausible in general, but was surprising regarding the Boolean functions. Thus, by simulation, data and expert validation we discovered a temporal invariant of the network. 

Expert decisions are not always obvious. During the first run of the exploration it turned out that too restrictive implications had been accepted and therefore different plausible counterexamples could not be introduced. In the small example, restarting the exploration with the counterexamples was no major problem, but it would be helpful to keep correct implications. This is possible if the implications accepted after the first error are checked again and used as background knowledge for the further exploration \cite[p. 14]{Baa09a}.

Corresponding to the stem base, but with the supplementary information of extents (states or transitions), Hasse diagrams\index{Hasse diagram} of the investigated formal contexts like in Figure \ref{fig:noStress} give detailed (albeit less compact) insight into gene regulatory processes. They should be generated for small subsets of genes. Nested line diagrams\index{Nested line diagram} as implemented in ToscanaJ \cite{Tos} offer the possibility to investigate achievable subsets of genes and to combine two of them at a time according to a chain of questions. 

Combining two well-developed algebraic, discrete and logical methods -- Boolean network construction and FCA -- it was possible to include human expert knowledge in all different phases (assembly of the network, adjustment to the data, choice of relevant temporal rules and interactive attribute exploration), with the exception of the challenging data discretisation step. \agrave Digital" and \agrave analog" thinking are combined by our approach. Bioinformatic algorithms are often complicated and -- even if carefully evaluated -- can produce hardly interpretable results. Against this tendency, we wanted to avoid uncritical trust in the objectivity of a method and let the main task of interpretation with human experts. Hence, the main biological results are not only single findings, but within the interdisciplinary collaboration the medical, biological and bioinformatic scientists as well as myself were inspired to complex reflections about regulatory processes. We got new insights regarding existing knowledge and data, which are verifiable by further experiments.


\newpage
\fancyhf{}
\cleardoublepage
\fancyhf{}

\fancyhead[RE]{\fancyplain{}{\leftmark}}
\fancyhead[LO]{\fancyplain{}{\rightmark}}
\fancyhead[LE,RO]{\fancyplain{}{\thepage}}
\fancypagestyle{plain}{
\fancyhf{}
\renewcommand{\headrulewidth}{0pt} 
\renewcommand{\footrulewidth}{0pt}}

\addcontentsline{toc}{chapter}{\numberline{}Bibliography}
\bibliographystyle{plain}
\bibliography{models}

\addcontentsline{toc}{chapter}{\numberline{}Index}
\printindex

\chapter*{Attachments}
\chaptermark{Attachments}
\addcontentsline{toc}{chapter}{\numberline{}Attachments}
\begin{enumerate}
 \item Ehrenw\"ortliche Erkl\"arung
 \item CD with supplementary files and the developed R scripts
\end{enumerate}
\newpage

\fancyhf{}
\fancyhead[LE,RO]{\fancyplain{}{\thepage}}
\cleardoublepage

\chapter*{Ehrenw\"ortliche Erkl\"arung} 

Hiermit erkl\"are ich ehrenw\"ortlich, dass mir die geltende Promotionsordnung der Fakult\"at f\"ur Mathematik und Informatik der Friedrich-Schiller-Universit\"at Jena bekannt ist. Ich
versichere, dass ich die vorliegende Arbeit selbst\"andig, ohne unzul\"assige Hilfe Dritter
und ohne Benutzung anderer als der angegebenen Hilfsmittel angefertigt habe. Die aus
fremden Quellen direkt oder indirekt \"ubernommenen Gedanken sowie pers\"onliche Mitteilungen sind als solche kenntlich
gemacht. Eigene Pr\"ufungsarbeiten wurden nicht verwendet. Bei der Auswahl und Auswertung des Materials sowie bei der Herstellung des Manuskripts haben mich die in der Danksagung
genannten Personen unterst\"utzt. Ich habe nicht die Hilfe von Vermittlungs- bzw.
Beratungsdiensten in Anspruch genommen. Niemand hat von mir unmittelbar oder mittelbar
geldwerte Leistungen f\"ur Arbeiten erhalten, die im Zusammenhang mit dem Inhalt der
vorgelegten Arbeit stehen. 

Weiterhin erkl\"are ich, dass ich diese Dissertation noch nicht als Pr\"ufungsarbeit f\"ur eine staatliche oder andere wissenschaftliche Pr\"ufung eingereicht habe. Auch habe ich noch keine Dissertation bei einer anderen Hochschule eingereicht.

\vspace*{3cm}
Erfurt, 9. Mai 2011 \hspace{4cm} Johannes Wollbold



\end{document}